\documentclass[a4paper,11pt]{article}
\usepackage{jheppub} 
\usepackage{appendix}
\usepackage{graphicx}
\usepackage{physics}
\usepackage{tensor}
\usepackage{todonotes}
\usepackage{amsthm}
\usepackage{amsmath}

\usepackage{lmodern}
\usepackage{eso-pic}
\usepackage{transparent}
\usepackage{graphicx}

\newtheorem{lemma}{Lemma}

\newtheorem{proposition}{Proposition}

\theoremstyle{definition}

\theoremstyle{plain}

\title{\boldmath As Cold as a Black Hole: Extended Photon Spheres}

\author{Marcos Riojas}
\affiliation{Center for the Fundamental Laws of Nature, Harvard University, Cambridge, MA, USA}
\emailAdd{marcos\_riojas@fas.harvard.edu}

\abstract{It is widely believed that self-gravitating radiation cannot reach thermal equilibrium with a black hole in asymptotically flat spacetime. We describe a marginally stable continuum exception to the standard instability argument. Our main observation is that the photon sphere controls the Israel junction conditions (IJCs), the Tolman-Oppenheimer-Volkoff (TOV) equation, and finite-radius black hole thermodynamics: the IJCs and TOV equation are equivalent at zero radial pressure, the inverse specific heat at the photon sphere is proportional to $-\Lambda$, and adding shells at fixed total mass \textit{lowers} the asymptotic Hawking temperature if and only if the local specific heat is positive. Using these results, we show how to compute thermodynamic entropies without the Euclidean path integral.

The exception described here results from companion work with M.J. Strassler, where we found that a ``hillingar black hole" (HBH) mimics an  ordinary Schwarzschild black hole of mass $M$, sharing its Hawking temperature, photon ring, and, in equilibrium, its coarse-grained entropy $S = 4 \pi M^2$. Here, we show these features are not tuned; they follow uniquely from joint mechanical and thermodynamic constraints. Conditions for thermodynamic mimicry and for the formation of extended photon spheres are found; for self-similar solutions they coincide. A one-parameter family of self-similar systems -- all of which, excepting the HBH, require massless walls at the edges of their extended photon spheres -- satisfies both conditions. This family includes ``stiffest stars" and ``frozen stars". }

\begin{document}
\maketitle
\flushbottom

\section{Introduction}

Black holes are thermodynamic systems \cite{Bekenstein:1972tm,Bekenstein:1973ur,Bekenstein:1974ax,Hawking:1974rv,Hawking:1975vcx} because asymptotic time stops violently at the horizon, but for macroscopic black holes this halting is gentle in Planck units. Hence astrophysical black holes are among the coldest objects in the universe, much colder than the cosmic microwave background, with Bekenstein-Hawking entropy $S = A/(4 \hbar G)$ dwarfing that of ordinary matter of the same mass. That raises the question we consider in this note: can a black hole be surrounded by matter in stable thermal and mechanical equilibrium, and if so, what configuration can match its temperature and entropy? 

A sharp obstruction has been known since the 1980s. First, black holes in asymptotically flat spacetime cannot reach stable thermal equilibrium with infinity because their asymptotic specific heat is negative \cite{Hawking:1976de}. In fact, asymptotically flat spacetime at finite temperature admits a Euclidean Schwarzschild saddle with a negative mode \cite{PhysRevD.25.330}. Addressing this, York \cite{York:1986it} showed black holes in a finite cavity at fixed temperature have two branches of solutions. If the photon sphere is outside the cavity, the specific heat is positive and the system is stable. These branches merge at the minimum temperature because the cavity wall coincides with the photon sphere; see Fig.~\ref{fig:swallowtailYork}. Simply put, the first constraint is thermal: a stable cavity must be supported within the photon sphere of the black hole.

Mechanical stability is a separate matter. The Israel junction conditions \cite{Israel:1966rt} determine the location of a shell outside a black hole \cite{Frauendiener_1990}. Their stability under radial perturbations was worked out by Brady, Louko, and Poisson (BLP) \cite{Brady:1991np}. They found that unless we allow for superluminal sound waves, the equation of state must satisfy $\sigma \ge 2p$, where $\sigma$ is the surface energy density and $p$ is the surface pressure. In brief, the second constraint is that massive shells must sit outside the photon sphere; see Fig.~\ref{fig:BLP}. In conclusion then \cite{Brady:1991np}, joint thermodynamic and mechanical stability appears impossible;  they credit Israel for asking whether a spherical box extending above and below the photon sphere might work.

The following thought experiment suggests a resolution. Consider an observer hovering just outside the photon sphere of a black hole; there the black hole appears as an infinite plane (Fig.~\ref{fig:thoughtexperiment}). Moving outward slightly, the horizon looms as an enormous sphere with positive Gaussian curvature. However, after an infinitesimally thin and mechanically stable layer of radiation with $\sigma \simeq 2p$ is added just outside the photon sphere, its backreaction displaces the photon sphere slightly outward, with the horizon again appearing as an infinite plane. This is in essence a \textit{hillingar}\footnote{ Hillingar is an Old Norse term describing the upheaving of images above the horizon by an arctic mirage. When light rays curve at the same rate as Earth's surface, the ocean's surface appears flat and the visible horizon is extended outward. The ``ocean" of the HBH is an exact  gravitational analogue.} mirage, seen in arctic conditions where light bends toward and travels along the frigid ocean's surface, extending the visible horizon outward. 

There is no obvious reason why this should work. Nonetheless, we will show it does; iterating this produces what we refer to here as a \textit{hillingar} black hole (HBH),
whose special properties were identified and studied with M.J. Strassler
\cite{Riojas:2026ytt,Riojas:2026mfu}. The metric appears earlier: it is that of a Florides star \cite{1974RSPSA.337..529F}, obtained with a horizon in \cite{DiFilippo:2024poc}; see also \cite{Maeda:2024tsg}. Its stability is marginal -- each layer sits at the boundary of the stability region of \cite{Brady:1991np}, see Fig.~\ref{fig:IJCTOV} -- and its interpretation as being in thermodynamic equilibrium is discussed further below. The radiation is \textit{anisotropic}, transverse to the horizon with zero radial pressure, and its \textit{extended} photon sphere is an ``ocean" precisely as cold as the black hole. See Fig.~\ref{fig:HBH}.

This is an exception to the obstruction for the following reason. The stability regions of York \cite{York:1986it} and BLP \cite{Brady:1991np} meet at one point, the photon sphere, where a discrete shell can sit only when massless; a shell of finite mass must choose a side, and either side is \textit{unstable} for a different reason. Each layer has an \textit{infinitesimal} mass that moves the photon sphere outward; the next layer sits at the photon sphere as well, evading the constraints in \cite{Brady:1991np}. 

The HBH is not only an exception to the constraints. It also has the same temperature, coarse-grained\footnote{By coarse-grained entropy we mean the thermodynamic entropy, not the fine-grained entanglement entropy of a particular state.} entropy, and optical signature as a Schwarzschild black hole, as originally shown in \cite{Riojas:2026ytt,Riojas:2026mfu}; here we obtain this result using alternative methods.  Its status as a black hole mimic \cite{Cardoso:2019rvt,Bambi:2025wjx}, including its Schwarzschild temperature, shadow, and coarse-grained entropy, is not tuned; these are consequences of the constraints. The HBH is the continuum realization of the stability-selected point -- an edge-case in \cite{Brady:1991np} -- in Fig.~\ref{fig:IJCTOV}.

\begin{figure}
    \centering
    \includegraphics[width=\linewidth]{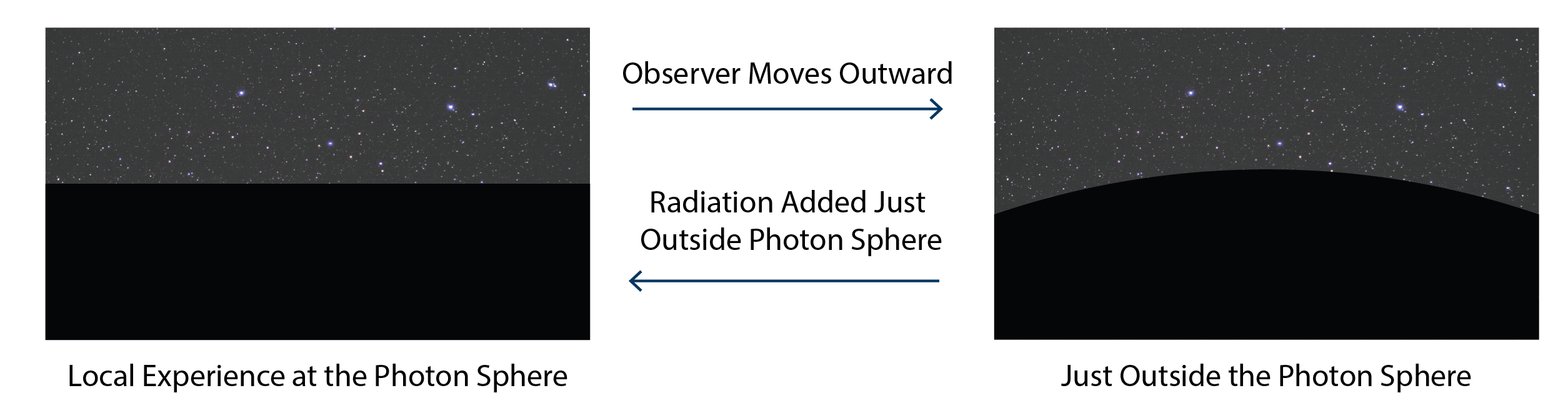}
    \caption{A procedure for constructing candidate equilibrium configurations. Massless particles orbit the black hole at the photon sphere \cite{Bardeen:1972fi}, so the shadow of the horizon appears as an infinite plane. The obstruction described in the introduction \cite{York:1986it,Brady:1991np} is evaded by placing successive infinitesimal layers of self-gravitating radiation with $\sigma \simeq 2p$ just outside the photon sphere. Moving outward causes the horizon to appear as an enormous sphere, as in the right panel; adding an infinitesimal layer moves the photon sphere outward, restoring the left panel. Iterating this procedure produces an extended region throughout which the shadow of the horizon appears as an infinite plane. } 
    \label{fig:thoughtexperiment}
\end{figure}

An outline of this note follows. The obstruction of \cite{York:1986it,Brady:1991np} is reviewed in Sec.~\ref{sec:obstruction}, as well as Sorkin, Wald, and Zhang's (SWZ) \cite{Sorkin:1981wd} analysis of the Tolman-Oppenheimer-Volkoff (TOV) \cite{Tolman:1939jz,Oppenheimer:1939ne} equation and the thermodynamics of self-gravitating radiation. The time-dependent Israel junction conditions (IJCs) for spherically symmetric spacetimes are derived in Sec.~\ref{sec:IJC}; the photon sphere is closely linked to the dynamics of Israel layers, with consequences for (e.g.) Randall-Sundrum (RS) branes \cite{Randall:1999ee,Randall:1999vf}. BLP's \cite{Brady:1991np} analysis is extended to general spherically symmetric spacetimes in Sec.~\ref{sec:Stability}, before specializing to anti-de Sitter spacetime. In Sec.~\ref{sec:SpecificHeatPS} we show the general result that the inverse specific heat at the photon sphere in $d$ dimensions is proportional to $-\Lambda_{\mathrm{CC}}$, with consequences for what follows.

In Sec.~\ref{sec:ExtendedPhotonSpheres} we extend the analysis of SWZ to anisotropic matter and show that in the limit of vanishing radial pressure, solutions to the TOV equation \cite{Oppenheimer:1939ne,Tolman:1939jz,Bowers:1974tgi} organize into discrete Israel layers.\footnote{Results for the entropy of self-gravitating radiation were obtained using different methods by Kim and Lee \cite{Kim:2019ygw}, see \cite{Riojas:2026mfu} for additional discussion. They deduced the scaling behavior for the entropy density of anisotropic matter with constant coefficients $w_i\equiv p_i/\rho$ using an entropy maximization method.} For traceless matter with $\sigma = 2p$ these layers are infinitesimal and continuous, forming an extended photon sphere at the point where the regions of mechanical and thermodynamical stability overlap in \cite{Brady:1991np}. In Sec.~\ref{sec:ascoldasablackhole} we show that the HBH is the coldest such configuration. For any finite number of mechanically stable shells surrounding a black hole, the asymptotic temperature strictly exceeds the Hawking temperature of a black hole of the same total mass; equality is reached in the continuum limit of infinitesimal traceless layers, which is the HBH. This explains the title of this note. We show in Sec.~\ref{sec:heatcapandtemp} that, at fixed total mass, adding a shell around a black hole decreases the asymptotic temperature of the system if and only if the specific heat at the shell's location is \textit{positive}.

Black hole thermodynamics at finite distance are obtained directly from the IJCs in Sec.~\ref{sec:EntropyandShells}. In this method, which does not need the Euclidean path integral, the temperature of the system is tracked while building it up layer-by-layer. In thermal equilibrium, we show (independently of the original argument in \cite{Riojas:2026mfu}) that $S = 4 \pi M^2$ for the HBH. 

In Sec.~\ref{sec:beyondnested} we generalize these arguments to show that extended photon spheres form when a certain condition (Eq.~\eqref{eq:extendedphotonspherecondition}) holds. This relation takes the same form as the constraint in Eq.~\eqref{eq:temperaturemimicry} that permits thermodynamic mimicry. Consequently, a one-parameter family of self-similar gases with this property has the temperature and coarse-grained entropy of a black hole. The frozen star \cite{Brustein:2018web,Brustein:2021lnr,Brustein:2023hic}, the stiffest star \cite{Zeldovich:1961sbr,Banks:2002fj}, and the HBH are members of this family; if the dominant energy condition (DEC) is satisfied, then thermodynamic mimicry occurs only for the HBH. The HBH is also an optical mimic \cite{Riojas:2026ytt}.

We also find that if thermal equilibrium is maintained, then in anti-de Sitter space large black holes can become colder and increase their entropy by forming stable shells. This observation was originally made for the HBH in \cite{Riojas:2026mfu}; here we give a slightly more general perspective. This work is concluded with a brief discussion.\footnote{While contemporary artificial intelligence tools were used for literature review, they are not responsible for the ideas and results in this paper.}

\subsection{Hillingar Black Holes and Related Work}\label{sec:HBHintro}

\begin{figure}
    \centering
    \includegraphics[width=\linewidth]{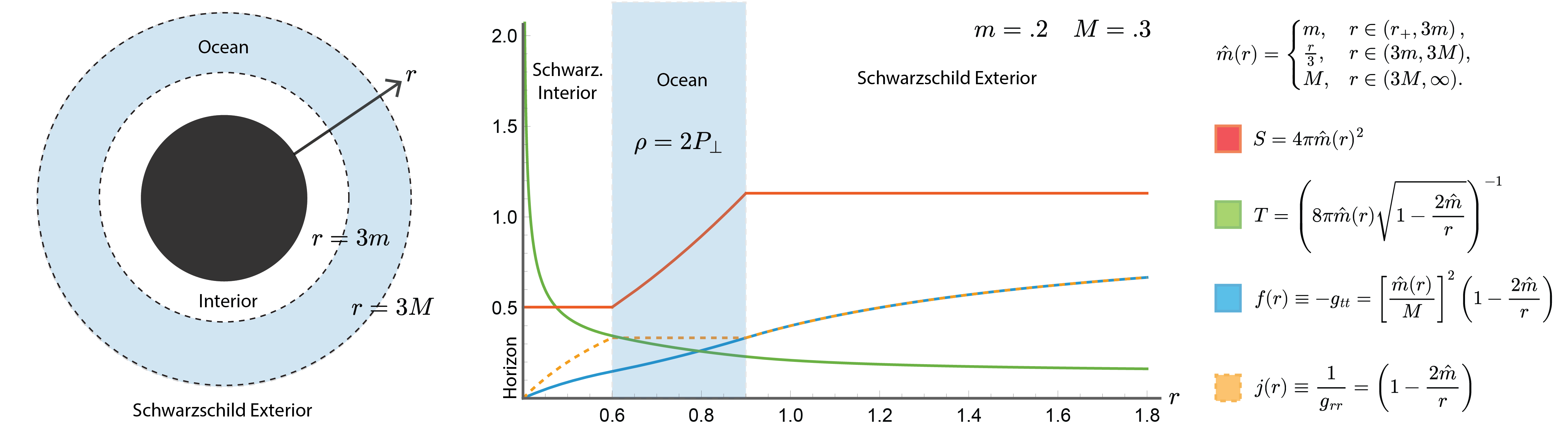}
    \caption{An extended photon sphere around a Schwarzschild black hole. When \textit{aligned} with a black hole's photon sphere, this gives a ``hillingar" black hole surrounded by an ``ocean" of orbiting massless particles. The continuous mass function $\widehat{m}(r)$ is piecewise defined, taking the value $m$ within the Schwarzschild interior, $M$ within the Schwarzschild exterior, and $\widehat{m}=r/3$ within the ocean. Here $\sigma=2p$ -- equivalently, $\rho=2 P_\perp$ -- in accordance with the thought experiment in Fig.~\ref{fig:thoughtexperiment}.
    }
    \label{fig:HBH}
\end{figure}

The HBH was investigated in \cite{Riojas:2026ytt,Riojas:2026mfu}, where its classical and thermodynamic properties were studied. In $3+1$ dimensional asymptotically flat spacetime, the HBH has the metric:
\begin{eqnarray}\label{fjmetric0}
    d s^2&=&-f(r) d t^2+j(r)^{-1} d r^2+r^2 d \theta^2+r^2 \sin ^2(\theta) d \phi^2 \\[4pt]
    j(r)&=&  1
    -\frac{2\widehat{m}(r)}{r}  \quad , \quad \quad 
    f(r) =\left[\frac{{\widehat{m}}(r)}{M}\right]^2  j(r) \ .
\end{eqnarray}
The mass function  ${\widehat{m}}(r)$, the mass interior to a sphere of areal radius $r$, takes the form:
\begin{equation}\label{massfunction0}
\widehat{m}(r)=
\begin{cases}
m, & r\in(r_+,\,3m),\\[2pt]
\dfrac{r}{3}, & r\in(3m,\,3M),\\[2pt]
M, & r\in(3M,\,\infty).
\end{cases}
\end{equation}
As shown in Fig.~\ref{fig:HBH}, $\widehat{m}(r)=m$ out to the black hole's photon sphere at $r=3m$; it grows linearly in the ocean $r\in(3m,3M)$ until reaching the critical photon orbit at $r=3M$. This metric was obtained with \cite{DiFilippo:2024poc,Maeda:2024tsg} and without \cite{1974RSPSA.337..529F,Joshi:2011zm} horizons, as noted further below.

Following \cite{Riojas:2026ytt,Riojas:2026mfu}, we shall refer to a continuous region where $rf'/2f=1$ as an \textit{extended photon sphere}. We demonstrate that these structures are permitted by Einstein's equations when the following condition for $T^\mu_\nu=(-\rho,P_r,P_\perp,P_\perp)$ holds; see App.~\ref{sec:UCO} for a derivation:
\begin{equation}\label{eq:intropscon}
    P_{\perp}(r)=\frac{1-2 \widehat{m}^{\prime}(r)}{8 \pi r^2}, \quad P_r=\frac{r-3 \widehat{m}(r)}{4 \pi r^3}, \quad \rho=\frac{\widehat{m}^{\prime}(r)}{4 \pi r^2}, \Longrightarrow \frac{P_r}{\rho}=\frac{r-3 \widehat{m}(r)}{r \widehat{m}^{\prime}(r)} \,.
\end{equation}
We also show that matter in thermal equilibrium with a black hole has the same temperature, and coarse-grained entropy, as a black hole of mass $M$ when Eq.~\eqref{eq:temperaturemimicry} holds:
\begin{equation}
    \frac{P_r}{\rho}=\frac{3 j(r)-1}{1-j(r)}=\frac{r-3 \widehat{m}(r)}{\widehat{m}(r)} \,.
\end{equation}
Remarkably, these equations are the same for self-similar matter, which requires $\widehat{m}(r) \propto r$. For example, the HBH ocean has energy density $\rho$ and transverse pressure $P_\perp$, where:
 \begin{eqnarray}\label{densitypressureintro}
     \rho = 2P_\perp = \frac{\widehat{m}'(r)}{4\pi r^2}=\frac{1}{12\pi G r^2} \  \implies P_r = 0, \quad r = 3 \widehat{m}(r) \, .
 \end{eqnarray}
For a single layer in the ocean, the energy density $\sigma$ and surface pressure $p$ satisfy $\sigma =2p$, which saturates the $\sigma\ge2p$ bound in \cite{Brady:1991np}. Similar HBH solutions exist in all space-time dimensions, with any cosmological constant $\Lambda$, for $d\geq 4$ \cite{Riojas:2026ytt}. 

The HBH ocean is \textit{aligned} in the following sense: at each radius within the ocean, the layers join smoothly at $r = 3 \widehat{m}$, and at its edges the ocean joins to the photon sphere of a Schwarzschild geometry without an IJC. Alignment occurs for the HBH because $\sigma = 2 p$ throughout the ocean; equivalently, $rf'/2f=1$ at each edge. 

This simple metric has been obtained as a limiting case in the literature, with \cite{DiFilippo:2024poc,Maeda:2024tsg} and without \cite{1974RSPSA.337..529F,Joshi:2011zm} horizons, because it is the unique luminal Einstein cluster \cite{Einstein:1939ms,1970GReGr...1...19K, 1971GReGr...2..321B,1974RSPSA.337..529F,Herrera:1997plx,Boehmer:2007az}. Einstein clusters near horizons were described surprisingly recently in \cite{Cardoso:2021wlq}; see also \cite{Jusufi:2022jxu}. 

In \cite{Andreasson:2021lsh} it is shown by Andréasson that a black hole can be surrounded by arbitrarily many nested, well-separated shells of massless particles, each generating a photon sphere; numerically, each shell can be placed close to the photon sphere just inside it. The metric first appears around a black hole in \cite{DiFilippo:2024poc}, where Di Filippo and Rezzolla considered an advanced civilization assembling a sequence of nested shells of light at successive light rings; the continuum limit, obtained from continuity and Birkhoff's theorem, is Eq.~\eqref{fjmetric0}. Computing the stress tensor of the discrete shells, which continuity and Birkhoff's theorem leave undetermined, corrects this picture. Each shell lies at the light ring of the original system, but within the full system's light ring; each has positive trace  and is unstable \cite{Brady:1991np}. The discrete shells of \cite{DiFilippo:2024poc} are therefore not shells of
light at the successive light rings of the system; this holds only in the continuum limit. The HBH is also a special case of Model I in \cite{Maeda:2024tsg} by Maeda, Cardoso, and Wang, where Einstein clusters model dark matter spikes that may form near galactic centers \cite{Gondolo:1999ef,Ferrer:2017xwm}. They noted marginal stability of Model I in the limit where the ISCO radius nears the photon sphere.

Related horizonless solutions at zero radial pressure have been long studied. Indeed, Eq.~\eqref{fjmetric0} is a Florides star
\cite{1974RSPSA.337..529F} whose core is replaced with a Schwarzschild interior. In a study of collapsing Einstein clusters by Joshi, Malafarina, and Narayan \cite{Joshi:2011zm}, the horizonless $m \to 0$ limit \cite{1974RSPSA.337..529F} of Eq.~\eqref{fjmetric0} is a special case of their naked singularity model; it separates solutions with stable and unstable circular orbits.\footnote{A related horizonless configuration was studied by Cunha \cite{Cunha:2025oeu}, who considered the backreaction of perturbations on the light rings of Florides-like solutions.} They considered whether these naked singularities are optically distinguishable from black holes, building on \cite{Claudel:2000yi,Virbhadra:2002ju}; restricting to the stable branch, they concluded that the singularity at the center of the cluster is optically distinguishable from a black hole of the same mass. The $m=0$ member of Eq.~\eqref{fjmetric0}, which lies at the boundary of their stable branch, is marginally stable with a naked singularity at the center and mimics the optical signature of a black hole; see \cite{Riojas:2026ytt}.

There have also been recent studies of extended regions where photons can execute circular orbits \cite{Hod:2025bsf,Li:2025wmd}. These can be written as a stiffest star \cite{Zeldovich:1961sbr,Banks:2002fj} with an integration constant (vacuum energy); they are self-similar and saturate the DEC when that constant vanishes. Those analyses assume isotropy, but spherical symmetry permits independent radial and transverse pressures. Here, we remove this restriction and obtain the necessary and sufficient stress-energy tensor conditions for an extended photon sphere, which is needed to understand the connection between optical and thermodynamic mimicry. To the best of our knowledge this characterization has not appeared before in the literature.

\section{Review: The Stability Obstruction in Asymptotically Flat Spacetime}\label{sec:obstruction}

A thermal bath surrounding a Schwarzschild black hole, in asymptotically flat spacetime, faces mechanical and thermodynamic constraints. These appear to be mutually exclusive. This section reviews three classic results that together sharpen this obstruction. 

Each is generalized in this paper. In Sec.~\ref{sec:York} we review York's saddle-point evaluation of the gravitational path integral in a finite spherical cavity. His analysis shows thermodynamically stable shells are inside the photon sphere. In Sec.~\ref{sec:BLP} we review the stability analysis of Brady, Louko, and Poisson (BLP). They found mechanically stable shells sit outside the photon sphere. Then, in Sec.~\ref{sec:SWZ}, we review Sorkin, Wald, and Zhang (SWZ), who considered entropy maximization for isotropic radiation. Self-gravitating isotropic radiation in equilibrium cannot reach the photon sphere: $\mu_{\max }\equiv M / R \approx .246<1 / 3$. 

In Sec.~\ref{sec:ExtendedPhotonSpheres}, we will see that the backreaction of \textit{anisotropic} radiation creates an extended photon sphere. This region satisfies both stability criteria, while the HBH  \cite{Riojas:2026ytt,Riojas:2026mfu} mimics the optics and thermodynamics of an ordinary Schwarzschild black hole of mass $M$. 

\subsection{York Cavity Thermodynamics}\label{sec:York}

York's classic analysis of Schwarzschild thermodynamics in a finite cavity \cite{York:1986it} uses the Euclidean saddle-point method of Gibbons and Hawking \cite{Gibbons:1976ue}, but at finite distance; the fixed boundary data are the local (Tolman \cite{Tolman:1930zza}) temperature $T$ and area of the cavity wall.

If a Schwarzschild black hole of mass $M$ sits at the center of the cavity, then the wall temperature is the Hawking temperature redshifted to the wall: 
\begin{equation}
     T \equiv T(r)=(8 \pi M)^{-1}\left(1-\frac{2 M}{r}\right)^{-1 / 2}.
    \label{eq:TolmanTempYork}
\end{equation}
For fixed $(r,T)$, Eq.~\eqref{eq:TolmanTempYork} is generically double-valued in $M$. Real positive solutions exist only above a minimum wall temperature:
\begin{equation}
     T \geq \frac{\sqrt{27}}{8 \pi r}.
\label{eq:TolmanTempYorkCondition}
\end{equation}
\begin{figure}
    \centering
    \includegraphics[width=.9\linewidth]{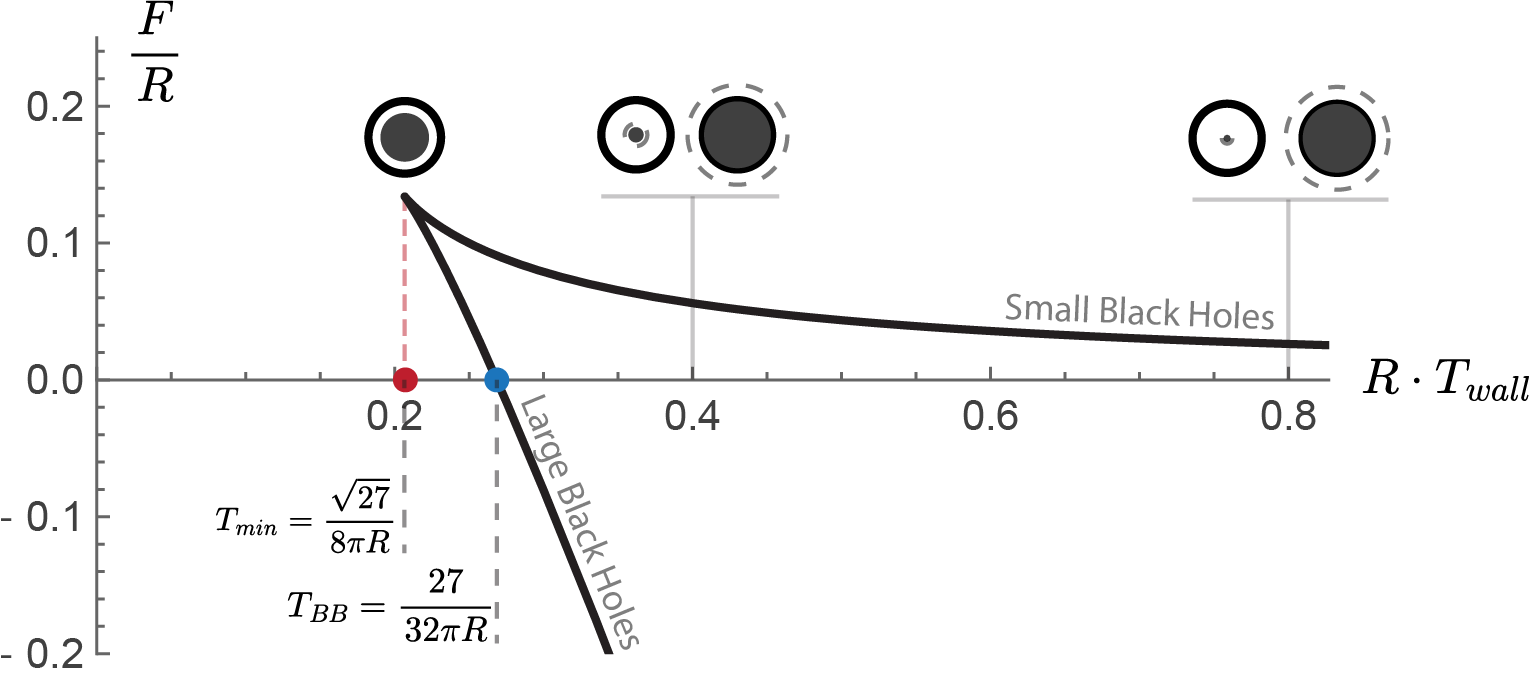}
    \caption{We illustrate York \cite{York:1986it} using the modern swallow-tail diagram. The cavity wall is a solid black line; the photon sphere is a dashed line. There are two branches of black holes at fixed $T_{wall}$. Large black holes are thermodynamically stable and dominate the canonical ensemble above $T_{BB}$ at $R = 9M/4$. At the minimum temperature $T_{min} = {\sqrt{27}}/{8 \pi R}$, the branches merge at the photon sphere. At the cusp $C_R^{-1}=0$; an infinitesimal traceless layer there leaves $T_\infty$ unchanged at fixed total mass (Eq.~\eqref{eq:specificheatchangetemp}). The HBH is built from them and satisfies $T(r)=\sqrt{27}/8 \pi r$ in the ocean.}
    \label{fig:swallowtailYork}
\end{figure}\noindent These two branches merge at the minimum temperature, which coincides with the local Tolman temperature at the photon sphere. The large black hole's photon sphere lies outside the cavity wall, while the small black hole's photon sphere lies inside. This is similar to the familiar situation for large and small black holes in AdS, see Fig.~\ref{fig:swallowtailYork}.

The Euclidean action includes the Einstein-Hilbert and Gibbons-Hawking-York terms: 
\begin{equation}
    I[g]=-\frac{1}{16 \pi G_N}\left(\int_M d^4 x \sqrt{|g|} R+2 \int_{\partial M} d^3 x \sqrt{|h|} K\right).
\end{equation}
With the subtraction prescription, the on-shell action takes the compact form:
\begin{equation}
I=12 \pi M^2-8 \pi M r+\beta r .
\end{equation}
The Brown-York energy \cite{York:1986it,Brown:1992br} and coarse-grained entropy follow from the standard prescription:
\begin{equation}
    E=\left(\frac{\partial I}{\partial \beta}\right)_A=r-r\left(1-\frac{2 M}{r}\right)^{1 / 2}\quad \quad S=\beta\left(\frac{\partial I}{\partial \beta}\right)_A-I=4 \pi M^2 .
\end{equation}
Here $E \ne M$ at finite distance, because the difference in the $S^2$ for Schwarzschild and Minkowski gives $M=E-\frac{1}{2} \frac{E^2}{r}$. The specific heat changes sign at the photon sphere:
\begin{equation}
C_R \equiv T\left(\frac{\partial S}{\partial T}\right)_A =\left(\frac{\partial E}{\partial T}\right)_A 
=8 \pi M^2\left[1-\frac{2 M}{r}\right]\left[\frac{3 M}{r}-1\right]^{-1}. 
\label{eq:fluctuations}
\end{equation}
In Sec.~\ref{sec:SpecificHeatPS} we show the specific heat at the photon sphere is inversely proportional to the cosmological constant for general spacetime dimensions; the particular definition of $E$ does not change this behavior.\footnote{Energy fluctuations would be naively expected to diverge at the photon sphere \cite{York:1986it}, but we show at third order in $F(E)$ that this effect is suppressed because the small and large black hole branches merge; see Appendix \ref{sec:fluctuations} for a calculation. Surprisingly, the path integral taken on the positive specific heat branch predicts large energy fluctuations on the order of \textit{tons} for a solar mass black hole, but with $\Delta E/E \sim A^{-1/3}$.}  

The large black hole branch has positive specific heat and dominates the canonical ensemble above the Buchdahl temperature $T_{\mathrm{BB}}=\frac{27}{32 \pi R}$, where $R = 9M/4$. The free energy $F=E-TS$ changes sign at the same point: 
\begin{equation}
    2 \leq r M^{-1}<2.25 \Leftrightarrow F<0 .
\end{equation}
York's results \cite{York:1986it} are similar to those of Hawking and Page in AdS \cite{Hawking:1982dh}, which were later interpreted in AdS/CFT by Witten \cite{Witten:1998zw,Witten:2024upt}.

\subsection{Mechanical Stability of Shells Around Black Holes}\label{sec:BLP}

The Israel junction conditions for a static thin shell surrounding a Schwarzschild black hole were first worked out by Frauendiener, Hoenselaers, and Konrad \cite{Frauendiener_1990}, who determined the surface energy density $\sigma$ and pressure $p$ required to support a shell at areal radius $R$ between Schwarzschild geometries of mass $m$ (interior) and $M$ (exterior).\footnote{They were largely interested in constraining energy extraction from black holes. The related problem of lowering a DEC saturating rope toward the horizon had been previously studied \cite{1972NPhS..240...77G} in this context. For a modern perspective on black hole mining, see \cite{Brown:2012un}. Mining is one element of AMPS \cite{Almheiri:2012rt}.}  See Fig.~\ref{fig:ShellAroundBH}. Defining $\alpha_{ \pm} \equiv \sqrt{f_{ \pm}(R)}$ with $f_{ \pm}=1-2 M_{ \pm} / R$, they found:
\begin{equation}
\sigma=\frac{1}{4 \pi R}\left(\alpha_m-\alpha_M\right), \quad p=\frac{1}{8 \pi R^2}\left(\frac{R-M}{\alpha_M}-\frac{R-m}{\alpha_m}\right) .
\label{eq:Frauendiener}
\end{equation}
They determined a minimum shell radius $R_{min} = 25 M/12$ by saturating DEC ($p= \sigma$) for $m=0$. They noted this was more compact than the Buchdahl radius $9 M/4$ for a fluid ball.\footnote{The other branch of DEC saturation, $p = - \sigma$, puts the shell at the photon sphere (Sec.~\ref{sec:IJC}).}

\begin{figure}
    \centering
    \includegraphics[width=0.35\linewidth]{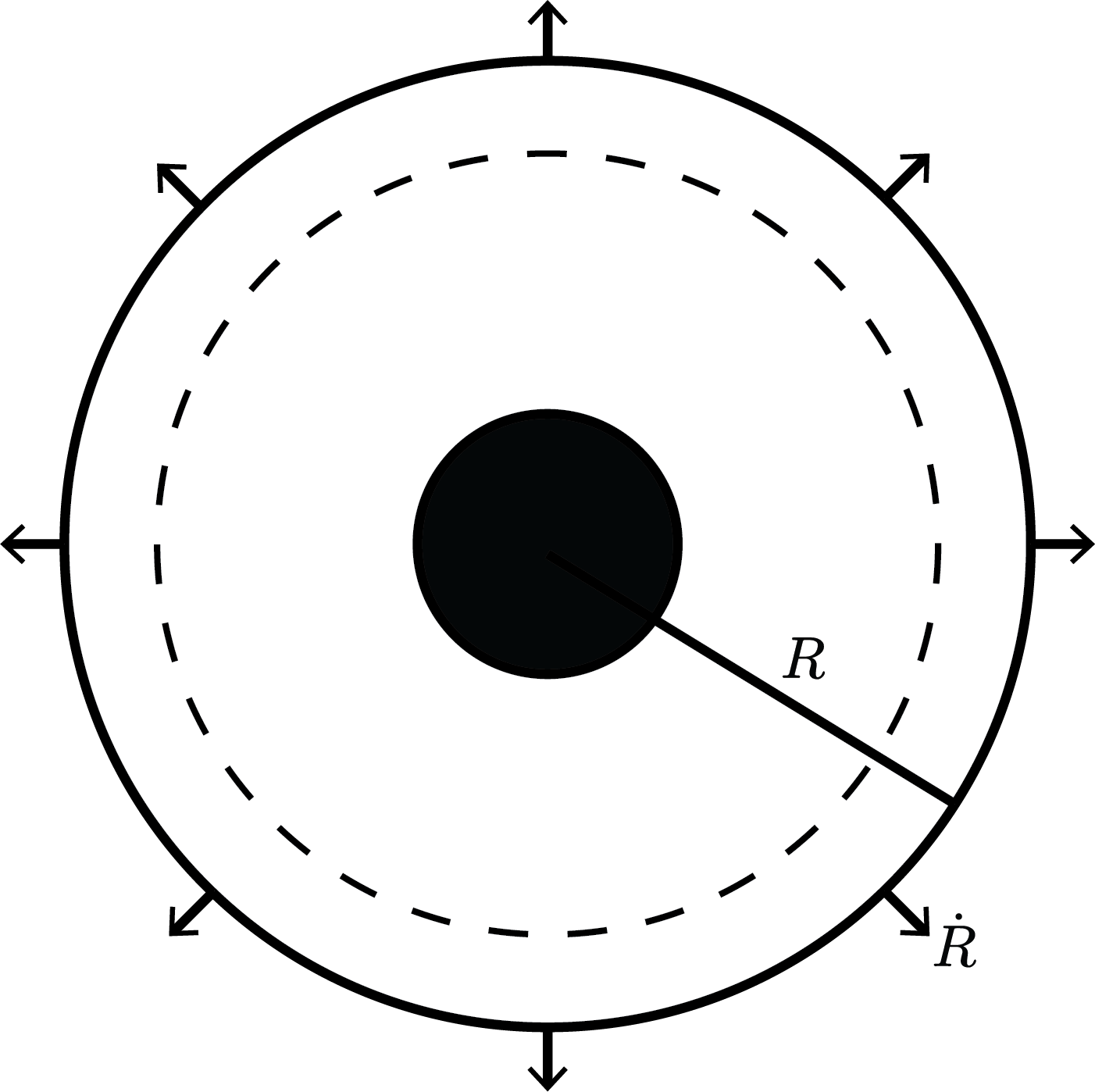}
    \caption{A thin shell around a black hole at radius $R$, with $\dot{R}\ne 0$. The stress tensor on the shell is $T^a_b = (-\sigma, p, \dots p). $ Shells placed inside the photon sphere in $3+1$ dimensional asymptotically flat space, depicted as a dotted line at $R=3M$, are mechanically unstable \cite{Brady:1991np}; see Fig.~\ref{fig:BLP}. }
\label{fig:ShellAroundBH}
\end{figure}

The connection between the photon sphere and the mechanical stability of thin shells around black holes was soon found by Brady, Louko, and Poisson (BLP) \cite{Brady:1991np}. The IJC for a radially moving shell satisfies a one-dimensional effective equation of motion:
\begin{equation}
    \beta_{-}-\beta_{+}=4 \pi r \sigma, \quad \quad  \beta_{ \pm}=\left(\dot{r}^2+1-2 m_{ \pm} / r\right)^{1 / 2},
\end{equation}
where $\beta_{ \pm}=\left(\dot{r}^2+1-2 m_{ \pm} / r\right)^{1 / 2}$ and $\dot{r}=d r / d \tau$. Together with the conservation equation, $\frac{d}{d \tau}\left(\sigma r^2\right)+p \frac{d}{d \tau}\left(r^2\right)=0$, the static equilibrium conditions are:
\begin{equation}
    \alpha_{-}-\alpha_{+}=4 \pi r_0 \sigma_0, \quad \quad \frac{m_{+}}{\alpha_+} -\frac{m_{-}}{ \alpha_{-}}=4 \pi r_0^2\left(\sigma_0+2 p_0\right),
    \label{eq:staticjunctionbh}
\end{equation}
where $\alpha_{ \pm}=\left(1-2 m_{ \pm} / r_0\right)^{1 / 2}$. 

To determine stability against radial perturbations, they linearized the equation of state around equilibrium as $p=\lambda(\sigma-\bar{\sigma})$, where $\lambda \equiv c_s^2$ is the squared speed of sound. Note $\lambda\equiv (dp/d\sigma)_0$ is independent of
$p_0/\sigma_0$. Integrating the conservation equation and using the equilibrium conditions to eliminate $\bar\sigma$ gives $\sigma$ as a function of $z \equiv r/r_0$: 
\begin{equation}
    (1+\lambda) \sigma / \sigma_0=F / 4 z^2, \qquad \lambda \equiv(d p / d \sigma)_0 = c_s^2.
\end{equation}
\begin{equation}
    F(z)=\left(4 \lambda+1-\frac{1}{ \alpha_{+} \alpha_{-}}\right) z^2+\left(3+\frac{1}{ \alpha_{+} \alpha_{-}}\right) z^{-2 \lambda} .
\end{equation}
Squaring the equation of motion twice yields an effective potential problem $r_0^2 \dot z^2 + V(z)=0$, with $V(1)=V'(1)=0$ by construction. Stability requires $V''(1)>0$, which they showed is equivalent to the positivity of: 
\begin{equation}
    G\left(\lambda, \alpha_{+}, \alpha_{-}\right)\equiv3(4 \lambda+1) \alpha_{+}^3 \alpha_{-}^3+4 \lambda \alpha_{+}^2 \alpha_{-}^2-\left(\alpha_{+}^2+\alpha_{-}^2+\alpha_{+} \alpha_{-}\right) .
    \label{eq:stabBLP}
\end{equation}
They illustrated this result as a stability triangle in the $(\alpha_-, \alpha_+)$ plane; the largest stable region is obtained with a luminal speed of sound, $\lambda=1$. See Fig.~\ref{fig:BLP}. In the light-shell limit, $\alpha_+ = \alpha_- = \bar{\alpha}$, the stability boundary is at: 
\begin{equation}
    \bar{\alpha}(\lambda)=\left(\frac{\sqrt{4 \lambda^2+36 \lambda+9}-2 \lambda}{3(4 \lambda+1)}\right)^{1 / 2} .
\end{equation}
For $\lambda=1$ this gives $\bar{\alpha}=1/\sqrt{3}$, i.e. $r_0 = 3m_-$, the photon sphere. Mechanically stable shells satisfy $\sigma \ge 2p$, with the curve where $\sigma = 2p$ intersecting $\lambda=1$ at the point where the shell becomes massless. They concluded that for massive shells, joint stability was obstructed, because York's result \cite{York:1986it} requires $r\le 3M$. 

\begin{figure}
    \centering
    \includegraphics[width=0.75\linewidth]{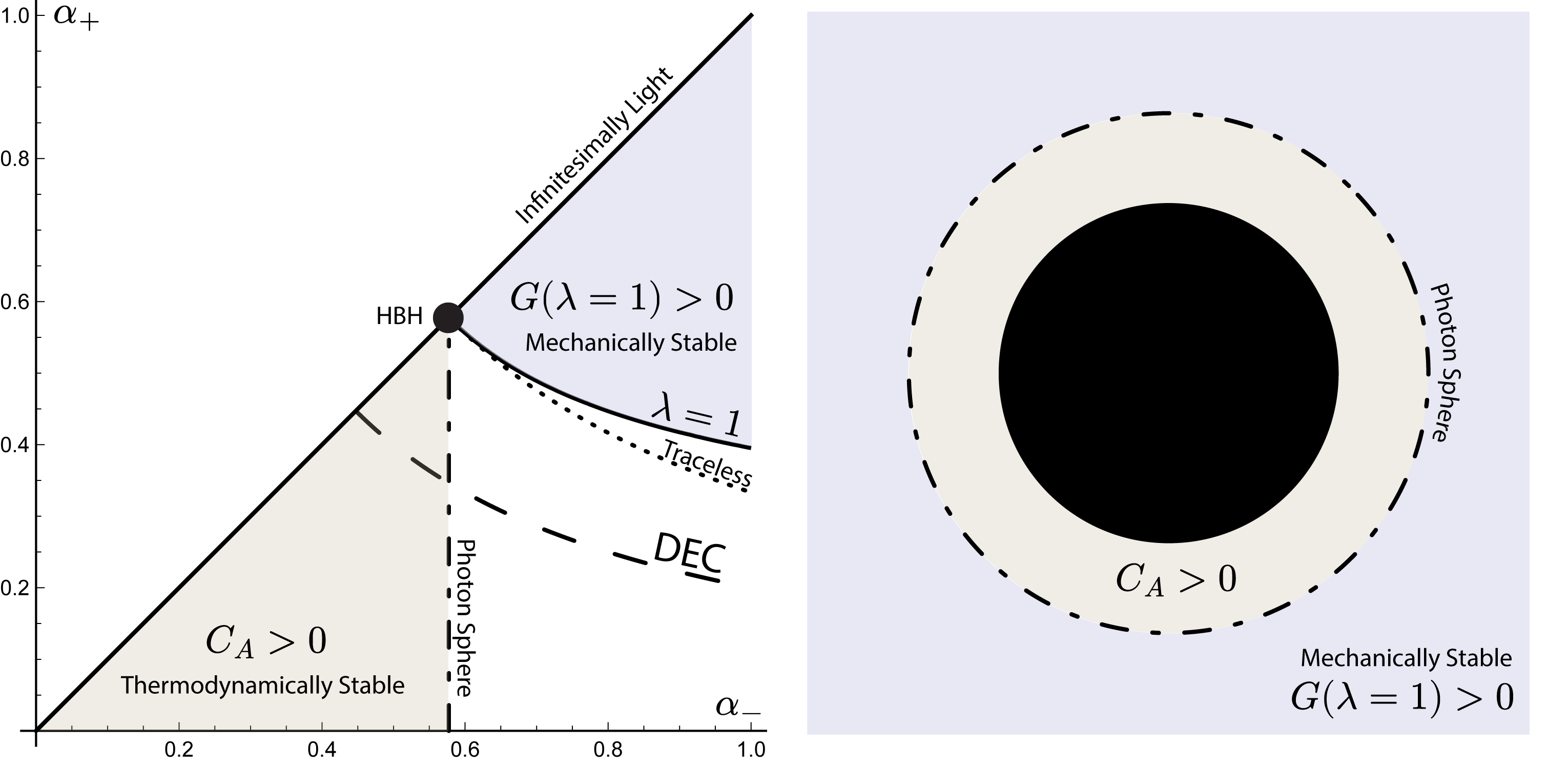}
    \caption{As shown by BLP \cite{Brady:1991np}, in asymptotically flat spacetime around a Schwarzschild black hole, the region of mechanical stability does not overlap with the thermodynamically stable region obtained by York \cite{York:1986it}. These regions lie on opposite sides of the photon sphere, where they meet for massless $\sigma=2p$ shells. This analysis is extended to AdS in Fig.~\ref{fig:BLPExtendedAdS}.}
    \label{fig:BLP}
\end{figure}

The flat-space stability triangle reviewed here corresponds to the limit $x \equiv r_0/L \rightarrow 0$ of the AdS analysis in Sec.~\ref{sec:Stability}, where we find that a finite window of joint mechanical and thermodynamic stability becomes available. For Schwarzschild, $f'(r) = 2m/r^2$ allows the equilibrium conditions to be expressed purely in terms of $\alpha_+ \alpha_-$; the stability analysis is instead rephrased in terms of $p_0/\sigma_0$ in Sec.~\ref{sec:Stability}, linking it more closely to the IJC in Sec.~\ref{sec:IJC}.

\subsection{Thermodynamics of Self-Gravitating Radiation}\label{sec:SWZ}

The thermodynamics of self-gravitating radiation within a spherical cavity was studied by Sorkin, Wald, and Zhang (SWZ) \cite{Sorkin:1981wd}. They showed that the problem of finding static equilibrium configurations can be recast as a variational principle for the total entropy. 

They considered isotropic radiation of fixed total mass $M = m(R)$, with an equation of state $P = \rho/3$, and entropy density $s = a \rho^{3/4}$. This was placed in a box of radius $R$. They showed that extremizing the total entropy functional
\begin{equation}I=\int_0^R\left(\frac{1}{r^2} \frac{d m}{d r}\right)^{3 / 4}\left(1-\frac{2 m(r)}{r}\right)^{-1 / 2} r^2 d r \quad \quad \delta m(r)\big|_{0, R}=0,
\label{eq:entropyfunctional}
\end{equation}
yields the Tolman-Oppenheimer-Volkoff (TOV) equation. In brief, they concluded that entropy extrema coincide with static equilibrium configurations. 

They exploited scale invariance, introducing the dimensionless variables: 
\begin{equation}
    \mu \equiv \frac{m(r)}{r}, \quad q \equiv \frac{d m}{d r}=4 \pi r^2 \rho, \quad z \equiv \ln r.
\end{equation}
This reduced the TOV equation to an autonomous system: 
\begin{equation}
    \frac{d \mu}{d z}=q-\mu, \quad \frac{d q}{d z}=\frac{2 q\left(1-4 \mu-\frac{2}{3} q\right)}{1-2 \mu} .
\end{equation}

Solution curves in the $(\mu, q)$ plane (see Fig.~\ref{fig:SWZfig}) can be understood using nullclines. The nullcline $q_V$, on which $d \mu/dz=0$, is the line $q=\mu$; curves cross it vertically. The nullcline $q_H$, on which $dq/dz=0$, is the line $4 \mu + \frac{2}{3} q = 1$; curves cross it horizontally. These nullclines intersect at the fixed point $(\mu_*,q_*)=(\frac{3}{14}, \frac{3}{14})$.  

The physically relevant solution $C$ is regular at the origin with slope $q/\mu \rightarrow 3$ and spirals into the fixed point. The eigenvalues of the linearized system near the fixed point are $-\frac{3}{4} \pm i \sqrt{47} / 4$, so the approach is oscillatory, with $\mu$ and $q$ executing damped oscillations around $\frac{3}{14}$ as $r \rightarrow \infty$.

\begin{figure}
    \centering
    \includegraphics[width=1\linewidth]{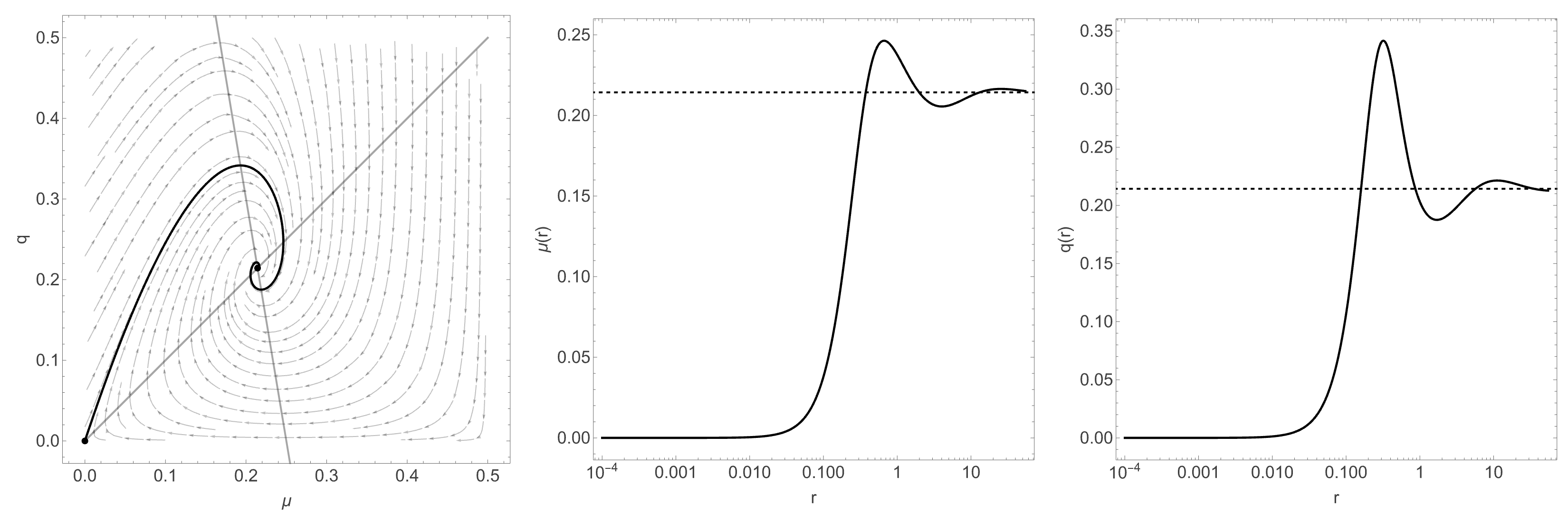}
    \caption{The solution to the Tolman-Oppenheimer-Volkoff (TOV) equation \cite{Sorkin:1981wd} for isotropic gas in thermal equilibrium. Here $q \equiv  4 \pi r^2 \rho(r)$ and $\mu \equiv m(r)/r$. The vertical and horizontal nullclines are shown, crossing at the fixed point. The curve $C$ spirals into the fixed point at $(\mu_*,q_*)=(\frac{3}{14}, \frac{3}{14})$; this approach is oscillatory, with $\mu$ and $q$ executing damped oscillations around $\frac{3}{14}$ as $r \rightarrow \infty$.}
    \label{fig:SWZfig}
\end{figure}

Each point on the curve is an equilibrium state at fixed box radius $R$, due to the scale invariance of this problem: $r \rightarrow \lambda r$, $m \rightarrow \lambda m$, $\rho \rightarrow \lambda^{-2} \rho$. Choosing where to terminate the curve at $r=R$ gives a one-parameter family of equilibria labeled by $M/R$. These are stable up to the first turning point in $\mu$, where $\mu_{max} \approx .246$, by the Poincaré turning point criterion \cite{Poincare:1886xr,Sorkin:1981jc,Green:2013ica}. The isotropic gas cannot reach the photon sphere because $\mu_{max} < \frac{1}{3}$. They noted that the entropy of these isotropic solutions scales as $S \sim M^{3/2}$. 

They knew that at fixed mass, Eq.~\eqref{eq:entropyfunctional} diverges at the local Schwarzschild radius due to the redshift factor $\epsilon = (1-2m/r)^{-1/2}$. The thermal radiation has finite wavelength $\lambda \sim \rho^{-1/4}$ with proper distance from the horizon $D(r) > \lambda(r)$. Near the horizon, $D(r)\sim r^{1 / 2} \epsilon^{1 / 2}$, giving the inequality $\epsilon^{-1 / 2} \lesssim r^{1 / 2} \rho^{1 / 4}$. Inserting this into Eq.~\eqref{eq:entropyfunctional} gives $S \lesssim M R$, and so $S \sim M^2$ for $R \sim 2M$. We note this seems to be closely related to frozen stars \cite{Brustein:2021lnr,Brustein:2023hic}, which have the same energy density scaling.\footnote{The maximally entropic energy density profile obtained by SWZ is related to the proposal of Brustein, Medved, and Simhon \cite{Brustein:2021lnr,Brustein:2023hic}, who called these configurations ``frozen stars", see also \cite{Sorkin:1981wd}. The $G_{rr}$ component of Einstein's equations Eq.~\eqref{eq:EErrwithj(r)}, for $\widehat{m}=r/2$, gives $P_r = -\rho$. The TOV equation in Eq.~\eqref{eq:anisotropic_TOV} gives $P_\perp = 0$. Then Einstein's equations require $\rho(r)=(4 \pi r^2)^{-1} (d \widehat{m}/dr)$; since $\widehat{m}=r/2$, we have $\rho = 1/8 \pi r^2$. } 

Now let us summarize. Consider a thermal bath with an inner surface at $R_0$ around a black hole. York's calculation shows positive specific heat requires $R_0 \le 3M$;  however, the mechanical stability analysis by BLP for shells requires $R_0 \ge  3M$. We conclude that a single shell will not suffice. In addition, isotropic radiation cannot reach the photon sphere because $\mu_{max} \le 1/3$. The missing ingredient is anisotropy, as we will see in Sec.~\ref{sec:ExtendedPhotonSpheres}.

\section{Time-Dependent Israel Junction Conditions}\label{sec:IJC}

The well-known Israel junction conditions in $D$ bulk dimensions are \cite{Israel:1966rt},
\begin{equation}
[K_{ab}]
=-\kappa_D^{\,2}\!\left(T_{ab}-\frac{1}{d-1}\,\gamma_{ab}\,T\right),
\quad
\kappa_D^{\,2}=8\pi G_D,
\label{eq:IJC}
\end{equation}
where $\gamma_{ab}$ is the induced metric on the shell, $d=D-1$ is the shell worldvolume dimension, and $T\equiv \gamma^{ab}T_{ab}$. The Israel layers it describes are best viewed as thin transition regions.

Here we extend an approach developed by Per Kraus \cite{Kraus:1999it} to more general spherically symmetric spacetimes. Let the shell have worldvolume dimension $d=n+1$; take:
\begin{equation}
T^{a}{}_{b}=\mathrm{diag}(-\sigma,p,p,\dots,p),
\label{eq:Tab_diag}
\end{equation}
with $n$ copies of $p$. The trace is $T=-\sigma+n p$. The two independent jump conditions are:
\begin{equation}
\frac{[K^{\tau}{}_{\tau}]}{8 \pi G_D}
=\left(\frac{n-1}{n}\sigma+p\right) \quad \quad \frac{[K^{\theta}{}_{\theta}]}{8 \pi G_D}
=-\left(\frac{\sigma}{n}\right).
\label{eq:jump_Ktautau}
\end{equation}
Here $K^{\theta}{}_{\theta}=K^{\phi}{}_{\phi}$, and so on, for the angular directions. Note $[K]$ vanishes when the stress-energy tensor on the Israel layer is traceless:  
\begin{equation}
    [K]=\left[K_a^a\right]=\frac{8 \pi G_D}{n}(n p-\sigma) = \frac{8 \pi G_D}{n}T^a_a.
\end{equation}
The ratio Eq.~\eqref{eq:jump_Ktautau} yields a compact expression for the Israel junction condition:
\begin{equation}
\frac{p}{\sigma}
=-\frac{n-1}{n}-\frac{1}{n}\frac{[K^{\tau}{}_{\tau}]}{[K^{\theta}{}_{\theta}]}.
\label{eq:p_over_sigma}
\end{equation}
Now we determine the motion of these thin Israel layers in a time-independent bulk spacetime. Consider the following spherically symmetric, time-independent bulk ansatz: 
\begin{equation}
\dd s^2=-f(r)\,\dd t^2+\frac{1}{j(r)}\,\dd r^2+r^2\dd\Sigma_n^2.
\label{eq:metric_general}
\end{equation}
Parameterize the shell, with proper time $\tau$, as $X^\mu= \left( t(\tau), R(\tau), \theta(\tau), \phi(\tau), \dots\right)$. Holding the angles fixed, the tangent vector $e_\tau^\mu\equiv \frac{\partial X^\mu}{\partial \tau}\equiv u^\mu$ is:
\begin{equation}
u^{\mu}=(\dot t,\dot R,0,\dots,0),
\qquad
u^{\mu}u_{\mu}=-1.
\label{eq:tangentvector}
\end{equation}
From Eq.~\eqref{eq:metric_general}-\eqref{eq:tangentvector},
\begin{equation}
\dot t^{\,2}=\frac{j(R)+\dot R^{\,2}}{f(R)\,j(R)}.
\label{eq:tdot}
\end{equation}
Let $n^{\mu}$ be the unit normal, with $n^{\mu}n_{\mu}=+1$ and $n_{\mu}u^{\mu}=0$. This gives:
\begin{equation}
n^{r}=\epsilon\,\sqrt{j(R)+\dot R^{\,2}},
\qquad \epsilon=\pm 1.
\label{eq:nr}
\end{equation}
The extrinsic curvature is defined as $K_{ab}=e_a^{\mu}e_b^{\nu}\nabla_{\mu}n_{\nu}$.
The angular components are: 
\begin{equation}
\label{eq:Kthetatheta}
K^{\theta}{}_{\theta}=\epsilon\,\frac{\sqrt{j(R)+\dot R^{\,2}}}{R}.
\end{equation}
For the $\tau\tau$ component it is convenient to write:
\begin{equation}
\label{eq:Ktautau_accel}
K_{\tau\tau}=-n_{\mu}a^{\mu},
\qquad
a^{\mu}=u^{\nu}\nabla_{\nu}u^{\mu}=\frac{d u^\mu}{d \tau}+\Gamma_{\alpha \beta}^\mu u^\alpha u^\beta.
\end{equation}
Direct calculation gives:
\begin{equation}
\label{eq:Ktautau}
K^{\tau}{}_{\tau}
=\frac{
\epsilon\left[ 
\ddot R+\frac{f'(R)}{2f(R)}\bigl(j(R)+\dot R^{\,2}\bigr)
-\frac{j'(R)}{2j(R)}\dot R^{\,2}
\right]}{{\sqrt{j(R)+\dot R^{\,2}}}}.
\end{equation}
See Fig.~\ref{fig:ShellAroundBH} for an illustration of this configuration, for the special case where $R$ is the radial direction to a spherical shell placed around a black hole.

\subsection{Static Shells around Black Holes}
To examine the stability of these shells we consider the static case, where $\dot R=\ddot R=0$. Let $f=j$ on each side of the shell, which holds when the \textit{bulk} stress-energy tensor vanishes. Then Eq.~\eqref{eq:Kthetatheta} and Eq.~\eqref{eq:Ktautau} reduce to:
\begin{equation}
K^{\theta}{}_{\theta}=\epsilon\,\frac{\sqrt{f(R)}}{R},
\qquad
K^{\tau}{}_{\tau}=\epsilon\,\frac{f'(R)}{2\sqrt{f(R)}}.
\label{eq:K_static}
\end{equation}
Here we consider shells with $\epsilon_{+}=+1$ and $ \epsilon_{-}=+1$:
\begin{equation}
[K^{\tau}{}_{\tau}]
=\frac{f_+'(R)}{2\sqrt{f_+(R)}}-\frac{f_-'(R)}{2\sqrt{f_-(R)}},
\label{eq:jump_Ktau_geom}
\end{equation}
\begin{equation}
[K^{\theta}{}_{\theta}]
=\frac{\sqrt{f_+(R)}-\sqrt{f_-(R)}}{R}.
\label{eq:jump_Ktheta_geom}
\end{equation}
Suppose each side is Schwarzschild-Tangherlini:
\begin{equation}
\label{eq:f_ST}
f(r)=1-\frac{\mu}{r^{\,n-1}},
\qquad
f'(r)=\frac{(n-1)\mu}{r^{\,n}}.
\end{equation}
It is convenient to define:
\begin{equation}
\label{eq:alpha_def}
\alpha_{\pm}\equiv \sqrt{f_{\pm}(R)}.
\end{equation}
Then Eqs.~\eqref{eq:jump_Ktau_geom}-\eqref{eq:alpha_def} yield:
\begin{equation}
\label{eq:ratio_Ktau_Ktheta}
\frac{[K^{\tau}{}_{\tau}]}{[K^{\theta}{}_{\theta}]}
=-\frac{n-1}{2}\left(1+\frac{1}{\alpha_{+}\alpha_{-}}\right),
\end{equation}
and therefore, via Eq.~\eqref{eq:p_over_sigma}:
\begin{equation}
\frac{p}{\sigma}
=\frac{n-1}{2n}\left(\frac{1}{\alpha_{+}\alpha_{-}}-1\right).
\label{eq:p_over_sigma_alpha}
\end{equation}
For the traceless case $T=0$, appropriate for massless radiation, one has $\sigma = n p$, so:
\begin{equation}
\label{eq:alpha_product}
\alpha_{+}\alpha_{-}=\frac{n-1}{n+1}.
\end{equation}
For example, for $n=2$ we have $\alpha_{+} \alpha_{-}=1 / 3$. 

In particular, Eq.~\eqref{eq:alpha_product} can be used to obtain a recursion relation for layered shells. Writing $\alpha_\pm^2=1-\mu_\pm/R^{\,n-1}$ and solving Eq.~\eqref{eq:alpha_product} for $R^{n-1}$ gives:
\begin{equation}
\label{eq:R_solution}
R^{n-1}
=\frac{(n+1)\left((n+1)\bar\mu
\pm \sqrt{(n-1)^2\bar\mu^{\,2}+n(\delta\mu)^2}\right)}{4n},
\end{equation}
where we defined the ``average" $\bar\mu\equiv \frac{\mu_{+}+\mu_{-}}{2}$ and the ``difference" $\delta\mu\equiv \mu_{+}-\mu_{-}$.

If $\mu_{+}=\mu_{-}\equiv\mu$ then Eq.~\eqref{eq:R_solution} reduces to
\begin{equation}
\label{eq:photon_sphere}
R^{\,n-1}=\frac{n+1}{2}\,\mu\equiv R_{ps}^{n-1},
\end{equation}
which is the standard photon-sphere radius $R_{ps}$ in $D=n+2$ dimensions. For $n=2$ we have $R = 3 M$. Then in the infinitesimally light shell limit, static nested shells are located at the photon sphere.

For example, for $n=2$ (4D) with $\mu_-=2m$ and $\mu_+=2M$, Eq.~\eqref{eq:R_solution} becomes:
\begin{equation}
R=\frac{3}{8}\left(3m+3M\pm\sqrt{9m^2-14mM+9M^2}\right)= \frac{9}{4} \bar{m} + \frac{3}{4} \bar{m} \sqrt{1 + 2 \left(\frac{\delta m}{\bar{m}}\right)^2},
\label{eq:nesting}
\end{equation}
where $3 m < R < 3M$, $\delta m \equiv M-m$, and $\bar{m} \equiv \frac{M+m}{2}$. One layer can be self-supporting; this recovers the Buchdahl-Andréasson bound \cite{Buchdahl:1959zz,Andreasson:2007ck}, since $m=0$ gives $R=\frac{9}{4}M$. In Sec.~\ref{sec:emergentisraellayers}, we will see that anisotropic gases satisfy these nesting relations at zero radial pressure.

\subsection{Static $Z_2$ Pure Tension Shells Sit at the Photon Sphere}

Here we show that static pure tension $Z_2$ symmetric shells (e.g. Randall-Sundrum branes)\footnote{In the Randall-Sundrum \cite{Randall:1999ee,Randall:1999vf} and Karch-Randall models \cite{Karch:2000ct}, one often takes a $Z_2$ reflection symmetry across the brane (an orbifold identification).} are located at the photon sphere of the geometry. The IJCs for $Z_2$ symmetry have $\epsilon_{+}=+1$ and $\epsilon_{-}=-1$, which gives:
\begin{equation}
    \left[ K^\theta{ }_\theta\right]=\frac{2\sqrt{j(R)+\dot{R}^2}}{R}
\end{equation}
\begin{equation}
\left[K_\tau^\tau\right]=\frac{2\left[\ddot{R}+\frac{f^{\prime}(R)}{2 f(R)}\left(j(R)+\dot{R}^2\right)-\frac{j^{\prime}(R)}{2 j(R)} \dot{R}^2\right]}{\sqrt{j(R)+\dot{R}^2}}
\end{equation}
Then Eq.~\eqref{eq:nr} and Eq.~\eqref{eq:Ktautau_accel} give:
\begin{equation}
    \frac{\left[K^\tau{ }_\tau\right]}{\left[K^\theta{ }_\theta\right]}=\frac{K^\tau{ }_\tau}{K^\theta{ }_\theta}=\frac{\epsilon a^r / n^r}{\epsilon n^r / R}=\frac{a^r R}{j(R)+\dot{R}^2},
\end{equation}
where $(n^r)^2$ eliminated $\epsilon=\pm 1$. The jump condition in Eq.~\eqref{eq:p_over_sigma} becomes:
\begin{equation}
    \frac{p}{\sigma}=-\frac{n-1}{n}-\frac{1}{n} \frac{a^r R}{j(R)+\dot{R}^2}.
    \label{eq:psigmaratioaccel}
\end{equation}
For a pure tension brane, $a^r=\frac{j(R) + \dot{R}^2}{R}$, and Eq.~\eqref{eq:Ktautau_accel} gives:
\begin{equation}
    \ddot{R}+\frac{f^{\prime}(R)}{2 f(R)}\left(j(R)+\dot{R}^2\right)-\frac{j^{\prime}(R)}{2 j(R)} \dot{R}^2=\frac{j(R)+\dot{R}^2}{R} .
\end{equation}
Now supposing the brane is static, we have:
\begin{equation}
    \frac{R f^{\prime}(R)}{2 f(R)}=1 \implies R = R_{ps}.
    \label{eq:photonspherecondition}
\end{equation}
The converse is similar; for a static $Z_2$ shell where $Rf'(R)/2f(R)=1$, inverting the above algebraic relations shows $p = -\sigma$, so the brane is pure tension. This is consistent with known results for Randall-Sundrum branes; see Appendix \ref{sec:RSbranes} for additional discussion.

\section{Stability of Shells Around Black Holes}\label{sec:Stability}

Here we generalize the analysis of Brady, Louko, and Poisson \cite{Brady:1991np}, reviewed in Sec.~\ref{sec:BLP}, to general spherically symmetric spacetimes. In Sec.~\ref{sec:thermAdS} and Sec.~\ref{sec:mechAdS}, the flat-space stability obstruction of Sec.~\ref{sec:obstruction} is generalized to AdS. There, a finite window of joint stability opens.

The specific heat at the photon sphere of a $d$-dimensional black hole is shown to be inversely proportional to the cosmological constant in Sec.~\ref{sec:SpecificHeatPS}. This is independent of how the energy is defined, provided it is monotonic in the mass enclosed within a finite region.

\subsection{An Extended Formalism for Shell Stability}

The Bianchi identity gives $\nabla^\mu T_{\mu \nu}=0$, which leads to the usual expression for energy-momentum conservation on the thin Israel layer: 
\begin{equation}
    \dot{\sigma}+2 \frac{\dot{r}}{r}(\sigma+p)=0.
    \label{eq:energyconservationshell}
\end{equation}
This can be rewritten using the chain-rule, since $\sigma(R) = \sigma(R(\tau))$ implies $\dot{\sigma}=(d \sigma / d R) \dot{R}$:
\begin{equation}
    \left(\frac{d \sigma}{d R}\right) \dot{R}+2 \frac{\dot{R}}{R}(\sigma+p)=0 \implies \frac{d \sigma}{d R}=-\frac{2(\sigma+p)}{R}.
    \label{eq:propagation}
\end{equation}
Here we fixed $\dot{\theta}=\dot{\phi}=0$, and assumed the equation of state for the brane is not pure tension \textit{after} the perturbation, $\sigma(R(\tau)) \ne -p(R(\tau))$. Then $\dot\sigma\ne0$ by Eq.~\eqref{eq:propagation}, sound waves propagate in the material, and we may cancel $\dot{R}$.\footnote{Pure tension branes with rigid equations of state, e.g. unstabilized Randall-Sundrum branes, are unstable. The brane is stabilized by its response to a perturbation; when the brane is pushed inward, the stiffness of the brane provides a restoring force that keeps it there and sound waves propagate. If the brane remains pure tension after the perturbation, a straightforward calculation at the photon sphere yields:
\begin{equation}
    V '' (r_0) = r_0^2 h''(r_0)>0, \quad \quad h(r_0) \equiv \frac{f(r_0)}{r_0^2}.
\end{equation}
For a ``thick" brane where an extended photon sphere exists, $h''=0$ and the stability is marginal. This occurs when the extended photon sphere condition is satisfied, see App.~\ref{sec:UCO}.}

The IJCs in Eq.~\eqref{eq:IJC}, Eq.~\eqref{eq:jump_Ktau_geom}, and Eq.~\eqref{eq:jump_Ktheta_geom} require that:
\begin{equation}
    \sqrt{\dot{R}^2 + f_{-}(R)}- \sqrt{\dot{R}^2 + f_{+}(R)} = \frac{ \kappa^2}{2} \sigma(R) R.
\label{eq:junctionsigma}
\end{equation}
This gives $\dot{R}^2 + V(R) = 0 $, where:
\begin{equation}
    V(R) = f_{-}(R) - \frac{\left(f_{+}(R)-f_{-}(R)-S^2\right)^2}{4 S^2},
    \label{eq:potentialformulation}
\end{equation}
and $S(R) \equiv \frac{\kappa^2}{2}\sigma(R) R$. Then a shell at $R_0$ is mechanically stable if: 
\begin{equation}
    V(R_0)= V'(R_0) = 0 \quad \quad V''(R_0)>0.
\end{equation}

To determine $\sigma(r)$ in the conservation equation Eq.~\eqref{eq:propagation}, we expand around the equilibrium value as $p(\sigma) \simeq p_0+\lambda\left(\sigma-\sigma_0\right)$, with $\lambda \equiv(d p / d \sigma)_0$. This perturbation changes the equation of state; note $\lambda \equiv c_s^2$ is the squared speed of sound for waves propagating on the shell. The shift vanishes at $\bar \sigma \equiv \sigma_0 - p_0/\lambda$, in terms of which Eq.~\eqref{eq:propagation} reads:
\begin{equation}\label{eq:diffeopert}
    \frac{d \sigma}{d r}+\frac{2(1+\lambda)}{r} \sigma=\frac{2 \lambda}{r} \bar{\sigma} , \qquad \lambda \equiv(d p / d \sigma)_0=c_s^2 , \qquad \bar{\sigma} \equiv \sigma_0-\frac{p_0}{\lambda}.
\end{equation}
For Eq.~\eqref{eq:diffeopert} the boundary condition $\sigma(r_0)=\sigma_0$ gives:
\begin{equation}\label{eq:linearizedpert}
    (1+\lambda) \frac{\sigma}{\sigma_0}=\lambda \frac{\bar{\sigma}}{\sigma_0}+\left(1+\lambda-\lambda \frac{\bar{\sigma}}{\sigma_0}\right)\left(\frac{r}{r_0}\right)^{-2(1+\lambda)} .
\end{equation}
To eliminate $\bar \sigma$ we use the  IJCs for a static shell, which in asymptotically flat space give:
\begin{equation}
    \frac{p}{\sigma}=\frac{n-1}{2 n}\left(\frac{1}{\alpha_{+} \alpha_{-}}-1\right).
\end{equation}
As an example, let us consider a perturbation for $n=2$ around the static value:
\begin{equation}\label{eq:flatspacepert}
    \frac{p_0}{\sigma_0}=\frac{1-\alpha_{+} \alpha_{-}}{4 \alpha_{+} \alpha_{-}} .
\end{equation}
Expanding the equation of state around the static solution gives:
\begin{equation}\label{eq:expandflat}
    \frac{\bar{\sigma}}{\sigma_0}=1-\frac{1}{\lambda} \frac{p_0}{\sigma_0}=\frac{4 \lambda+1-\frac{1}{\alpha_{+} \alpha_{-}}}{4 \lambda} .
\end{equation}
Then we can define $z = r/r_0$ and get: 
\begin{equation}
    (1+\lambda) \frac{\sigma}{\sigma_0}=\frac{1}{4}\left(4 \lambda+1-\frac{1}{\alpha_{+} \alpha_{-}}\right)+\frac{1}{4}\left(3+\frac{1}{\alpha_{+} \alpha_{-}}\right) z^{-2(1+\lambda)} .
\end{equation}
This gives, defining $F(z)$ for the asymptotically flat case: 
\begin{equation}\label{eq:Fdefflat}
        (1+\lambda) \frac{\sigma}{\sigma_0}=\frac{F(z)}{4 z^2}, \qquad F(z) \equiv\left(4 \lambda+1-\frac{1}{\alpha_{+} \alpha_{-}}\right) z^2+\left(3+\frac{1}{\alpha_{+} \alpha_{-}}\right) z^{-2 \lambda}
\end{equation}
Writing the IJCs using $z = r/r_0$ and $\dot{z} = \dot{r}/r_0$, with $V(z) \equiv V\left(r_0 z\right)$, we recover BLP \cite{Brady:1991np}:
\begin{equation}
    r_0^2 \dot{z}^2+V\left(z\right)=0 .
\end{equation}
The more general case (not asymptotically flat) is more difficult. The IJCs are:
\begin{equation}
    \frac{p}{\sigma}=-\frac{n-1}{n}-\frac{1}{n} \frac{\left[K_\tau^\tau\right]}{\left[K^\theta{ }_\theta\right]}
\end{equation}
The jumps read: 
\begin{equation}
\left[K_\tau^\tau\right]=\frac{f_{+}^{\prime}(R)}{2 \sqrt{f_{+}(R)}}-\frac{f_{-}^{\prime}(R)}{2 \sqrt{f_{-}(R)}}, \quad \quad \left[K_\theta^\theta\right]=\frac{\sqrt{f_{+}(R)}-\sqrt{f_{-}(R)}}{R}.
\end{equation}
Specializing to $n=2$ (4D) gives:
\begin{equation}
    \frac{p(r_0)}{\sigma(r_0)}=-\frac{1}{2} -\frac{1}{2}\frac{\frac{f'_{+}\left(r_0\right)}{2 \sqrt{f_{+}\left(r_0\right)}}-\frac{f'_{-}\left(r_0\right)}{2 \sqrt{f_{-}\left(r_0\right)}}}{\frac{\sqrt{f_{+}\left(r_0\right)}-\sqrt{f_{-}\left(r_0\right)}}{r_0}}
    \label{eq:ratioptosigma}
\end{equation}

A challenge when generalizing the analysis of BLP \cite{Brady:1991np} is that $f'(r_0)$ simplifies in asymptotically flat space (Eq.~\eqref{eq:flatspacepert}), but not in general. The first part of Eq.~\eqref{eq:expandflat}:
\begin{equation}
    \frac{\bar{\sigma}}{\sigma_0}=1-\frac{1}{\lambda} \frac{p_0}{\sigma_0},
\end{equation}
inserted into Eq.~\eqref{eq:linearizedpert} gives:
\begin{equation}
    (1+\lambda) \frac{\sigma(z)}{\sigma_0}=\left(\lambda-\frac{p_0}{\sigma_0}\right)+\left(1+\frac{p_0}{\sigma_0}\right) z^{-2(1+\lambda)} .
\end{equation}
Repeating the same procedure as in Eq.~\eqref{eq:Fdefflat}, we define $F(z)$:
\begin{equation}
    (1+\lambda) \frac{\sigma}{\sigma_0}\equiv\frac{F(z)}{4 z^2} \implies F(z)\equiv4\left(\lambda-\frac{p_0}{\sigma_0}\right) z^2+4\left(1+\frac{p_0}{\sigma_0}\right) z^{-2 \lambda}.
\end{equation}
To obtain $S(R)$ we use Eq.~\eqref{eq:junctionsigma} for $\dot{R}=0$. Here $z\equiv r/r_0$, so $z=1$ is the shell's location:
\begin{equation}
    S_0 \equiv S(1)=\sqrt{f_{-}\left(r_0\right)}-\sqrt{f_{+}\left(r_0\right)}, \quad \quad S(z)=S_0 \frac{F(z)}{4(1+\lambda) z}.
\end{equation}
\begin{equation}\label{eq:Vstabcon}
    V(z)=f_{-}\left(r_0 z\right)-\frac{\left(f_{+}\left(r_0 z\right)-f_{-}\left(r_0 z\right)-S(z)^2\right)^2}{4 S(z)^2}.
\end{equation}
To determine stability at $z=1$ (where, by construction, $V(z)=V'(z)=0$) we only need to check if $V(z)$ is concave up or concave down. For AdS see Sec.~\ref{sec:mechAdS}. In passing, we note traceless light shells at the photon sphere can be viewed as orbiting massless particles.\footnote{Defining $f_{-}(R)\equiv f(R)$ for notational simplicity, write $f_{+}(R)=f(R)+\delta f(R)$ and take the light-shell limit of Eq.~\eqref{eq:potentialformulation}. Noting $|\delta f| = 2 |\delta m|/R$, a short calculation yields: 
\begin{equation}
    \left(\frac{\dot{R}}{R}\right)^2=\frac{1}{b(R)^2}-\frac{f(R)}{R^2}, \quad \quad \frac{1}{b(R)^2}=\left(\frac{2 |\delta m|}{\kappa^2 \sigma(R) R^3}\right)^2.
\end{equation}
Energy conservation for the shell in Eq.~\eqref{eq:energyconservationshell}, for an equation of state $p \equiv w \sigma$, gives $\sigma(R) \propto R^{-2(1+w)}$. For $w=1/2$ (traceless) shells $\sigma(R)R^3 \propto (\sigma A)R= ER$ is conserved. Indeed, evaluating Eq.~\eqref{eq:potentialformulation} at $r_0=3m$ and inserting it into Eq.~\eqref{eq:ratiotoredshift} gives $\sigma_0 = 2 p_0$. Taking $R d\phi = d\tau$ gives the null geodesic equation.}

\subsection{Thermal Stability in anti-de Sitter Space}\label{sec:thermAdS}

The regions of mechanical and thermodynamic stability do not coincide in flat space; they only intersect at $r_0=3M$. However, we find that in AdS the regions overlap. To avoid ambiguity in this section, we write $C_{r_0}$ instead of $C_R$, which is the standard notation. 

To repeat the York computation for the specific heat $C_{r_0}$ in AdS, we write:
\begin{equation}
    T\left(r_0\right)=T_H(m, L) \sqrt{\frac{f_{\mathrm{AdS} }\left(r_0\right)}{f(r_0)}}.
\end{equation}
Here $f_{\mathrm{AdS}}=1+\frac{r_0^2}{L^2}$ and $f=1+\frac{r_0^2}{L^2}-\frac{2 m}{r_0}$. For convenience we define: 
\begin{equation}
    g\left(r_0\right) \equiv \sqrt{\frac{f_{\mathrm{AdS}}\left(r_0\right)}{f\left(r_0\right)}}, \quad \Rightarrow \quad T\left(r_0\right)=T_H g\left(r_0\right) .
\end{equation}
The specific heat within a cavity at some $r_0$ is: 
\begin{equation}
    C_{r_0} \equiv\left(\frac{d E}{d T}\right)_{r_0}=\frac{(d E / d m)_{r_0}}{(d T / d m)_{r_0}}.
\end{equation}
Following York \cite{York:1986it}, the thermodynamic energy is: 
\begin{equation}
    E\left(r_0\right)=r_0 f_{\mathrm{AdS}}\left(r_0\right)\left(1-\sqrt{\frac{f\left(r_0\right)}{f_{\mathrm{AdS}}\left(r_0\right)}}\right).
\end{equation}

A straightforward computation gives  $\left(dE/dm\right)_{r_0}=g\left(r_0\right)$. This useful simplification gives an expression for $C_{r_0}$ where $g$ cancels:
\begin{equation}
    \frac{d T}{d m}=g \frac{d T_H}{d m}+T_H \frac{d g}{d m}=g(r_0)\left(\frac{d T_H}{d m}+\frac{T_H}{r_0 f\left(r_0\right)}\right) \implies C_{r_0}^{-1}={\frac{d T_H}{d m}+\frac{T_H}{r_0 f\left(r_0\right)}}.
     \label{eq:general_specific_heat}
\end{equation}

A more general discussion can be found in Sec.~\ref{sec:SpecificHeatPS}, where we note the definition of $E(r_0)$ is unimportant for this analysis. The specific heat changes sign at the pole: 
\begin{equation}
    r_0 f\left(r_0\right)=-\frac{T_H}{d T_H / d m} \implies r_0+\frac{r_0^3}{L^2}-2 m=-\frac{T_H(m, L)}{d T_H / d m}.
    \label{eq:polelocation}
\end{equation}
For instance, for AdS Schwarzschild black holes:
\begin{equation}
    m=\frac{r_{+}}{2}\left(1+\frac{r_{+}^2}{L^2}\right), \quad T_H=\frac{f'(r_+)}{4 \pi} =\frac{1}{4 \pi r_{+}}\left(1+\frac{3 r_{+}^2}{L^2}\right) .
\end{equation}
This gives, for the first term in $C_{r_0}^{-1}$:
\begin{equation}
    \frac{d T_H}{d m}=\frac{d T_H / d r_{+}}{d m / d r_{+}}=\frac{\left(\frac{3}{L^2}-\frac{1}{r_{+}^2}\right)}{2 \pi\left(1+\frac{3 r_{+}^2}{L^2}\right)} .
\end{equation}
The first term in $C_{r_0}^{-1}$ reflects the well-known spinodal transition, where the specific heat diverges at $r_+ = \frac{L}{\sqrt{3}}$. The finite distance computation gave the additional second term, which shifts the location of the pole in $C_{r_0}$. 

This can be understood as follows. The specific heat changes sign only on the small black hole branch, where $\frac{dT_H}{dm}<0$. In particular:
\begin{equation}
    r_0 f\left(r_0\right)=r_0+\frac{r_0^3}{L^2}-2 m, \quad \frac{d}{d r_0}\left(r_0 f\left(r_0\right)\right)=1+\frac{3 r_0^2}{L^2}>0.
    \label{eq:psigmaratioAdS}
\end{equation}
The LHS is strictly increasing in $r_0$, so there can be at most one radius $r_0^*$ where the specific heat changes sign for fixed $(m,L)$.\footnote{This mirrors the mechanical stability analysis in the next section, because the $A(x) = (1 + 3x^2)^2$ terms in $V''(1)$ are sensitive to: 
\begin{equation}
    \frac{p_0}{\sigma_0}=-\frac{1}{4}+\frac{1+3 x^2}{4 \alpha_{+} \alpha_{-}}, \quad x=\frac{r_0}{L}.
\end{equation}
That is, the $r^2/L^2$ term enhances the shell's restoring force and enlarges the mechanically stable region.} 

This can be seen in more detail as follows. In flat space the specific heat  changes sign at $r_0 = 3m$, which can be seen by asymptotically expanding at large $L$: 
\begin{equation}
    r_+(m) = 2m - \frac{8m ^3}{L^2}+\frac{96 m^5}{L^4} + O\left(\frac{1}{L^6}\right).
\end{equation}
Inserting this into the Hawking temperature $T_H(r_+) = \frac{1}{4 \pi r_{+}}\left(1+\frac{3 r_{+}^2}{L^2}\right)$ and once again expanding at large $L$ yields: 
\begin{equation}
    T_H(m) = \frac{1}{8 \pi m} + \frac{2 m}{L^2 \pi} - \frac{10 m^3}{L^4 \pi}+ O\left(\frac{1}{L^6}\right)
\end{equation}
Now we solve asymptotically at large $L$ for the location of the pole in Eq.~\eqref{eq:polelocation}, which corresponds to the small black hole branch. This gives: 
\begin{equation}
   -\frac{T_H(m, L)}{d T_H / d m} = m + \frac{32 m^3}{L^2}  + \frac{192 m^5}{L^4} + O\left( \frac{1}{L^6}\right).
\end{equation}
Finally, solving for the location of the pole and expanding again in large $L$ gives: 
\begin{equation}
    \left.r_0\right|_{C_{r_0}^{-1}=0} = 3m + \frac{5 m^3}{L^2} + \frac{57 m^5}{L^4} + O\left( \frac{1}{L^6}\right).
\end{equation}
In conclusion then, in A(dS) the thermodynamically stable region sits outside (inside) the photon sphere. This procedure can be easily carried out to all orders. 

To summarize, the specific heat $C_{r_0}$ changes sign at $r_0$ \textit{outside} the photon sphere on the small black hole branch. As we approach the spinodal point, at $r_+ \rightarrow {L}/{\sqrt{3}}$, the location where the specific heat changes sign runs off to asymptotic infinity. Indeed, for large black holes the specific heat is positive for all locations of the cavity wall $r_0$. 

More generally, the specific heat at the photon sphere is inversely proportional to the cosmological constant in all dimensions; see Sec.~\ref{sec:SpecificHeatPS}.

\subsection{Joint Stability of Shells at the Photon Sphere of AdS Black Holes}\label{sec:mechAdS}

As reviewed in Sec.~\ref{sec:BLP}, mechanical stability of an infinitesimally light, traceless shell around a black hole in asymptotically flat spacetime is \textit{marginal}
at $r=3M$ when the squared speed of sound is luminal: $\lambda=1$. We now show that in AdS spacetimes the regions of \textit{strict} mechanical and thermodynamic stability overlap at the photon sphere. 

Assuming the AdS radius $L$ takes the same value on each side of the shell, we will see such shells are conditionally stable for $\lambda \in (\frac{1}{2}, 1)$. Defining $x \equiv r_0/L$:
\begin{equation}
    \frac{m}{r_0} = \frac{1+x^2-\alpha_{-}^2}{2} \quad \quad \frac{M}{r_0} = \frac{1+x^2-\alpha_+^2}{2}.
\end{equation}
Since $\alpha_{\pm}$ are not bounded above by 1, as in flat space, for convenience we rescale:
\begin{equation}\label{eq:alphatilde}
    \alpha_{\pm} \equiv \tilde\alpha_{\pm} \sqrt{1 + x^2}, \qquad \tilde{\alpha}^2=1-\frac{2 m_{\pm} / r_0}{1+x^2},
\end{equation}
with $m_-=m$ and $m_+=M$. Inserting this into Eq.~\eqref{eq:Vstabcon}, a calculation gives:
\begin{equation}
    V^{\prime \prime}(1)=\tfrac{3}{2}(1+4 \lambda)\left(1+x^2\right) \tilde{\alpha}_{+} \tilde{\alpha}_{-}+2 \lambda+6(1+\lambda) x^2-\tfrac{\left(1+3 x^2\right)^2}{2\left(1+x^2\right)}\left(\tfrac{1}{\tilde{\alpha}_{+}^2}+\tfrac{1}{\tilde{\alpha}_{-}^2}+\tfrac{1}{\tilde{\alpha}_{+} \tilde{\alpha}_{-}}\right) .
    \label{eq:VppAdS}
\end{equation}
Stability needs $V''(1)\ge 0$, which is marginal when $V''(1)=0$. It follows from $\tilde{\alpha}_{\pm}\le 1$ that:
\begin{equation}\label{eq:AdSeasybound}
    V^{\prime \prime}(1) \leq 12\left( \lambda - \frac{1}{2}\right)x^2 + 8 \lambda + 6\left( 1 - \frac{1}{1+x^2}\right).
\end{equation}
For large $x \equiv r_0/L$ shells are unstable unless $\lambda \ge 1/2$. In the opposite limit $x \rightarrow 0$, taking $\tilde{\alpha}_+=\tilde{\alpha}_-=\tilde{\alpha}$ and solving for $V''(1)=0$ in Eq.~\eqref{eq:Vstabcon} recovers the known result \cite{Brady:1991np}: 
\begin{equation}
    \tilde{\alpha}(\lambda)=\left(\frac{\sqrt{4 \lambda^2+36 \lambda+9}-2 \lambda}{3(4 \lambda+1)}\right)^{1 / 2}
+O(x^2).
\end{equation}
At asymptotic infinity every static shell is traceless: $\tilde{\alpha}^2_\pm=1-\frac{2 m_{ \pm} / r_0}{1+x^2} \simeq 1 - \frac{2 m_{ \pm} L^2}{r_0^3}$, so $p_0/\sigma_0 \rightarrow 1/2$ and $\tilde{\alpha}_+=\tilde{\alpha}_-$. Using Eq.~\eqref{eq:staticjunctionbh}, the density of a shell at infinity reads:
\begin{equation}
\sigma_0=\frac{M-m}{2 \pi r_0^2\left(\alpha_{-}+\alpha_{+}\right)} \simeq \frac{M-m}{4 \pi r_0^2} \frac{L}{r_0},
\end{equation}
consistent with that of a conformal fluid ($\sigma_0 \propto r_0^{-3}$). In fact the AdS boundary is a photon sphere, in the sense that it satisfies $rf'/2f=1$. Using Eq.~\eqref{eq:AdSeasybound}, such shells are stable under perturbations when $c_s^2\ge 1/2$, and remain conformal for $c_s^2=1/2$.

In any case, it follows from the IJCs, see Eq.~\eqref{eq:ratioptosigma}, that: 
\begin{equation}
    \frac{p_0}{\sigma_0}=-\frac{1}{4}+\frac{1+3 x^2}{4 \alpha_{+} \alpha_{-}} = -\frac{1}{4} + \frac{(1+3x^2)}{4(1+x^2) \tilde{\alpha}_+ \tilde{\alpha}_-}.
    \label{eq:ratiotoredshift}
\end{equation}
This is enough to repeat the stability analysis of BLP in AdS; see Fig.~\ref{fig:BLPExtendedAdS}. As in flat space, the $\sigma_0 = 2 p_0$ curve intersects the photon sphere in the infinitesimally light shell limit.

\begin{figure}
    \centering
    \includegraphics[width=\linewidth]{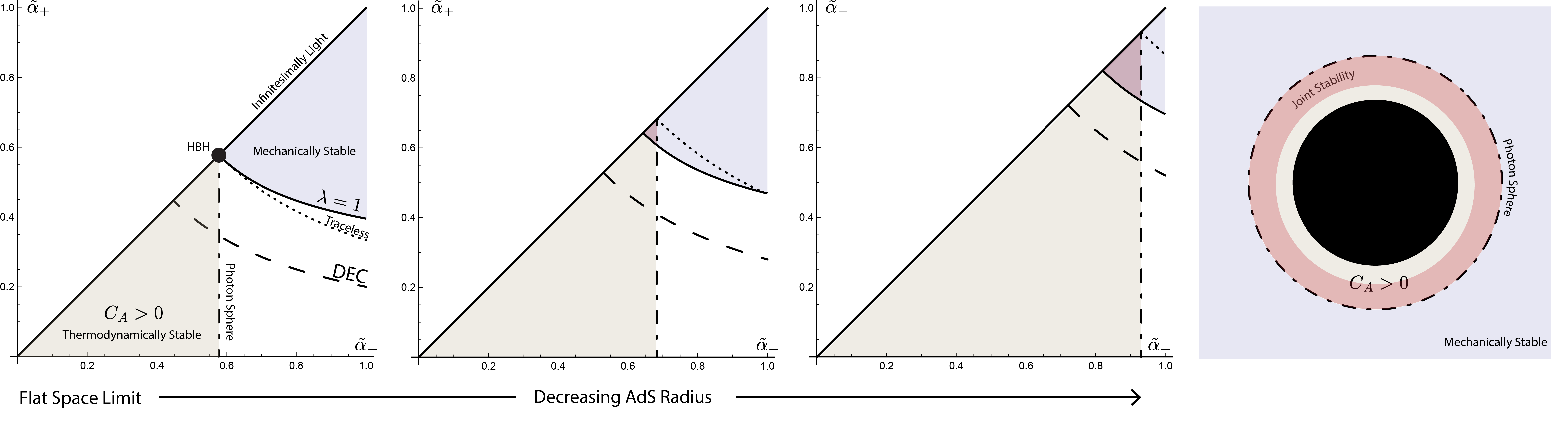}
    \caption{ Overlapping regions of mechanical and thermodynamic stability for shells around small black holes satisfying $m \ll L_{AdS}$. The thermodynamically stable region grows with $r_+$. The flat-space limit recovers Fig.~\ref{fig:BLP}. Shells are \textit{stable} at the photon sphere, where these regions now overlap. For large black holes every location outside the horizon is thermodynamically stable, as shown in Sec.~\ref{sec:SpecificHeatPS}.  The mechanically stable region grows inward, approaching $1.61 \, r_+$, in the limit $L_{AdS} \rightarrow 0$.}
    \label{fig:BLPExtendedAdS}
\end{figure}

An infinitesimally light shell at the photon sphere is mechanically stable in AdS. A straightforward computation shows that a stable light shell, where $\tilde{\alpha}\equiv \tilde{\alpha}_+=\tilde{\alpha}_-$, satisfies:
\begin{equation}\label{eq:soundspeedAdS}
    \lambda \equiv c_s^2 \ge\frac{3\left(1+3 x^2\right)^2-12 x^2\left(1+x^2\right) \tilde{\alpha}^2-3\left(1+x^2\right)^2 \tilde{\alpha}^4}{4\left(1+x^2\right)\left(1+3 x^2\right) \tilde{\alpha}^2+12\left(1+x^2\right)^2 \tilde{\alpha}^4}, \quad \tilde{\alpha}^2=1-\frac{2 m / r_0}{1+x^2} .
\end{equation}
Note the critical sound speed $\lambda_c$ depends on $x$. At the photon sphere of the black hole:
\begin{equation}\label{eq:lightshellads}
    \lambda_c= \frac{1}{2} + \frac{1}{2(1+3x^2)}
   , \quad \tilde{\alpha}^2=\frac{1+3 x^2}{3\left(1+x^2\right)}, \quad x = 3m/L,
\end{equation}
with stability for $\lambda > \lambda_c$ (marginal for $\lambda=\lambda_c$), ranging from $1/2$ (small $L$) to $1$ (large $L$). For $\lambda=1$ the shell is mechanically stable at the photon sphere: $V''(1)=54m^2/L^2$. For small $L$ at the photon sphere, Eq.~\eqref{eq:lightshellads} saturates the sound speed of a (2+1)-dimensional CFT; the pattern  $\lambda_c \equiv c_s^2 = 1/(d-1)$ holds at general $d$. Note that, unlike in flat space, sound waves may causally propagate on a shell inside the photon sphere of the black hole.

At small AdS curvature one recovers $\tilde{\alpha}^2=1/3$, consistent with \cite{Brady:1991np}. Dialing up a negative cosmological constant, the mechanically stable region extends inward from the photon sphere without quite reaching the horizon. For an infinitesimally light shell, Eq.~\eqref{eq:soundspeedAdS} at $\lambda=1$ as $x \rightarrow \infty$ gives $\tilde{\alpha}^2 = \left( \sqrt{61}-4\right)/5$. Taking the planar limit for $\tilde{\alpha}^2 \equiv f(r_0)/(1+x^2)$, which gives $1-(r_+/r_0)^3$, demonstrates instability within $[5 /(9-\sqrt{61})]^{1 / 3} r_{+} \approx 1.61 \, r_{+}$.

Finally, one can show that $V''(1)$ (Eq.~\eqref{eq:VppAdS}) increases monotonically in $r_0$ when $c_s^2 \ge 1/2$; in this case, thin shells are always stable outside some $r_0$. However, if $c_s^2 < 1/2$ there will be at most a window of stability where shells are mechanically stable. The window of stability for $c_s^2<1/2$ vanishes at infinite AdS curvature because of Eq.~\eqref{eq:AdSeasybound}.

\subsection{The Specific Heat at the Photon Sphere is Inversely Proportional to $-\Lambda$}
\label{sec:SpecificHeatPS}

Here we show the following general result, which holds in all dimensions where photon spheres exist ($n \ge 2$). The inverse specific heat $C_{r_0}^{-1}$ at the photon sphere of a black hole is proportional to (minus) the cosmological constant $\Lambda$. 

This fact is independent of the definition of $E$, provided $E$ is monotonic in $\mu$, because:
\begin{equation}
C_{r_0}^{-1}=\left.\frac{d T}{d E}\right|_{r_0}=\left.\frac{d T / d \mu}{d E / d \mu}\right|_{r_0} \propto \left. \frac{dT}{d\mu} \right|_{r_0}.
\end{equation}
To determine where $\left.\frac{d T}{d \mu}\right|_{r_0}$ changes sign, write the local temperature in terms of the surface gravity $\kappa \equiv \frac{f'(r_+)}{2}$ and calculate:
\begin{equation}
    T\left(r_0\right)=\frac{\kappa}{2 \pi \sqrt{f\left(r_0\right)}} \implies \left.\frac{d \ln T}{d \mu}\right|_{r_0}=\frac{d \ln \kappa}{d \mu}-\left.\frac{1}{2} \frac{1}{f\left(r_0\right)} \frac{\partial f\left(r_0\right)}{\partial \mu}\right|_{r_0}.
\end{equation}
For general Schwarzschild-Tangherlini black holes where $f(r) = 1 + \epsilon \frac{r^2}{L^2} - \frac{\mu}{r^{n-1}}$, the $r$ and $\mu$ derivatives are related by $r \partial_r f=-(n-1) \mu \partial_\mu f+2 \epsilon \frac{r^2}{L^2}$. At the photon sphere, this gives:
\begin{equation}
    \left.\frac{d \ln T}{d \mu}\right|_{r_{\mathrm{ps}}}=\frac{d \ln \kappa}{d \mu}+\frac{1}{(n+1) \mu f\left(r_{\mathrm{ps}}\right)}.
    \label{eq:generalsh}
\end{equation}
Note that $\mu=r_h^{n-1}\left(1+\epsilon \frac{r_h^2}{L^2}\right)$, where $r_{\mathrm{ps}}^{n-1}=\frac{n+1}{2} \mu$. This gives:
\begin{equation}
    \kappa=\frac{1}{2 r_h}\left[(n-1)+\epsilon(n+1) \frac{r_h^2}{L^2}\right] \quad \quad f_{\mathrm{ps}}=f\left(r_{\mathrm{ps}}\right)=\frac{n-1}{n+1}+\epsilon \frac{r_{\mathrm{ps}}^2}{L^2}.
\end{equation}
For later use, it is convenient to define the ratio: 
\begin{equation}
    \left(\frac{r_{\mathrm{ps}}}{r_h}\right)^{n-1}=\frac{n+1}{2}\left(1+\epsilon \frac{r_h^2}{L^2}\right) \equiv Q .
    \label{eq:photonsphererelation}
\end{equation}
Then we obtain:
\begin{equation}
    \left.\mu \frac{d \log T}{d \mu}\right|_{r_{\mathrm{ps}}}=\frac{\left(1+\epsilon \frac{r_h^2}{L^2}\right)\left[\epsilon(n+1) \frac{r_h^2}{L^2}-(n-1)\right]}{\left[(n-1)+\epsilon(n+1) \frac{r_h^2}{L^2}\right]^2}+\frac{1}{(n-1)+\epsilon(n+1) \frac{r_{\mathrm{ps}}^2}{L^2}}.
    \label{eq:specificheatpseqn}
\end{equation}
In asymptotically flat spacetime, $\epsilon=0$ and these terms cancel exactly. We now consider AdS ($\epsilon=1)$ and dS $(\epsilon=-1)$ in turn, pausing to consider how certain features are reflected in their mirror arguments. 

In AdS $(\epsilon=1)$ the first term in Eq.~\eqref{eq:specificheatpseqn} changes sign at the spinodal point, where $ \frac{r_h^2}{L^2}=\frac{n-1}{n+1}$ and the two black hole branches merge. Both terms are manifestly positive on the large black hole branch. 
On the small black hole branch, we need to show that:
\begin{equation}
    \left(\frac{(n-1)+(n+1) \frac{r_h^2}{L^2}}{(n-1)+(n+1) \frac{r_{\mathrm{ps}}^2}{L^2}}\right)>\left(1+\frac{r_h^2}{L^2}\right)\frac{(n-1)-(n+1) \frac{r_h^2}{L^2}}{(n-1)+(n+1) \frac{r_h^2}{L^2}} .
\end{equation}
For this purpose it is convenient to rewrite our expressions in terms of $Q$ using Eq.~\eqref{eq:photonsphererelation}:
\begin{equation}
    \left.\mu \frac{d \log T}{d \mu}\right|_{r_{\mathrm{ps}}}=\frac{Q(Q-n)}{(n+1)(Q-1)^2}+\frac{1}{(n-1)+(2 Q-n-1) Q^{2 /(n-1)}}.
    \label{eq:specificheatcondition}
\end{equation}
In this language the spinodal transition in AdS occurs at $Q = n$. Both terms are manifestly positive on the large black hole branch. On the small black hole branch, we need:
\begin{equation}
    n>Q>\frac{n+1}{2}, \quad  \quad (2 Q-n-1)\left[n Q-1-(n-Q) Q^{\frac{n+1}{n-1}}\right]>0 .
    \label{eq:nQeqn}
\end{equation}
The first factor in Eq.~\eqref{eq:nQeqn} is positive because $Q>(n+1)/2$ on the small black hole branch. The second factor is also positive; to see this, call the second term $H(Q)$: 
\begin{equation}
    H(Q)\equiv n Q-1-(n-Q) Q^{\frac{n+1}{n-1}}, \quad \quad H\left(\frac{n+1}{2}\right) > 0 \quad \text{for} \quad n \ge 2.
\end{equation}
To establish proportionality later, we will use the fact that $H(Q)>0$; this is shown below for both AdS and dS. Note that $H(Q)>0$ on the small black hole branch, since:
\begin{equation}
    H^{\prime}(Q)=n+Q^{\frac{2}{n-1}} \left(\frac{2n(Q-\frac{n+1}{2})}{n-1}\right)>0.
\end{equation}

In dS $(\epsilon=-1)$ the black hole horizon coincides with the cosmological horizon in the Nariai limit, which occurs when $\frac{r_h^2}{L^2}=\frac{n-1}{n+1}$ and $Q=1$. The relation has mirror interpretations in AdS and dS. In AdS this is the spinodal point, while in dS it is the Nariai limit. We are not aware of this parallel being emphasized before in this form.\footnote{Curiously, if the analytic continuation of the AdS Hawking-Page transition in Eq.~\eqref{eq:specificheatpseqn} at $r_h=L$ occurred on a physical branch, it would give $Q = 0$.} 

In dS the AdS argument is essentially mirrored with signs reversed. Black holes below the Nariai limit satisfy $1 < Q < \frac{n+1}{2}$; by similar reasoning we need to show: 
\begin{equation}
    \frac{(n-1)-(n+1) \frac{r_h^2}{L^2}}{(n-1)-(n+1) \frac{r_{\mathrm{ps}}^2}{L^2}}<\left(1-\frac{r_h^2}{L^2}\right) \frac{(n-1)+(n+1) \frac{r_h^2}{L^2}}{(n-1)-(n+1) \frac{r_h^2}{L^2}}.
\end{equation}
In other words: 
\begin{equation}
    1<Q<\frac{n+1}{2}, \quad(2 Q-n-1)\left[n Q-1-(n-Q) Q^{\frac{n+1}{n-1}}\right]<0 .
\end{equation}
The first factor is manifestly negative in dS because $1<Q<(n+1)/2$. The second factor is the same $H(Q)$ used in the AdS case. 

We now show that $H'(Q)$ has a single root at $Q=1$; we will see that $H(Q)>0$. When photon spheres exist, $n\ge 2$, and so $Q>1$ in both AdS and dS.  

It is helpful to write $H'(Q)$ as:
\begin{equation}
    H^{\prime}(Q)=n\left[1-\frac{(n+1-2 Q) Q^{2 /(n-1)}}{n-1}\right] \equiv n \left[ 1-\frac{G(Q)}{n-1}{}\right].
\end{equation}
Note that $G(1)=n-1$. Additionally:
\begin{equation}
    G^{\prime}(Q)=\frac{2(n+1)}{n-1} Q^{2 /(n-1)-1}(1-Q)<0.
\end{equation}
Then for $Q>1$ we have $G(Q)<n-1$ and $H'(Q)>0$. It follows that $H(Q)>0$ because $H(1)=0$. We conclude that $H(Q)>0$ in both AdS and dS; $Q = \frac{n+1}{2}$ is precisely the flat space limit where $C_R^{-1}=0$. 

Meanwhile, it follows from Eq.~\eqref{eq:generalsh} and Eq.~\eqref{eq:photonsphererelation} that:
\begin{equation}
    (n-1)+(2 Q-n-1) Q^{2 /(n-1)}=(n+1) f\left(r_\mathrm{ps}\right).
\end{equation}
It follows from Eq.~\eqref{eq:photonsphererelation} that $2 Q-n-1=-\frac{2 \Lambda r_h^2}{n}$. Then Eq.~\eqref{eq:specificheatcondition} can be written as:
\begin{equation}\label{eq:PSspecificheatlambda}
    \boxed{\left.\mu \frac{\partial \log T}{\partial \mu}\right|_{r_{\mathrm{ps}}}=\frac{(2 Q-n-1) H(Q)}{(n+1)(Q-1)^2 (n+1) f\left(r_{\mathrm{ps}}\right)} \propto -\Lambda.}
\end{equation}
Then the inverse specific heat $C_{r_0}^{-1}$ at the photon sphere is proportional to $-\Lambda$, as claimed.\footnote{The inverse specific heat vanishes at the photon sphere in asymptotically flat spacetime \cite{York:1986it,Andre:2021ctu}.} 

\section{Self-Gravitating Radiation and the Photon Sphere}\label{sec:ExtendedPhotonSpheres}

Here we present our main arguments. First we derive the anisotropic TOV equation, clarifying the role of the photon sphere, in Sec.~\ref{sec:anisotropictov}. In Sec.~\ref{sec:emergentisraellayers} we show that Israel layers form in the limit where the radial pressure goes to zero. We consider their thermodynamics using entropy maximization in Sec.~\ref{sec:entropyfixedpoint}. See Kim and Lee \cite{Kim:2019ygw} for previous work, which found scaling relations using a temperature independent entropy maximization method.

In Sec.~\ref{sec:hillingarcontinuumlimit}, we find an exception to the constraints reviewed in Sec.~\ref{sec:obstruction}: the HBH \cite{Riojas:2026ytt,Riojas:2026mfu}. At zero radial pressure, the self-similar fixed point in SWZ's analysis (Fig.~\ref{fig:SWZfig}) becomes the line of fixed points in Fig.~\ref{fig:IJCTOV}. The solution curve ($C$ in Fig.~\ref{fig:SWZfig}) crosses these fixed points in locations which we show are equivalent to the IJCs. If this line is crossed once for traceless matter, it gives an extended photon sphere realizing the thought experiment in Fig.~\ref{fig:thoughtexperiment}. Each marginally stable layer of the HBH ocean straddles the interface between the mechanically and thermodynamically stable regions in Fig.~\ref{fig:BLP}. 

In Sec.~\ref{sec:ascoldasablackhole}, we show the HBH is the coldest stable nested shell configuration around a black hole; the Hawking temperature is unchanged at fixed total mass. In Sec.~\ref{sec:heatcapandtemp}, we show the following general result: at fixed total mass, adding a shell around a black hole decreases the asymptotic temperature iff the local specific heat at the shell is \textit{positive}. This gives a simple interpretation of the result in Sec.~\ref{sec:ascoldasablackhole}: shells at the photon sphere leave the asymptotic temperature unchanged, because $\Lambda=0$ there (see Sec.~\ref{sec:SpecificHeatPS}), making the resulting configuration (e.g. the HBH) precisely as cold as a black hole of mass $M$. 

In Sec.~\ref{sec:coarsegrainedIJC}, an alternative approach to black hole thermodynamics is developed; instead of the Euclidean path integral, it makes direct use of shells and anisotropic thermodynamic relations. We show that an HBH in thermal equilibrium has the same coarse-grained entropy and the same temperature as a black hole of mass $M$. These independent methods further verify the key results in \cite{Riojas:2026ytt,Riojas:2026mfu}, where we used the Euclidean path integral and entropy maximization methods. 

Finally, we determine when Einstein's equations allow extended photon spheres in Sec.~\ref{sec:beyondnested}, and find thermodynamic mimicry requires a similar constraint. We show a self-similar family of solutions, including ``frozen stars" \cite{Brustein:2018web,Brustein:2021lnr,Brustein:2023hic} and ``stiffest stars" \cite{Zeldovich:1961sbr,Banks:2002fj} can satisfy $S=4\pi M^2$. Thermodynamic mimicry holds for these solutions when they have massless walls, which generically violate DEC. Walls can break optical mimicry, and more reasonable massive walls break thermodynamic mimicry. 

In contrast, the HBH does not need walls to hold its orbiting massless particles. This luminal ``Einstein cluster" \cite{Einstein:1939ms,Boehmer:2007az,Herrera:1997plx,1974RSPSA.337..529F,Cardoso:2021wlq,Maeda:2024tsg} satisfies the usual energy conditions.

\subsection{The Anisotropic Tolman-Oppenheimer-Volkoff (TOV) Equation}\label{sec:anisotropictov}

Here we review the well-known anisotropic Tolman-Oppenheimer-Volkoff equation, from a perspective that makes the structure related to the photon sphere manifest. 

Take a static and spherically symmetric metric ansatz: 

\begin{equation}
    ds^2 = -f(r) dt^2 + \frac{1}{j(r)} dr^2 + r^2 d\Omega^2,
    \label{eq:metricform}
\end{equation} 
as well as a spherically symmetric, anisotropic stress-energy tensor: 
\begin{equation}
    T_{\nu}^\mu = {\rm diag}(-\rho, P_r, P_\perp, P_\perp) \quad \quad G^\mu_\nu+\Lambda \delta^\mu_\nu=8 \pi G T^\mu_\nu \ .
\label{eq:mixedeinsteinequation}
\end{equation}
The (rr) and (tt) components of Einstein's equation are:
\begin{equation}
    \frac{f'(r)j(r)}{r  f(r)}+\frac{j(r)}{r^2 }-\frac{1}{r^2}=8 \pi  P_r -\Lambda,
    \label{eq:EErrcomponent}
\end{equation}
\begin{equation}
  -\frac{j^{\prime}(r)}{ r}+\frac{1-j(r)}{r^2}=  \frac{1}{r^2}\frac{d}{d r}\left[r\left(1-j(r)\right)\right]=8 \pi \rho + \Lambda \ .
    \label{eq:EEttcomponent}
\end{equation}
The (tt) component Eq.~\eqref{eq:EEttcomponent} can be integrated directly:
\begin{equation}
    j(r) = 1-\frac{2 \widehat{m}(r)}{r}-\frac{\Lambda r^2}{3} \quad \quad \widehat{m}^{\prime}(r)=4 \pi r^2 \rho(r).
   \label{eq:EEttinteg}
\end{equation}
Combined with  the (rr) component Eq.~\eqref{eq:EErrcomponent}, this yields:
\begin{equation}
    \frac{r f'(r)}{2 f(r)}=\frac{\widehat{m}(r)+4 \pi  r^3 P_r-\frac{1}{3}\Lambda  r^3}{r-2 \widehat{m}(r)-\frac{1}{3}\Lambda  r^3}.
    \label{eq:EErrwithj(r)}
\end{equation}
The Bianchi identity  $\nabla^\mu T_{\mu \nu}=0$ gives the condition of hydrostatic equilibrium: 
\begin{equation}
    P_r^{\prime}(r)=-\frac{f^{\prime}(r)\left(P_r(r)+\rho(r)\right)}{2 f(r)}-\frac{2\left(P_r(r)-P_{\perp}(r)\right)}{r} .
    \label{eq:hydrostaticequilibrium}
\end{equation}
Using Eq.~\eqref{eq:EErrwithj(r)} and Eq.~\eqref{eq:hydrostaticequilibrium}, the anisotropic TOV equation takes the form:
\begin{equation}
    P_r^{\prime}(r)=-\left(P_r(r)+\rho(r)\right) \frac{\widehat{m}(r)+r^3\left(4 \pi P_r(r)- \frac{1}{3}\Lambda\right)}{r(r-2 \widehat{m}(r)-\frac{1}{3}\Lambda r^3)}-\frac{2\left(P_r(r)-P_{\perp}(r)\right)}{r}.
    \label{eq:anisotropic_TOV}
\end{equation}
It is worth noting that the photon sphere condition $rf'/2f=1$ appears in these equations. The fact that Eq.~\eqref{eq:anisotropic_TOV} becomes simple when $P_r'=P_r=0$ will be useful for our purposes.

\subsection{Shells of Anisotropic Self-Gravitating Radiation}\label{sec:emergentisraellayers}

Sorkin, Wald, and Zhang (SWZ) argued \cite{Sorkin:1981wd} that the TOV equation follows from an entropy maximization principle, by varying the action: 
\begin{equation}
    I=\int_0^{r_0}\left(\frac{1}{r^2} \frac{d m}{d r}\right)^{3 / 4}\left[1-\frac{2 m(r)}{r}\right]^{-1 / 2} r^2 d r.
\end{equation}
This can be straightforwardly generalized to anisotropic gases. The following argument, which is adapted from companion work \cite{Riojas:2026mfu}, begins by writing:
\begin{equation}
    I=\int s(\rho) \sqrt{h} d^3 x=4 \pi \int s(\rho)\left(1-\frac{2 m(r)}{r}\right)^{-1 / 2} r^2 d r.
\end{equation}
Varying the action and integrating by parts, with the usual assumption that $\delta m=0$ at the edges of the interval, and defining the temperature $(\partial s/\partial\rho)|_r = 1/T$, a straightforward calculation gives:
\begin{equation}
    \frac{T^{\prime}}{T}=\frac{r m'-m-4 \pi r^3(s T)}{r(r-2 m)}.
\end{equation}
Assuming that the anisotropic TOV equation holds, it \textit{follows} that:
\begin{equation}\label{eq:anisothermo}
    s T=\rho+P_r.
\end{equation}
Taking the radial derivative of $s T = P_r + \rho$, we find
\begin{equation}
    \frac{d P_r}{dr} + \left[ \frac{d \rho}{dr} - T \frac{ds}{dr}\right] = s \frac{dT}{dr}.
\end{equation}
This resembles the condition for energy conservation: 
\begin{equation}
    \frac{dP_r}{dr} + \left[\frac{2(P_r - P_\perp)}{r}\right] =-\frac{f^{\prime}(r)\left(P_r(r)+\rho(r)\right)}{2 f(r)}=s \frac{d T}{d r} .
\end{equation}
Combining these two equations and using
\begin{equation}
    \frac{ds}{dr}=\left(\frac{\partial s}{\partial \rho}\right)_r \frac{d \rho}{dr} + \left( \frac{\partial s}{\partial r}\right)_\rho = \frac{1}{T} \frac{d \rho}{dr} + \left( \frac{\partial s}{\partial r}\right)_\rho,
\end{equation}
gives the following thermodynamic equation:
\begin{equation}
    \frac{d \rho}{d r}-T \frac{d s}{d r} = -T\left( \frac{\partial s}{\partial r}\right)_\rho= \frac{2\left(P_{r}-P_\perp\right)}{r}\ .
    \label{eq:masterformulascaling}
\end{equation}
Note this implies that $s$ depends explicitly on $r$ only in anisotropic gases.  Also, since $(\partial \rho/\partial s)_r=T$, we obtain the differential relation: 
\begin{equation}\label{eq:firstlawaniso}
    d \rho = T ds - (P_\perp - P_r) d \log A, \quad \quad A \equiv 4 \pi r^2.
\end{equation}

These relations hold for spherically symmetric $P_r(r)$ and $P_\perp(r)$. The relations found by Kim and Lee \cite{Kim:2019ygw} in the special case $P_r=w_r \rho$, $P_\perp = w_\perp \rho$ where $w_r,w_\perp$ are constants, can be recovered from this approach. For additional discussion, see companion work \cite{Riojas:2026mfu}.

\subsubsection{TOV $=$ IJC at Zero Radial Pressure}

The anisotropic TOV equation reduces to the IJCs in the zero radial pressure limit.\footnote{We expect it is more generally true that setting one pressure component to zero gives an IJC for shells oriented along that normal direction. This will be revisited in an upcoming article.} We apply the SWZ analysis of the TOV equation to anisotropic radiation. 

We take a static, spherically symmetric metric:
\begin{equation}
    d s^2=-f(r) d t^2+\frac{d r^2}{1-2 m(r) / r}+r^2 d \Omega^2
\end{equation}
As in SWZ \cite{Sorkin:1981wd}, we take the definitions:
\begin{equation}
    \mu(r) \equiv \frac{m(r)}{r}, \quad q(r) \equiv 4 \pi r^2 \rho(r), \quad z=\ln r .
\end{equation}
Here $m(r)$ satisfies: 
\begin{equation}
    \frac{d m}{d r}=4 \pi r^2 \rho=q(r) \implies \frac{d \mu}{d z}=q-\mu.
\end{equation}
For anisotropic null gas, this reads:
\begin{equation}
    P_r=w_r \rho, \quad P_\perp=w_\perp \rho, \quad w_r+2 w_\perp=1.
\end{equation}
In these variables, the TOV equation in Eq.~\eqref{eq:anisotropic_TOV} becomes: 
\begin{equation}
    q^{\prime}(z)=\frac{-\left(1+w_r\right) q^2+q\left(\frac{2 w_\perp}{w_r}-\frac{1+4 w_\perp+w_r}{w_r} \mu\right)}{1-2 \mu}.
\end{equation}
This reduces to the problem studied by SWZ when the gas is isotropic, i.e. $w_r = w_\perp = 1/3$. For the null gas this simplifies to: 
\begin{equation}
    q^{\prime}(z)=-\frac{q\left[\mu\left(3-w_r\right)+w_r\left(1+w_r\right) q+\left(w_r-1\right)\right]}{w_r(1-2 \mu)}  
    \label{eq:anisotropicSWZTOV}
\end{equation}

To understand the behavior of these solutions it is useful to consider the curves where $q(\mu)$ becomes vertical and horizontal: the vertical and horizontal nullclines. When $q'(z)=0$ the trajectory $q(\mu)$ is horizontal; when $q(\mu) =\mu$ we have $\mu'(z)=0$ and the curve is vertical. Calling these curves $q_V$ and $q_H$: 
\begin{equation}
    q_H(\mu) \equiv \frac{2 w_\perp-\left(1+4 w_\perp+w_r\right) \mu}{\left(1+w_r\right) w_r}  \quad \quad q_V(\mu)\equiv\mu.
\end{equation}
These solutions $q(\mu)$ spiral around the fixed point, which occurs where they cross at:
\begin{equation}
    (\mu_*, q_*) \equiv \left(\frac{2 w_{\perp}}{\left(1+w_r\right)^2+4 w_{\perp}} , \frac{2 w_{\perp}}{\left(1+w_r\right)^2+4 w_{\perp}} \right).
\end{equation}
At large $r$, the mass function becomes linear, with $m(r) \sim \left(\frac{2 w_{\perp}}{\left(1+w_r\right)^2+4 w_{\perp}}\right)r$, as the curve in the $(\mu,q)$ plane approaches the fixed point at $(\mu_*,q_*)$.

As the radial pressure is decreased, solutions spiral violently as in Fig.~\ref{fig:spiraling} and Fig.~\ref{fig:spontaneouslayers}. In the limit where $w_r \rightarrow 0$, the horizontal nullcline $q_H$ becomes vertical and centered on:
\begin{equation}
    \mu = \frac{2 w_\perp}{1 + 4 w_\perp}.
\end{equation}

\begin{figure}[p!]
    \centering
    \includegraphics[width=\linewidth]{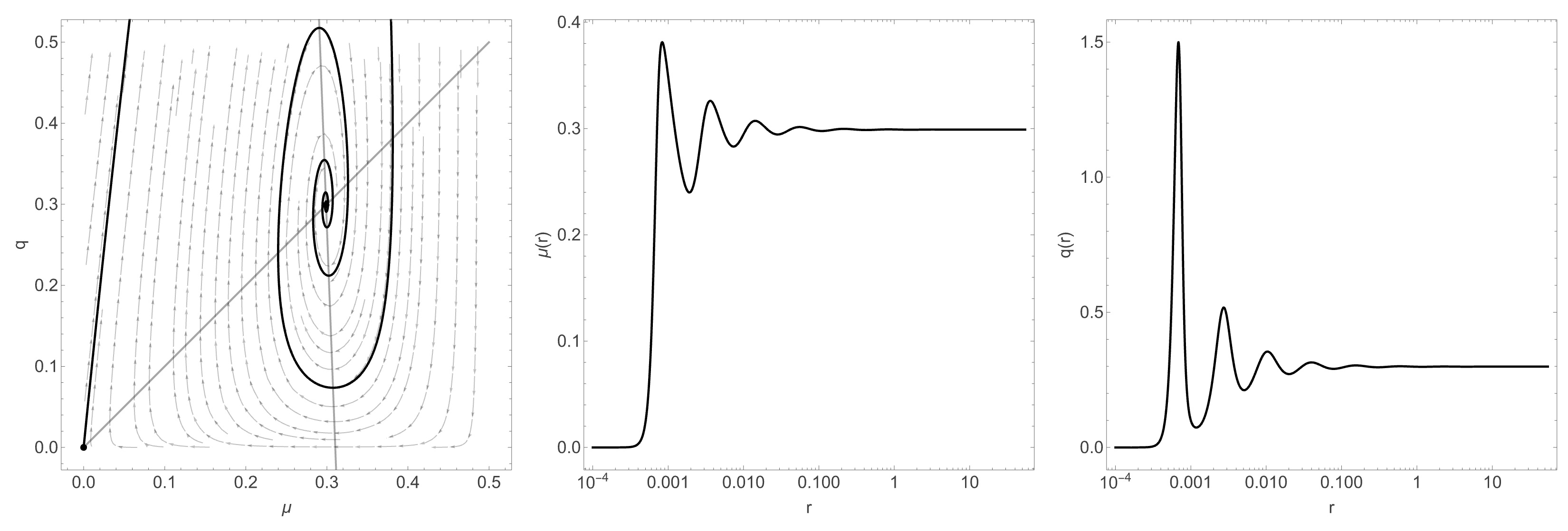}
    \caption{Self-gravitating gas solutions spiral around the fixed point as $w_r\rightarrow 0$. These are unstable beyond the first turning point in $\mu(r)$ \cite{Sorkin:1981jc,Sorkin:1981wd}. The existence of an infinite number of discrete non-singular layers in this limit has been noted by Andréasson \cite{Andreasson:2021lsh}.}
    \label{fig:spiraling}
\end{figure}
\begin{figure}[p!]
    \centering
    \includegraphics[width=1\linewidth]{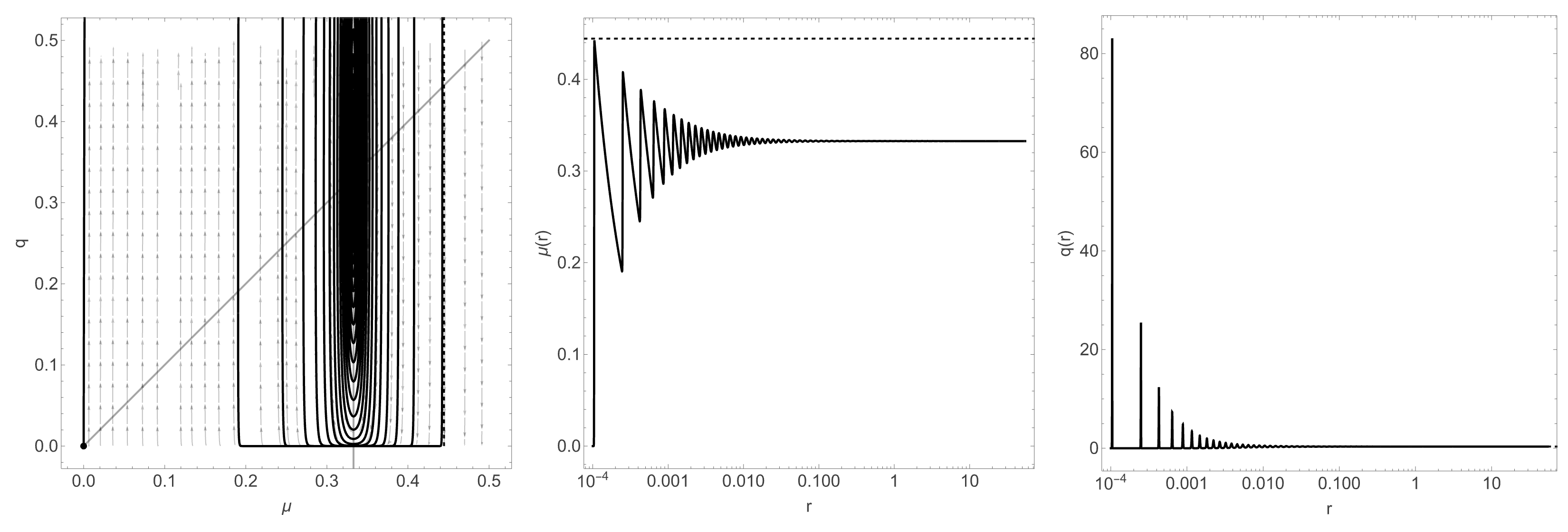}
    \caption{The self-gravitating anisotropic radiation spontaneously forms thin Israel layers as $w_r\rightarrow 0$. The streamlines that become vertical in this limit spiral around $\mu=1/3$. The parameter $\mu \equiv \frac{M}{R}$ saturates the Buchdahl-Andréasson bound \cite{Buchdahl:1959zz,Andreasson:2007ck}, pictured as a dotted line, in the first layer. The isolated fixed point at $(\mu,q) = \left(\frac{1}{3},\frac{1}{3}\right)$ is a line of fixed points at $w_r=0$, see Fig.~\ref{fig:Qmusystem}.}
    \label{fig:spontaneouslayers}
\end{figure}
\begin{figure}[p!]
    \centering
    \includegraphics[width=\linewidth]{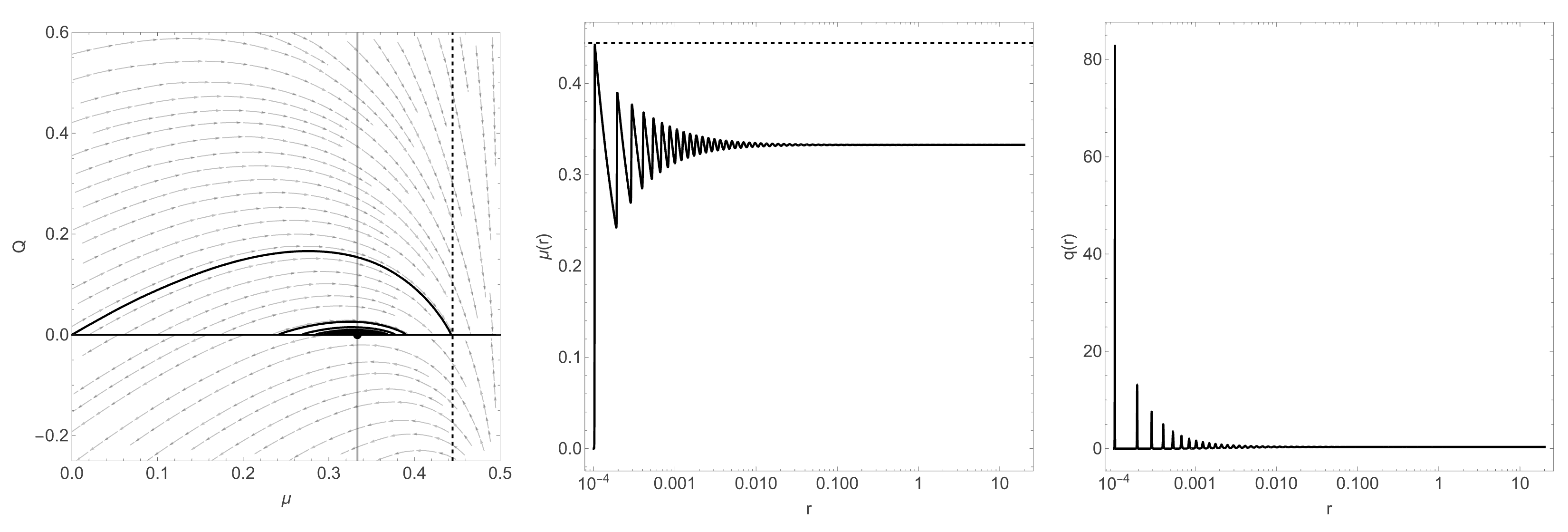}
        \caption{The solution in $(Q,\mu)$ coordinates. In the $\epsilon \rightarrow 0$ limit the isolated fixed point becomes a line of fixed points. The endpoints where the curve intersects the line of fixed points recover the IJCs. Consistent with mechanical stability, shells appear where $\mu\le 1/3$; $\mu\equiv M/R$ then jumps temporarily to $\mu \ge 1/3$, giving alternating stable and unstable light rings. }
    \label{fig:Qmusystem}
\end{figure}

To parse the degenerate limit $w_r \rightarrow 0$, it is convenient to introduce the variables: 
\begin{equation}
    \varepsilon \equiv w_r, \quad Q\equiv\varepsilon q, \quad \xi\equiv\frac{z}{\varepsilon} .
\end{equation}
The limit where $\varepsilon \rightarrow 0$ gives:
\begin{equation}
    \frac{d \mu}{d \xi}=Q, \quad \frac{d Q}{d \xi}=-\frac{Q(3 \mu+Q-1)}{1-2 \mu},  \quad  \frac{d Q}{d \mu}=\frac{(d Q / d \xi)}{(d \mu / d \xi)}=\frac{1-3 \mu-Q}{1-2 \mu} .
\end{equation}

In these $(Q,\mu)$ coordinates, the degenerate behavior observed at small $w_r$ becomes well-behaved; see Fig.~\ref{fig:Qmusystem}. This is solved by the family of solutions: 
\begin{equation}
    Q(\mu)=3 \mu-2+C \sqrt{1-2 \mu} .
\end{equation}
To give one example, consider the case that saturates the Buchdahl-Andréasson bound \cite{Buchdahl:1959zz,Andreasson:2007ck}. It is a traceless shell enclosing vacuum; with regularity at the origin $(\mu, q)=(0,0)$, we find $C=2$:
\begin{equation}
    Q(\mu)=3 \mu-2+2 \sqrt{1-2 \mu}.
\end{equation}
The curve for that shell ends when $Q(\mu_*)=0$,  which occurs precisely at the Buchdahl radius $\mu = 4/9$. This traceless matter saturates the bound obtained by Andréasson \cite{Andreasson:2007ck}.

\begin{figure}
    \centering
    \includegraphics[width=.8\linewidth]{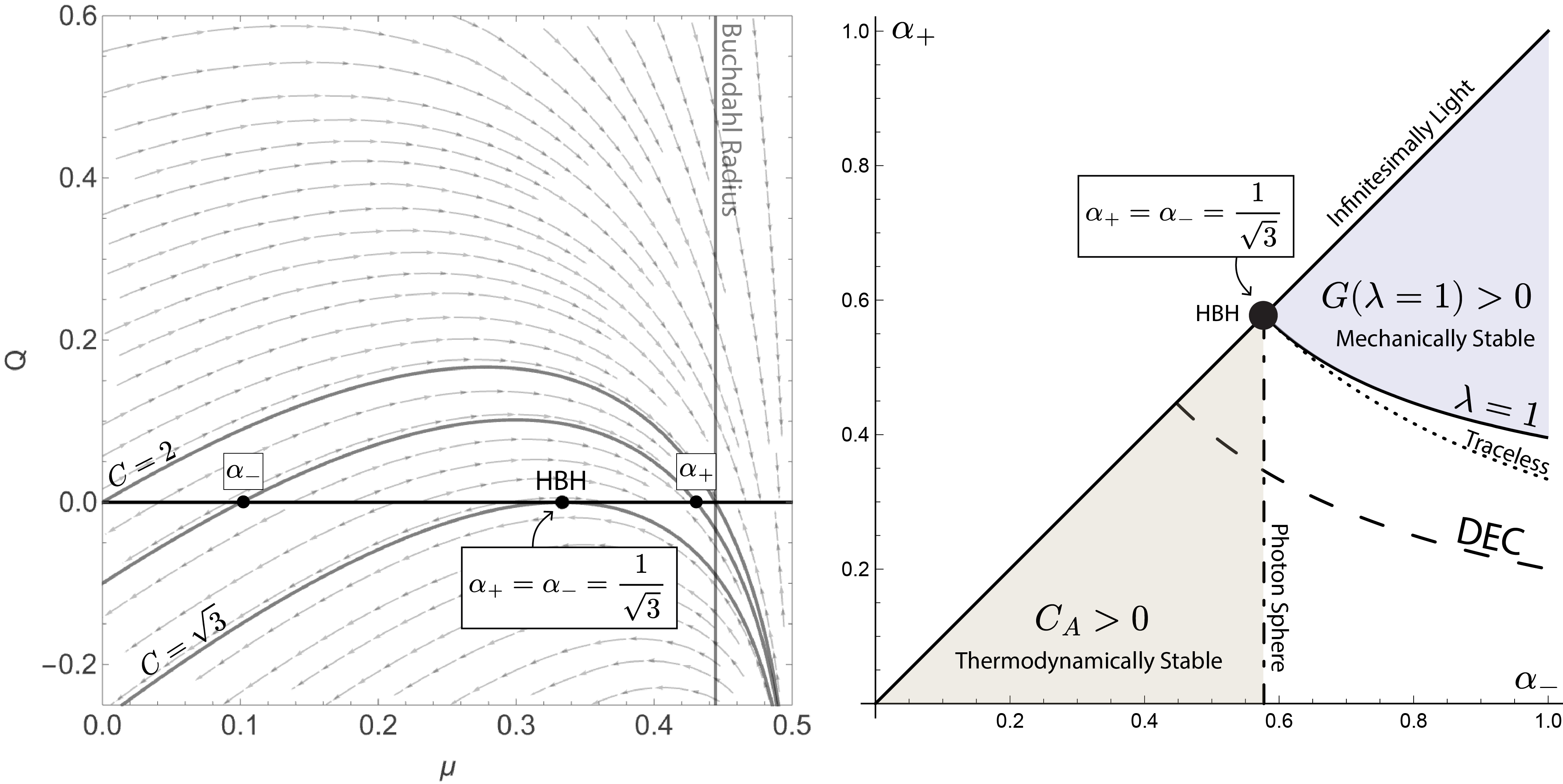}
    \caption{The IJCs are solved using TOV in the $w_r \rightarrow 0$ limit (left). $Q(\mu)$ crosses the line of fixed points, yielding $\alpha_{\pm}=\sqrt{1-2 \mu}$. For $w_\perp = \frac{1}{2}$, equivalently $\sigma = 2p$, the HBH touches this line once for $\alpha_\pm=1/\sqrt{3}$, matching BLP \cite{Brady:1991np} (right). The Buchdahl-Andréasson bound satisfies $C=2$. As discussed below, the diagrams for other $w_\perp$ are analogous. Self-similar Einstein clusters satisfy $\alpha_+=\alpha_-$. The HBH is a marginally stable, luminal Einstein cluster around a black hole.  }
    \label{fig:IJCTOV}
\end{figure}

For the other cases, we require that two layers are joined smoothly at the boundary between layers where $Q=0$. The same constant $C$ applies at both ends: 
\begin{equation}
    C=\frac{2-3 \mu_{-}}{\sqrt{1-2 \mu_{-}}}=\frac{2-3 \mu_{+}}{\sqrt{1-2 \mu_{+}}} .
\end{equation}
Suppose the total masses enclosed within neighboring layers are, respectively, $\mu_{-}=m/R$ and $\mu_{+}=M/R$. Solving for $R$ gives: 
\begin{equation}
   R=\frac{3}{8}\left(3 m+3 M \pm \sqrt{9 m^2-14 m M+9 M^2}\right).
\end{equation}
This matches our Israel junction condition for a shell around a black hole; note that for $m=0$ we recover the Buchdahl radius $\mu = 4/9$. 

We have just considered Eq.~\eqref{eq:anisotropicSWZTOV} for null gas. For more general, not necessarily null, gravitating radiation an identical procedure yields: 
\begin{equation}
    Q(\mu)=\left(1+4 w_{\perp}\right) \mu-\left(1+2 w_{\perp}\right)+C \sqrt{1-2 \mu} .
    \label{eq:pressurelessSWZTOVsol}
\end{equation}
For the first layer, $Q(0)=0$ fixes $C=1+2 w_{\perp}$, so: 
\begin{equation}
    Q(\mu)=\left(1+4 w_{\perp}\right) \mu+\left(1+2 w_{\perp}\right)\left(\sqrt{1-2\mu}-1\right).
\end{equation}
The first layer ends where $Q(\mu_*)=0$, which occurs at the radius: 
\begin{equation}
    \mu_*=\frac{4 w_{\perp}\left(1+2 w_{\perp}\right)}{\left(1+4 w_{\perp}\right)^2} .
\end{equation}
The Buchdahl radius is recovered in the limit $w_\perp \rightarrow 1/2$. The first layer approaches its local horizon $\mu = 1/2$ in the limit where $w_\perp \rightarrow \infty$.

To determine the other layers, we need only require that two layers are joined smoothly where $Q=0$. This gives the matching condition: 
\begin{equation}
    \frac{\left(1+2 w_{\perp}\right)-\left(1+4 w_{\perp}\right) \mu_{-}}{\sqrt{1-2 \mu_{-}}}=\frac{\left(1+2 w_{\perp}\right)-\left(1+4 w_{\perp}\right) \mu_{+}}{\sqrt{1-2 \mu_{+}}} .
\end{equation}
The mass contained within these layers satisfies, respectively, $\mu_- = m/R$ and $\mu_+ = M/R$. This gives a nesting rule for successive layers:
\begin{equation}
    R=\frac{\left(1+4 w_{\perp}\right)\left(\left(1+4 w_{\perp}\right)(m+M) \pm \sqrt{\left(1+4 w_{\perp}\right)^2(M-m)^2+4 m M}\right)}{8 w_{\perp}\left(1+2 w_{\perp}\right)} .
\end{equation}
The IJCs in Sec.~\ref{sec:IJC} are recovered directly in the $w_r \rightarrow 0$ limit, equivalently $Q=0$. Returning to our earlier definition in Eq.~\eqref{eq:alpha_def}, we have $\alpha\equiv \sqrt{1 - 2 \mu}$. Then at the boundary layer, where $Q=0$, Eq.~\eqref{eq:pressurelessSWZTOVsol} yields: 
\begin{equation}
    C(\alpha)=\frac{1+\left(1+4 w_{\perp}\right) \alpha^2}{2 \alpha} .
\end{equation}
Requiring smoothness at the boundary layers gives: 
\begin{equation}
    \frac{1+\left(1+4 w_{\perp}\right) \alpha_{-}^2}{2 \alpha_{-}}=\frac{1+\left(1+4 w_{\perp}\right) \alpha_{+}^2}{2 \alpha_{+}} .
\end{equation}
Solving for $w_\perp$ gives: 
\begin{equation}
    \frac{p}{\sigma} \equiv w_{\perp}=\frac{1}{4}\left(\frac{1}{\alpha_{+} \alpha_{-}}-1\right),
    \label{eq:TOVIJCformula}
\end{equation}
a perfect match for the Israel junction condition in Eq.~\eqref{eq:p_over_sigma_alpha}. The surface density is fixed by the enclosed mass jump, so this ratio suffices to obtain both Israel junction conditions. 

The physical meaning of $C$ is as follows. For a given $w_\perp$, the two boundary values $\alpha_{\pm}$ are the two roots of Eq.~\eqref{eq:pressurelessSWZTOVsol}: 
\begin{equation}
    \alpha_{+} \alpha_{-}=\frac{1}{1+4 w_{\perp}}, \quad \alpha_{+}+\alpha_{-}=\frac{2 C}{1+4 w_{\perp}} .
    \label{eq:IJCfromTOV}
\end{equation}
Using Eq.~\eqref{eq:alpha_def}, $\alpha_{ \pm} \equiv \sqrt{f_{ \pm}(R)}$, so $C$ is the average of the redshift factors on each side: 
\begin{equation}
    C=\frac{1}{2}\left(\frac{1}{\alpha_{+}}+\frac{1}{\alpha_{-}}\right) .
    \label{eq:Cdef}
\end{equation}
If the region inside the Israel layer is empty, regularity at the origin gives the maximum value of $C$, which occurs at $\alpha = {1}/{\sqrt{1+4 w_{\perp}}}$:
\begin{equation}
    1 \le \sqrt{1+4 w_{\perp}} \le C \le 1 + 2 w_\perp \le 2.
    \label{eq:genC}
\end{equation}
The roots $\alpha_{\pm}$ are the inner and outer edges of a vacuum region between matter layers. When $C = \sqrt{1 + 4 w_\perp}$, the curve is tangent to the line of fixed points, making the layers continuous. For $w_\perp = 1/2$, Eq.~\eqref{eq:genC} recovers the values of $\sqrt{3}\le C \le 2$ for the null case; $C=2$ recovers the Buchdahl-Andréasson bound \cite{Buchdahl:1959zz,Andreasson:2007ck}, while $C = \sqrt{3}$ places each layer at its local photon sphere; $w_\perp=0$ gives $C=1$. In particular, $w_\perp=1/2$ gives an ultra-compact Einstein cluster, where massless radiation orbits at the speed of light within its local photon sphere. We will see in Sec.~\ref{sec:hillingarcontinuumlimit} that this is the HBH.

The IJC can be determined from Fig.~\ref{fig:IJCTOV} as follows. First choose some $C$ and find where the solution for $Q$ crosses the line of fixed points. Their intersections give values of $\mu$ that determine the IJC through $\alpha_{\pm}=\sqrt{1-2 \mu}$. Alternatively, one can use Eq.~\eqref{eq:TOVIJCformula}.

This section established three results in the zero radial pressure limit: (1) the IJCs are given precisely by the points where the exact solutions to the TOV equation Eq.~\eqref{eq:pressurelessSWZTOVsol} cross the line of critical points in Fig.~\ref{fig:IJCTOV}, (2) the photon sphere sits at $\alpha_{\pm}=1/\sqrt{3}$, dividing stable and unstable fixed points; for the HBH, the $Q(\mu)$ curve is tangent to the critical line, and is the same point where mechanical and thermodynamical stability coincide in Fig.~\ref{fig:BLP}, and (3) the Buchdahl-Andréasson bound \cite{Buchdahl:1959zz,Andreasson:2007ck} is the $C = 2$ endpoint of this same family of solutions. Similar comments apply to the self-similar family of solutions in Sec.~\ref{sec:beyondnested}.

\subsubsection{The Buchdahl-Andréasson bound}
\label{sec:Buchdahlbound}

Here we comment briefly on the Buchdahl-Andréasson bound \cite{Buchdahl:1959zz,Andreasson:2007ck}. This gives an alternative perspective on the related discussion that appears in the companion papers \cite{Riojas:2026ytt,Riojas:2026mfu}. 

The saturation occurs for the solution that is \textit{luminal} at the fixed point in the continuum limit. To see this, let us consider the velocity of these orbits more generally. The local circular orbit velocity as measured by a static observer at radius $r$ is given by: 
\begin{equation}
    \Omega^2 \equiv\left(\frac{d \phi}{d t}\right)^2=\frac{f^{\prime}(r)}{2 r} \implies v_{\mathrm{loc}}^2(r)=\frac{r^2 \Omega^2}{f(r)}=\frac{r f^{\prime}(r)}{2 f(r)}=\frac{m(r)+4 \pi r^3 P_r(r)}{r-2 m(r)} .
        \label{eq:shellspeed}
\end{equation}
At the fixed point we have:
\begin{equation}
    v_{\text{loc}}^2=\frac{\nu\left(1+w_r\right)}{1-2 \nu} .
\end{equation}
For $w_r=0$ and $w_\perp=1/2$ as is appropriate for orbiting massless radiation, we have $\nu=1/3$ and so the motion is luminal: $v_{\text{loc}}=1$. At the fixed point, we have photons orbiting their local photon sphere. This condition will be of interest later in Sec.~\ref{sec:beyondnested}.

Away from the fixed point the motion on shells can be acausal; this is not an issue for the HBH. Note that, in terms of the $\mu$ and $Q$ variables:
\begin{equation}
    v_{\text{loc}}^2(r)=\frac{\mu(r)+w_r q(r)}{1-2 \mu(r)}=\frac{\mu+Q}{1-2 \mu}.
\end{equation}
For $Q=0$ the orbits of the constituents of the shell are superluminal for $\mu>1/3$. For the shell enclosing vacuum, where $Q(\mu)=3 \mu-2+2 \sqrt{1-2 \mu}$, this becomes: 
\begin{equation}
    v_{\mathrm{loc}}^2=2\left(\frac{1}{\sqrt{1-2 \mu}}-1\right).
\end{equation}
For traceless radiation a single shell enclosing vacuum saturates the Buchdahl-Andréasson bound $\mu = 4/9$ \cite{Buchdahl:1959zz,Andreasson:2007ck}; $v_{\text{loc}}=2$, which is superluminal. More generally, the velocity of the orbiting geodesics oscillates around a constant value,\footnote{The velocity at the fixed point satisfies $v_{\text{loc}}^2=\frac{r f^{\prime}(r)}{2 f(r)}=\frac{\hat{m}(r)+4 \pi r^3 P_r}{r-2 \hat{m}(r)}=\frac{\nu\left(1+w_r\right)}{1-2 \nu} \equiv \frac{\delta}{2}$, as shown in \cite{Riojas:2026ytt}.} saturating in the limit. The fixed point for traceless radiation saturates at ($\nu=1/3, w_r=0)$, which gives $v_{\mathrm{loc}}=1$. We will see shortly that this gives the HBH: a marginally stable, luminal Einstein cluster.

\subsubsection{Entropy Density at the Fixed Point}\label{sec:entropyfixedpoint}

The solution at the fixed point where $q_H$ and $q_V$ intersect is self-similar. This motivates studying the thermodynamics of self-similar fluids in more detail. For these solutions we now repeat the discussion that appears in \cite{Riojas:2026mfu} concerning the scaling behavior $S \propto M^p$,  for some power $p$.

In the previous section we saw that for anisotropic gas, $\rho = T s - P_r$. For $P_r = w_r \rho$ and $P_\perp = w_\perp \rho$:
\begin{equation}
    s(r) = \frac{\rho(r)(1 + w_r)}{T(r)}.
\end{equation}
For the self-similar fixed-point solutions this gives: 
\begin{equation}
    \rho(r) = \frac{q(r)}{4 \pi r^2} = \frac{\nu}{4 \pi r^2}, \quad  m(r) = \nu r, \quad  \nu=\frac{2 w_{\perp}}{1+4 w_{\perp}}.
    \label{eq:selfsimilarrelations}
\end{equation}
For $\Lambda=0$, Einstein's equations Eq.~\eqref{eq:EErrcomponent} require:
\begin{equation}
    \frac{r f^{\prime}(r)}{2 f(r)}=\frac{\hat{m}(r)+4 \pi r^3 P_r}{r-2 \hat{m}(r)}=\frac{\nu\left(1+w_r\right)}{1-2 \nu} \equiv \frac{\delta}{2},
    \label{eq:selfsimilardelta}
\end{equation}
where we specialized to the fixed point\footnote{The Stefan-Boltzmann law holds only for \textit{isotropic} radiation. These calculations reproduce the usual expectations: hydrostatic equilibrium Eq.~\eqref{eq:anisotropic_TOV} gives $f(r) \propto \rho(r)^{-\frac{2 w}{w+1}}$, and $T\propto f^{-1/2}$ gives $\rho \propto T^{(1+w) / w}$. For \textit{anisotropic} radiation Eq.~\eqref{eq:firstlawaniso} is not enough to fix $\rho(f)$; the fixed point gives $\rho \propto r^{-2}$, and $f \propto r^\delta$, so $T \propto f^{-1 / 2} \propto \rho^{\delta / 4}$. The stiffest isotropic star \cite{Zeldovich:1961sbr,Banks:2002fj}, $w=1$, also gives $s(r) \propto 1/r$ and $S \propto M^2$. } in the last step. This gives the $g_{tt}$ component $f(r)$, and the local temperature $T(r)$, up to a coefficient:
\begin{equation}
    f(r) \propto r^\delta \implies T(r) = \frac{T_H}{\sqrt{f(r)}} \propto r^{-\delta/2}.
\end{equation}
At the fixed point we have $g_{rr}=\frac{1}{1-2\nu}$, which is constant. This gives the power-law scaling: 
\begin{figure}
    \centering
    \includegraphics[width=0.75\linewidth]{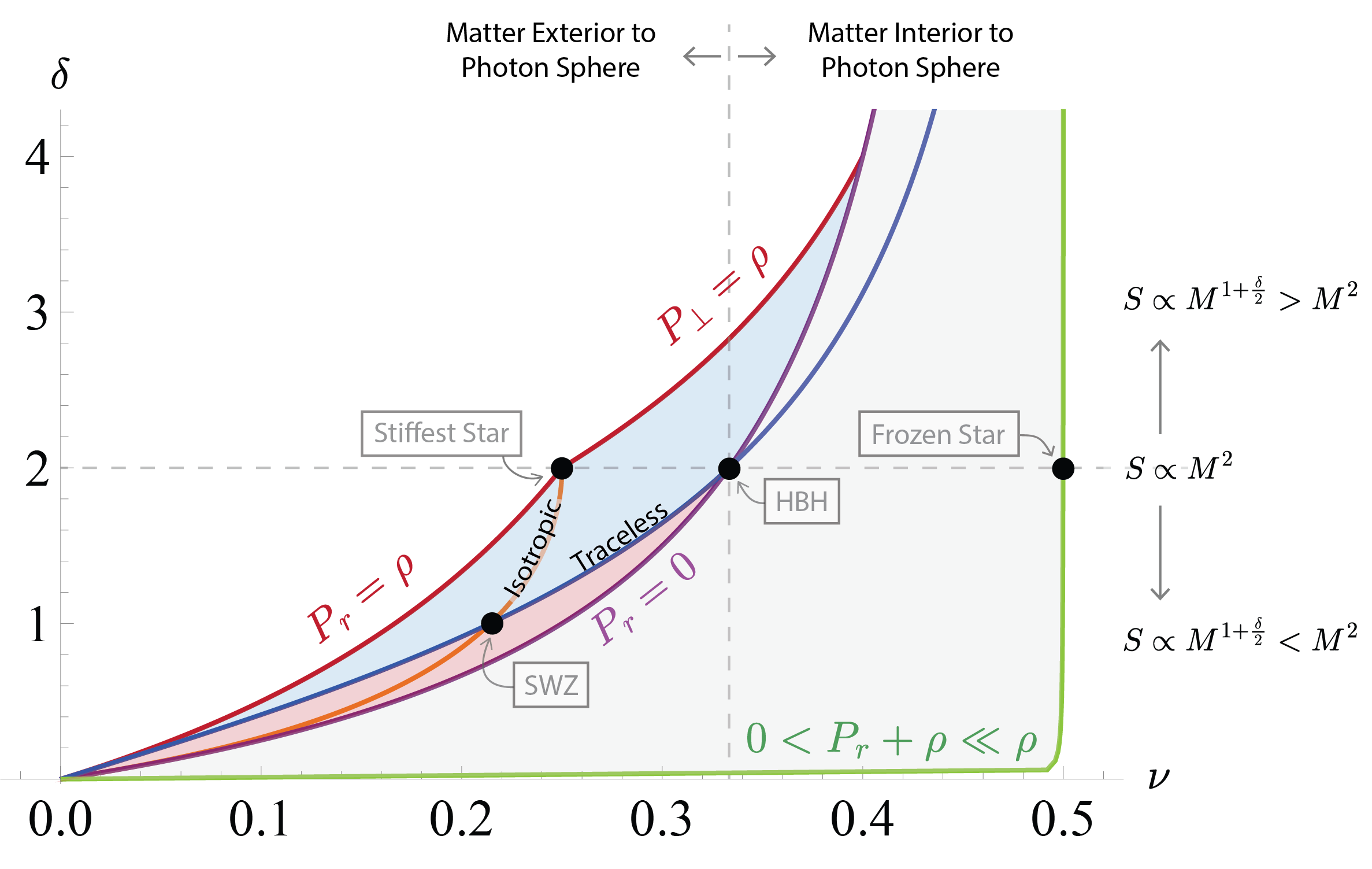}
    \caption{From \cite{Riojas:2026ytt,Riojas:2026mfu}, the entropy density of self-similar solutions in the $(\delta, \nu)$ plane. In \cite{Riojas:2026mfu} it was shown that the HBH satisfies $T=(8 \pi M)^{-1}$ and $S = 4 \pi M^2$. In Sec.~\ref{sec:beyondnested}, we show that for $\sigma=0$ this can be conditionally extended to the entire $\delta=2$ family, including ``frozen stars" \cite{Brustein:2018web,Brustein:2021lnr,Brustein:2023hic} (see also \cite{Sorkin:1981jc}) and ``stiffest stars" \cite{Banks:2002fj}. The conditions $\delta=2$ and $\nu=1/3$ occur together for the HBH, which is \textit{aligned} as discussed in Sec.~\ref{sec:HBHintro}. In particular, since $\sigma=2p$ as in the thought experiment (Fig.~\ref{fig:thoughtexperiment}), no IJC is needed in Eq.~\eqref{eq:p_over_sigma}. Other cases require an IJC and are \textit{misaligned}.}
    \label{fig:classesofmatter}
\end{figure}

\begin{equation}
    s(r) \propto r^{\frac{\delta}{2}-2} \implies S \equiv \int \sqrt{g_{rr}}(4 \pi r^2) s(r) dr \propto M^{1+\frac{\delta}{2}}.
\end{equation}
For the entropy to grow as $S \propto M^2$, we need $\delta=2$. This gives the photon sphere condition:
\begin{equation}
    \frac{r f^{\prime}(r)}{2 f(r)}=1, \quad \quad \nu = \frac{1}{3} \implies \rho = \frac{1}{12 \pi r^2},  \quad m(r) = \frac{r}{3}.
\end{equation}

In the limit where $w_r \rightarrow 0$, where the radiation orbits the black hole and discrete Israel layers emerge, we obtain the solution discussed in the previous section where $\alpha_{\pm}=1/\sqrt{3}$. The layers form a continuum that extends the photon sphere of the black hole outward. 

The anisotropic generalization of SWZ \cite{Sorkin:1981wd} yields a single solution that is marginally stable for both York \cite{York:1986it} and BLP \cite{Brady:1991np}. The HBH does not require walls, because $P_r=0$; it has an extended photon sphere, because $\delta=2$; its particles orbit within, because $T^\mu {}_\mu=-\sigma+2p=0$. It lies at the intersection of these three curves in the $(\delta, \nu)$ plane; see Fig.~\ref{fig:classesofmatter}, from \cite{Riojas:2026ytt,Riojas:2026mfu}, where the entropy density of self-similar solutions is organized.  
Related solutions with $S \propto M^2$ have appeared in the literature, see \cite{Brustein:2023hic,Banks:2002fj}.

In Sec.~\ref{sec:beyondnested}, we show the $\delta=2$ family of self-similar solutions satisfies $S=4 \pi M^2$ under stringent conditions. This thermodynamic mimicry holds only for massless walls, so every member of the $\delta=2$ family other than the HBH (where $\sigma = 2p$) violates DEC.  

\subsection{Marginal Joint Stability Selects the Hillingar Black Hole}\label{sec:hillingarcontinuumlimit}

We have seen that the anisotropic TOV equation, in the zero radial pressure limit, organizes matter into discrete Israel layers. At the photon sphere these layers are infinitesimal and continuous, saturating the BLP bound $\sigma=2p$, placing it at the marginal boundary of thermodynamic and mechanical stability. No IJC is needed because $\sigma = 2p$ in Eq.~\eqref{eq:p_over_sigma}. Each shell sits at its local photon sphere; the mass interior to this shell satisfies $\widehat{m}(r)=r/3$. 

These constraints lead to the solution in \cite{Riojas:2026ytt} by identical reasoning. For the self-similar and transverse anisotropic null gas, the TOV equation becomes algebraic. Inserting $\widehat{m}(r)=r/3$ into Eq.~\eqref{eq:EEttinteg} and solving for $\Lambda$ determines $j(r)$ in Eq.~\eqref{eq:metricform}:
\begin{equation}
    j(r)=\frac{1-r^2 \Lambda}{3}={\frac{1}{3} \mp \frac{r^2}{L^2}}.
\end{equation}
The photon sphere condition $r f'/2f=1$ gives $f(r) \propto r^2$. Its coefficient can be determined from the Lyapunov exponent $\lambda_M$, which controls massless geodesics near the photon sphere:
\begin{equation}
    \lambda_M^2 \equiv \frac{f(r)}{r^2}\Bigg|_{r=3M} = \left(\frac{1}{27 M^2}-\frac{\Lambda}{3}\right).
\end{equation}
It is the inverse timescale associated with the rate that particles spiral toward or away from the photon sphere \cite{Cardoso:2008bp}. To match to the exterior Schwarzschild metric we must choose:
\begin{equation}
    f(r)=\lambda_M^2 r^2=\left(\frac{1}{27 M^2}-\frac{\Lambda}{3}\right) r^2.
    \label{eq:matchingexteriorsolution}
\end{equation}
Here $M$ is the ADM mass. Then the metric Eq.~\eqref{eq:metricform} takes the form:
\begin{align}\label{HBHAdSmetric0}
    j(r)\equiv a^{-1} = 1-\frac{\Lambda r^2}{3}-\frac{2\widehat{m}(r)}{r}  , \quad \quad 
    f(r)
    = \left[\frac{\lambda_M}{\hat \lambda(r)}\right]^2  j(r) .
\end{align}
The Lyapunov exponent $\hat{\lambda}(r)$ and mass-function $\widehat m(r)$ are piecewise defined:
\begin{equation}
\hat{\lambda}^2(r)=\frac{1}{27\,\widehat{m}(r)^2}-\frac{\Lambda}{3},
\qquad
\widehat{m}(r)=
\begin{cases}
m, & r\in (r_+,3m)\\
r/3, & r\in (3m,3M)\\
M, & r\in (3M,\infty)   
\end{cases} 
\end{equation}
This has recovered the solution considered in \cite{Riojas:2026ytt,Riojas:2026mfu}: the HBH.

This system sits at the single location where thermodynamical and mechanical stability for thin shells overlap in Fig.~\ref{fig:BLP}, realizing the thought experiment illustrated in Fig.~\ref{fig:thoughtexperiment}; see Fig.~\ref{fig:HBH}. Entropy maximization gives $\nu=1/3$ (the photon sphere position) and $\delta=2$ (the photon sphere condition); every radial shell within the ocean sits at the local photon sphere. 

The HBH remains marginally stable in the continuum (see Sec.~\ref{sec:Stability}) for a strictly traceless equation of state. To show this result, we use the trace-reversed IJCs \eqref{eq:IJC}, which are:
\begin{equation}
    \left[K_{a b}\right]-\gamma_{a b}[K]=-\kappa_D^2 S_{a b} .
\end{equation}
Since $2 G_{\mu \nu} n^\mu n^\nu=-{ }^{(3)} R+K^2-K_{a b} K^{a b}$ (Wald E.2.27 \cite{Wald:1984rg}), this takes the useful form:
\begin{equation}
    2 \kappa_D^2\left[P_r\right]=\left[K^2\right]-\left[K_{a b} K^{a b}\right] \implies S^{a b} \bar{K}_{a b}=\left[P_r\right]=0,
\end{equation}
where $\bar{K}_{a b} \equiv \frac{1}{2}\left(K_{a b}^{+}+K_{a b}^{-}\right)$ is the average over the two sides of the layer, and we assumed $P_r=0$ on each side. Then using Eq.~\eqref{eq:Kthetatheta}, Eq.~\eqref{eq:Ktautau}, and $d\big(\dot{R}^2\big) / d R=2 \ddot{R}$, this gives:
\begin{equation}
    \frac{d}{d R} \ln \left(j+\dot{R}^2\right)=\frac{4 p}{\sigma R}-\frac{f^{\prime}}{f}+\frac{j^{\prime}}{j} .
\end{equation}
The conservation law Eq.~\eqref{eq:energyconservationshell} then gives $\frac{d}{d R} \ln \left(\sigma R^2\right)=\frac{-2p}{\sigma R}$. (The integration constant is determined by setting $\dot{R}=0$ for a static shell.) Defining $z \equiv R/R_0$ then gives: 
\begin{equation}
    \dot{R}^2+V(z)=0, \quad V(z)=j\left(R_0 z\right)\left[1-\frac{f\left(R_0\right) \sigma_0^2}{f\left(R_0 z\right) \sigma(z)^2 z^4}\right],
\end{equation}
where as before the shell's location is at $z =1$. Consistently, a straightforward computation shows the ocean is static when made of traceless matter, since $V'(1)=0$ when $R_0 f'(R_0)/2f(R_0)=2p_0/\sigma_0$. Following the same reasoning as in Sec.~\ref{sec:Stability}: 
\begin{equation}
    V^{\prime \prime}(1)=j_0\left[8\left(1+\frac{p_0}{\sigma_0}\right)\left(\lambda-\frac{p_0}{\sigma_0}\right)+\frac{4 p_0}{\sigma_0}+R_0^2\left(\frac{f^{\prime \prime}}{f}-\frac{f^{\prime 2}}{f^2}\right)\right],
\end{equation}
where $\lambda \equiv(d p / d \sigma)_0$. For the HBH ocean in asymptotically flat spacetime we have $f \propto R^2$ and $j=1/3$,  which  gives $V''(1) = 4(\lambda-1/2)$. In AdS the argument is similar, but $j=1/3 + R^2/L^2$, and so $V''(1) = (4 + 12 R_0^2/L^2)(\lambda-1/2)$. We conclude that the HBH ocean is marginally stable even when the equation of state is held fixed at $p/\sigma=1/2$.

\subsubsection{The Coldest Stable Nested Shell Configuration: the HBH}
\label{sec:ascoldasablackhole}

Consider a black hole of horizon mass $m_0$, surrounded by finitely many static spherical shells, with successive vacuum regions $R_k$ with Schwarzschild masses $m_{k}$: 
\begin{equation}
    R_1<\cdots<R_N, \quad \quad m_0<m_1<\cdots<m_N=M.
\end{equation}
Let us require that each of these shells is mechanically stable; then by Sec.~\ref{sec:BLP}: 
\begin{equation}
    R_k \geq 3 m_{k-1} \quad(k=1, \ldots, N) .
\end{equation}
We show the temperature exceeds $T_\infty = \frac{1}{8 \pi M}$ for all configurations with finitely many shells, extending the argument shown in Appendix C of \cite{Riojas:2026mfu}, for traceless shells, to the general case.  The exceptions are the trivial case with no shells, which is an ordinary Schwarzschild black hole of mass $M$, and a non-trivial case with infinitely many infinitesimally light traceless shells at the photon sphere, the HBH. 

Let $T_H$ denote the temperature at infinity of a Schwarzschild black hole of mass $M$. Let $T_\infty$ be the resulting temperature of the combined configuration. The Schwarzschild temperature measured at infinity is:
\begin{equation}
    T_{\infty}=\frac{1}{8 \pi m_0} \prod_{k=1}^N \sqrt{\frac{1-\frac{2 m_k}{R_k}}{1-\frac{2 m_{k-1}}{R_k}}}.
\end{equation}
As compared to the Hawking temperature of a black hole of mass $M$, the ratio is: 
\begin{equation}
    \frac{T_{\infty}}{1 /(8 \pi M)}=\frac{M}{m_0} \prod_{k=1}^N \sqrt{\frac{1-\frac{2 m_k}{R_k}}{1-\frac{2 m_{k-1}}{R_k}}}\equiv \prod_{k=1}^N \Phi_k(R_k),
\end{equation}
\begin{equation}\label{eq:redshiftfactor}
    \Phi_k(R)\equiv \frac{m_{k}}{m_{k-1}} \sqrt{\frac{1-\frac{2 m_k}{R}}{1-\frac{2 m_{k-1}}{R}}} .
\end{equation}
For fixed $m_{k-1}<m_k$, $\Phi_k(R)$ is strictly increasing for $R > 2m_k$, because: 
\begin{equation}
    \frac{d}{d R} \ln \Phi_k=\frac{m_k-m_{k-1}}{\left(R-2 m_k\right)\left(R-2 m_{k-1}\right)}>0 .
\end{equation}
A short calculation shows the unique crossing point where $\Phi_k(R_k)=1$ is: 
\begin{equation}
    R_k^*=\frac{2\left(m_k^2+m_k m_{k-1}+m_{k-1}^2\right)}{m_k+m_{k-1}} , \quad \quad 3 m_{k-1}<R_k^*<3 m_k .
    \label{eq:crossing}
\end{equation}

Now we show that $\Phi(R_k)>1$ for stable nested shells. Because $\Phi(R)$ is a monotonically increasing function of $R$, this amounts to showing that $G(\lambda, \alpha_+, \alpha_-)$ is strictly negative when evaluated at the crossing point $R = R_k^*$, with $\lambda \equiv c_s^2$ the squared speed of sound: 
\begin{equation}
    G\left(\lambda, \alpha_{+}, \alpha_{-}\right) \equiv 3(4 \lambda+1) \alpha_{+}^3 \alpha_{-}^3+4 \lambda \alpha_{+}^2 \alpha_{-}^2-\left(\alpha_{+}^2+\alpha_{-}^2+\alpha_{+} \alpha_{-}\right).
\end{equation}
Here $\alpha_{ \pm} \equiv \sqrt{1-2 m_{ \pm} / R}$, with $m_-=m_{k-1}$ and $m_+ =m_k$. For shells where sound wave propagation is causal, we need only consider $\lambda \le 1$. The condition $\Phi_k\left(R_k^*\right)=1$ reads $m_k^2 \alpha_{+}^2=m_{k-1}^2 \alpha_{-}^2$, giving $\alpha_+ = \eta \alpha_-$, where for convenience we define: 
\begin{equation}
    \eta \equiv \frac{m_{k-1}}{m_k} \in(0,1).
\end{equation}
It is straightforward to show that inverting Eq.~\eqref{eq:crossing} for $\alpha_{-}^2(R_k^*)$ gives: 
\begin{equation}
    \alpha_{-}^2(R_k^*)=1-\frac{2 m_{k-1}}{R_k^*}=\frac{m_k^2}{m_k^2+m_k m_{k-1}+m_{k-1}^2}=\frac{1}{1+\eta+\eta^2}.
\end{equation}

Substituting $\alpha_+ = \eta\, \alpha_-$ and $\alpha_-^2 = \frac{1}{1+\eta+\eta^2}$ into $G(\lambda, \alpha_+, \alpha_-)$, one finds that it is negative when the following polynomial is negative:
\begin{equation}
    G<0 \iff P(\eta, \lambda) \equiv 3(4 \lambda+1) \eta^3+4 \lambda \eta^2\left(1+\eta+\eta^2\right)-\left(1+\eta+\eta^2\right)^3<0.
\end{equation}
Causality requires $\lambda \le 1$, which suggests separating this polynomial as: 
\begin{equation}
    P(\eta, \lambda)=P(\eta, 1)+\left(\lambda-1\right)\cdot \left(4 \eta^2\left(1+4 \eta+\eta^2\right)\right) .
    \label{eq:conP}
\end{equation}
\begin{equation}
    P(\eta, 1)=-1-3 \eta-2 \eta^2+12 \eta^3-2 \eta^4-3 \eta^5-\eta^6.
    \label{eq:palindrome}
\end{equation}
The second term in Eq.~\eqref{eq:conP} is manifestly negative for causal shells, so it suffices to show the first term Eq.~\eqref{eq:palindrome} is negative on $\eta \in (0,1)$. That term takes the form of, essentially, the simplest possible negative palindromic polynomial that manifestly vanishes at $\eta=1$.\footnote{It is simple to numerically verify this claim. However a satisfying analytic argument exists as follows. The Abel-Ruffini theorem states that no solution in radicals exists for the roots of a polynomial of degree $5$ or higher with arbitrary coefficients. However, \textit{palindromic} polynomials of degree $n$ have the following remarkable property: for even degree $2k$, divide by $x^k$ and substitute $u=x+\frac{1}{x}$ to reduce their degree by half. In our case, this process (with $N \equiv \eta + \eta^{-1}$) gives a cubic polynomial that factors as $-(N-2)(N^2 + 5N + 9)$. For $\eta \in (0,1)$ we have $N>2$, so $(N-2)$ is positive. Then $P(\eta, 1)<0$ for $\eta \in (0,1)$, and so $G<0$ at $R_k^*$.} 

Then $\Phi(R_k)>1$ for stable nested shells. Multiplying over all shells gives $T_{\infty}>\frac{1}{8 \pi M}$. If there are no shells, we have simply a Schwarzschild black hole of mass $M$, with temperature $T_\infty = \frac{1}{8 \pi M}$. There is also a non-trivial limit suggested by Eq.~\eqref{eq:crossing}. Let these shells' masses become infinitesimal, as they approach the local photon sphere: 
\begin{equation}
    \delta m_k \equiv m_k-m_{k-1} \rightarrow 0, \quad \quad  R_k \rightarrow 3 m_k .
\end{equation}
In this limit each $\Phi_k \rightarrow 1$. More precisely, writing $m_k=\left(1+\varepsilon _k\right) m_{k-1}$:  
\begin{equation}
    \Phi_k\left(3 m_k\right) \equiv \frac{m_k}{m_{k-1}} \sqrt{\frac{1-\frac{2 m_k}{3m_k}}{1-\frac{2 m_{k-1}}{3m_k}}} =\frac{\left(1+\varepsilon_k\right)^{3 / 2}}{\sqrt{1+3 \varepsilon_k}}=1+\frac{3 \varepsilon_k^2}{2} + O(\varepsilon_k^3).
\end{equation}
In the continuum limit the product tends to $1$, and the resulting configuration is the HBH: 
\begin{equation}
    \widehat{m}(r) = \frac{r}{3}, \quad \quad  r \in (3m, 3M). \quad \quad T_{\infty}= \frac{1}{8 \pi M}.
\end{equation}
In AdS, $\Phi(R)$ is again monotonic, but $\Phi(R_k)=1$ can now occur within a broad region where mechanical and thermal stability are satisfied (Sec.~\ref{sec:mechAdS}); these stable shells can cool the combined system while holding $M$ constant. If that radiation is in thermal equilibrium with the black hole, then the coarse-grained entropy increases.

\subsection{Shells Cool Black Holes $\iff$ Specific Heat is Positive}\label{sec:heatcapandtemp}

Here we clarify the physical meaning of the divergent specific heat at the photon sphere of a black hole. The same mechanism controls nested shell configurations. Define: 
\begin{equation}
    T_{\infty}=T_H\left(m_0, L\right) \prod_{k=1}^N \sqrt{\frac{f_{m_k}\left(R_k\right)}{f_{m_{k-1}}\left(R_k\right)}}.
\end{equation}
Here we are allowing for an AdS curvature $L$. The ratio of the temperature at infinity for the shell configuration, $T_\infty$, to the Hawking temperature of a black hole of mass $M$ is: 
\begin{equation}
    \frac{T_{\infty}}{T_H(M)}=\prod_{k=1}^N \Phi_k\left(R_k\right), \quad \quad \Phi(R)=\frac{T_H(m, L)}{T_H(m+\delta m, L)} \sqrt{\frac{f_{m+\delta m}(R)}{f_m(R)}}.
\end{equation}
For one infinitesimal shell, varying the mass while holding $R$ fixed gives:
\begin{equation}
    \delta \ln \frac{T_{\infty}}{T_H(M, L)}=\delta \ln \Phi(R)=-\delta \ln T_H+\frac{1}{2} \delta \ln f_m(R).
\end{equation}
The second term takes a simple form for black hole metrics:
\begin{equation}
     \frac{1}{2} \delta \ln f_m(R)=\frac{1}{2} \frac{\partial_m f_m(R)}{f_m(R)} \delta m.
\end{equation}
Then the variation of the redshift factor for the shells, $\Phi_k$, gives: 
\begin{equation}
    \delta \ln \Phi(R)=-\left(\frac{d \ln T_H}{d m}+\frac{1}{R f_m(R)}\right) \delta m =-\frac{1}{T_H}\left(\frac{d T_H}{d m}+\frac{T_H}{R f_m(R)}\right) \delta m.
    \label{eq:variationredshift}
\end{equation}
The specific heat at finite radius is given by Eq.~\eqref{eq:general_specific_heat}:
\begin{equation}
    C_R^{-1}=\frac{d T_H}{d m}+\frac{T_H}{R f(R)}.
\end{equation}
Comparing to Eq.~\eqref{eq:variationredshift} gives one of the main results of this note, illustrated in Fig.~\ref{fig:shellsandtemp}:
\begin{equation}
    \boxed{\delta \ln \frac{T_{\infty}}{T_H(M, L)}=-\frac{C_R^{-1}}{T_H} \delta m.}
    \label{eq:specificheatchangetemp}
\end{equation}
Forming shells of thermal radiation around a black hole, with fixed total mass $M$, cools the system if and only if the specific heat at the location of the shell is \textit{positive}. 

To determine whether a black hole surrounded by nested shells is colder than the original configuration, we compute the ratio: 
\begin{equation}
    \frac{T_{\infty}}{T_H(M)}=\prod_{k=1}^N \Phi_k\left(R_k\right), \quad \Phi_k=\frac{T_H\left(m_{k-1}\right)}{T_H\left(m_k\right)} \sqrt{\frac{f_{m_k}\left(R_k\right)}{f_{m_{k-1}}\left(R_k\right)}}.
\end{equation}
It was just shown that: 
\begin{equation}
    \delta \ln \Phi=-\frac{C_R^{-1}}{T_H} \delta m, \quad C_R^{-1}=\frac{d T_H}{d m}+\frac{T_H}{R f(R)}.
\end{equation}
So $\Phi<1$ if and only if $C_R^{-1}>0$. Mechanically stable shells in asymptotically flat spacetime have $C_R^{-1}\le 0$, with equality at the photon sphere; this is why the HBH, the coldest such configuration, is exactly as cold as a black hole. 

\begin{figure}
    \centering
    \includegraphics[width=\linewidth]{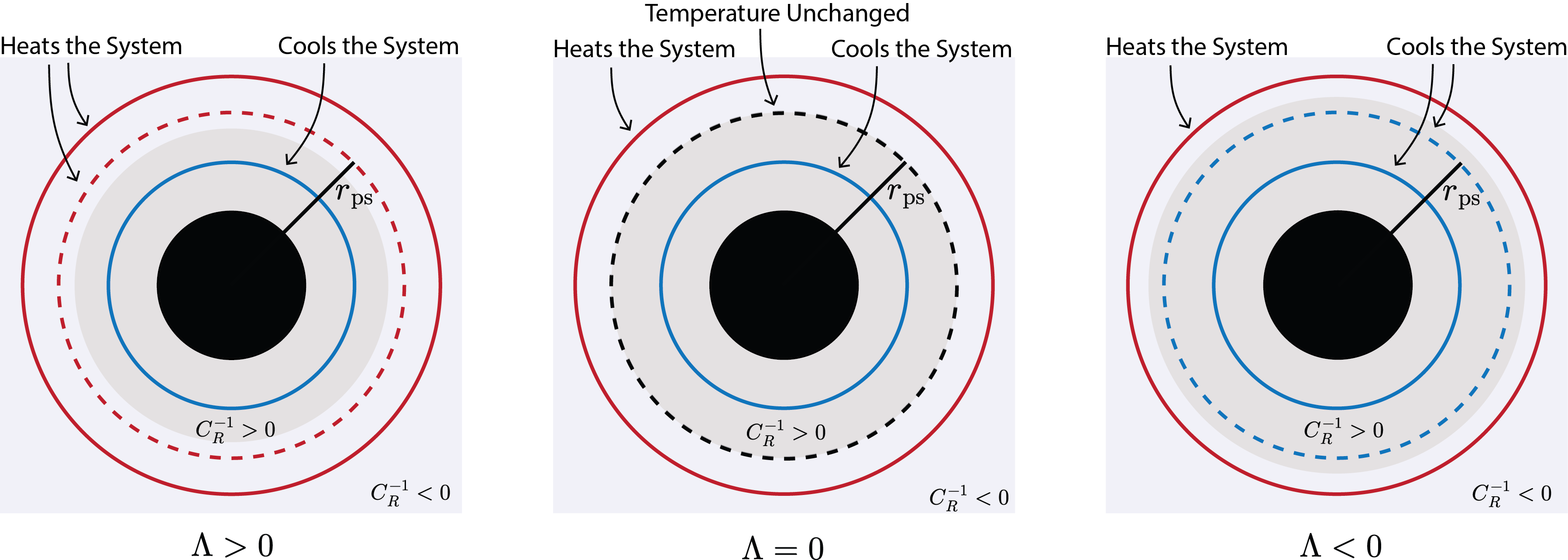}
    \caption{Locations of possible shells. The asymptotically measured Hawking temperature $T_\infty$, after building a shell around a black hole, changes with sign opposite to that of the locally measured specific heat (Eq.~\eqref{eq:specificheatchangetemp}). Shells placed where $C_R^{-1}<0$ increase the asymptotically measured temperature (red), while shells placed where $C_R^{-1}>0$ decrease it (blue). At the photon sphere, $C_{r_0}^{-1}(r_{\mathrm{ps}}) \propto -\Lambda$ (Eq.~\eqref{eq:PSspecificheatlambda}).  As illustrated (center) for asymptotically flat spacetime, the temperature is unchanged when shells are added at the photon sphere at fixed total mass.   }
    \label{fig:shellsandtemp}
\end{figure}

However, in asymptotically AdS spacetime, $C_R^{-1}>0$ at the photon sphere, so the evaporation process would run away. This gives a different perspective on \cite{Riojas:2026mfu,RiojasStrasslerAdS}, where we pointed out that AdS black holes may be thermodynamically unstable in the microcanonical ensemble. From this perspective, they can increase their entropy by emitting radiation and cooling off. They thus increase their entropy in the \textit{microcanonical} ensemble for the unintuitive reason that their \textit{local} specific heat is positive in the \textit{canonical} ensemble.\footnote{This counter-intuitive result was also found for the HBH from entropy maximization \cite{Riojas:2026mfu,RiojasStrasslerAdS}; the AdS case is similar to Sec.~\ref{sec:entropyfixedpoint}. One could demand that this \textit{positive} specific heat system cannot reach thermal equilibrium with a bath, due to AdS/CFT. This would constrain bulk matter in AdS/CFT. }

\subsection{Entropy of Systems in Equilibrium with Black Holes}
\label{sec:EntropyandShells}

Here we introduce a practical method for computing coarse-grained entropies without using the Euclidean path integral, which we then greatly generalize in Sec.~\ref{sec:beyondnested}.\footnote{Given the Hawking temperature, the IJC reproduce York's quasilocal energy, surface pressure and first law without the Euclidean action. We also make use of the free energy $F = E_{BY}-TS$.}

First we consider a simpler case at $P_r=0$, where our procedure can be understood as follows. By adding shells and tracking the asymptotically measured temperature, the coarse-grained entropy of a system in thermal equilibrium with a black hole is calculated. Shells in equilibrium around black holes have been considered previously, see \cite{Martinez:1996ni,Andre:2019zzo,Andre:2021ctu,Lemos:2023yiz}. 

\subsubsection{Thermodynamics from IJC: Shells are Heat}\label{sec:shellsareheat}

Here we derive black hole thermodynamics directly from the IJCs, with only the Hawking temperature and thermal equilibrium as inputs; the geometry is built layer by layer while tracking the asymptotic temperature. A summary is as follows. 

The IJCs in Eq.~\eqref{eq:staticjunctionbh} determine the surface density $\sigma$ of an infinitesimal thin shell at radius $R$ between Schwarzschild geometries of mass $m$ and $m + \delta m$. Expanding to first order in $\delta m$, we obtain the well-known Brown-York quasilocal energy \cite{York:1986it}. The asymptotic temperature Eq.~\eqref{eq:tempformulafromshell} follows from gravitational redshift. If  local thermodynamic equilibrium holds, the gravitational first law $\delta S(\widehat{m})=\delta \widehat{m}/T_\infty(\widehat{m})$ Eq.~\eqref{eq:entropyofshells} follows, in a local form depending on the mass $\widehat{m}(r)$ enclosed within areal radius $r$.

We begin with the free energy $F \equiv E_{BY} - T S$ \cite{York:1986it}, where $E_{BY}$ is the Brown-York energy satisfying $(\delta E_{BY})_{r} \equiv \delta M \left( 1-\frac{2M}{r}\right)^{-1/2}$. Expanding the Israel junction condition for infinitesimally thin static shells Eq.~\eqref{eq:staticjunctionbh} in $\delta M$:
\begin{equation}
 \sigma = \frac{1}{4 \pi R}\left( \alpha_- - \alpha_+\right) \approx \frac{1}{4 \pi R^2} \frac{\delta M}{\sqrt{1-\frac{2M}{R}}} \implies (\delta E_{BY})_r \equiv  \frac{\delta M}{\sqrt{1 - \frac{2M}{R}}} = A \sigma.
 \label{eq:BYEshell}
\end{equation}
It follows from direct integration in $\delta M$ that $E_{B Y}=R\left(1-\sqrt{1-2 m_{+} / R}\right)$. York showed \cite{York:1986it} that $dE_{BY} = T(r) d S  - p_Y dA$, where the surface pressure is $p_Y = (8 \pi R)^{-1}(\frac{1-m_+/R}{\alpha_+}-1)$. To recover his exact differential from our approach, we compute:
\begin{equation}
    \left.\frac{\partial E_{BY}}{\partial R}\right|_{m_+} dR = \left( 1 - \frac{1-m_+/R}{\alpha_+}\right) d R = \frac{1}{8 \pi R}\left( 1 - \frac{1-m_+/R}{\alpha_+}\right) dA.
\end{equation}
Then $d E_{B Y}=T(r) d S-p_Y d A$. We have derived York's local thermodynamics directly from the IJCs and Hawking temperature, without the Euclidean path integral. 

The next step gives a physical picture. The $T dS$ term in York's thermodynamics increases the mass of the system -- which amounts to adding heat -- by depositing a thin Israel layer at the boundary. The $p_Y dA$ term holds the mass fixed while varying the boundary area. Then holding the boundary area fixed, this increases the entropy:
\begin{equation}
    \boxed{\delta S(\widehat{m})=\frac{\delta E_{BY}}{T(r)}=\frac{\delta \widehat m / \sqrt{f}}{T_{\infty} / \sqrt{f}}=\frac{\delta \widehat m}{T_{\infty}(\widehat{m})}.}
    \label{eq:entropyofshells}
\end{equation}
This is the first law of gravitational thermodynamics, obtained through IJCs. Normalization is set by (e.g.) Bogoliubov coefficients, see Appendix \ref{sec:bogoliubov}. 

The coarse-grained entropy of infinitesimally light nested shells in thermal equilibrium with a black hole can be computed from Eq.~\eqref{eq:entropyofshells}. For one infinitesimal step at radius $r$: 
\begin{equation}
    T_{\infty}(\widehat{m}+\delta\widehat{m})=T_H(m_0) \prod_k \sqrt{\frac{1+\frac{R_k^2}{L^2}-\frac{2\left(\widehat{m}+\delta \widehat{m}\right)}{R_k}}{1+\frac{R_k^2}{L^2}-\frac{2 \widehat{m}}{ R_k}}} \implies \boxed{d \ln T_{\infty}=-\frac{d \widehat{m}}{r f(r)}.}
    \label{eq:tempformulafromshell}
\end{equation}
This can determine $T_\infty$ from $\widehat{m}(r)$. For instance, the HBH ($L\rightarrow \infty$) satisfies $\widehat{m}=r/3$, for which the continuum product of shell factors gives:
\begin{equation}
    d \ln T_{\infty}=-\frac{d \widehat{m}}{\widehat{m}} \implies T_{\infty}(M)=T_{\infty}(m) \frac{m}{M} = \frac{1}{ 8 \pi M}.
    \label{eq:asymptempshell}
\end{equation}

These expressions (Eq.~\eqref{eq:entropyofshells} and Eq.~\eqref{eq:tempformulafromshell}) suffice to compute coarse-grained entropies of configurations of nested shells using simple elementary integrals. To the best of our knowledge, this method of ``printing" spacetime layer by layer while tracking the asymptotic temperature has not appeared before. This method is used below to recover results obtained from Euclidean path integral and entropy maximization methods in  \cite{Riojas:2026ytt,Riojas:2026mfu}. It is generalized further in Sec.~\ref{sec:beyondnested}, yielding a general Euler relation in Eq.~\eqref{eq:totalentropy}.

\subsubsection{Computing Coarse-Grained Entropies from IJC}\label{sec:coarsegrainedIJC}

Here we compute coarse-grained entropies using the IJC, beginning with the HBH in asymptotically flat spacetime, assuming thermal equilibrium for these systems.

The coarse-grained entropy of the ocean -- a cloud of massless particles, e.g. photons, orbiting and extending the photon sphere of the black hole -- can be computed from Eq.~\eqref{eq:entropyofshells} and Eq.~\eqref{eq:tempformulafromshell}. First Eq.~\eqref{eq:tempformulafromshell} gives $T(\widehat{m})= (8 \pi \widehat{m})^{-1}$; then Eq.~\eqref{eq:entropyofshells} gives: 
\begin{equation}
    S_{\mathrm{ocean}}\equiv S(M)-S(m)=\int_m^M 8 \pi \widehat{m} \ d \widehat{m}=4 \pi\left(M^2-m^2\right).
\end{equation}
After including the black hole with $S_{\mathrm{bh}} = 4 \pi m^2$, the total entropy of the HBH \cite{Riojas:2026ytt,Riojas:2026mfu} is:\footnote{It is important to realize that, while we lack a microphysical model, the Hawking temperature from a Bogoliubov calculation has precisely the right value for the ocean to reach thermal equilibrium, see App.~\ref{sec:bogoliubov}. The ocean's entropy density $s \propto 1/Gr$ is consistent with anisotropic thermodynamics, see Sec.~\ref{sec:entropyfixedpoint}. }
\begin{equation}
    S(M)= S_{\mathrm{bh}}(m) + S_{\mathrm{ocean}}(M,m)=4 \pi M^2 .
\end{equation}

Evaporation into an ocean is not entropically forbidden since the initial Schwarzschild black hole has the same coarse-grained entropy as an HBH. Throughout the evaporation process, for thermal equilibrium to be maintained the constant value of $T_\infty$ given by Eq.~\eqref{eq:asymptempshell} is required. This is a non-trivial \textit{output} of the model; $C_R^{-1}=0$ at the photon sphere, so the asymptotic temperature is fixed as the black hole evaporates; see Eq.~\eqref{eq:specificheatchangetemp}. 

This procedure works equally well for an ordinary Schwarzschild black hole; placing layers on the horizon of the hole builds up an entropy $S = 4 \pi M^2$. To see this, suppose successive shells are placed at their own Schwarzschild radius $R = 2\widehat{m}$. Recall that: 
\begin{equation}
    \sigma = \frac{\delta m}{4 \pi R^2 \sqrt{1-\frac{2 \widehat{m}}{r}}}, \quad \quad T_{\mathrm{loc}} = \frac{T_H}{\sqrt{1-\frac{2 \widehat{m}}{r}}}.
\end{equation}
These quantities diverge at the horizon, but the shell entropy $\delta S$ is redshift independent:
\begin{equation}
    \delta S = \frac{\delta E_{BY}}{T(r)} = \frac{\delta \widehat{m}/\alpha_-}{T_\infty(\widehat{m})/\alpha_-}  = 8 \pi \widehat{m} \ \delta \widehat{m} \implies S_{bh}(M) = 4 \pi M^2.
\end{equation}
This is consistent with the Bekenstein bound. The HBH case has been illustrated in Fig.~\ref{fig:3Dprintingentropy}.

\begin{figure}
    \centering
    \includegraphics[width=\linewidth]{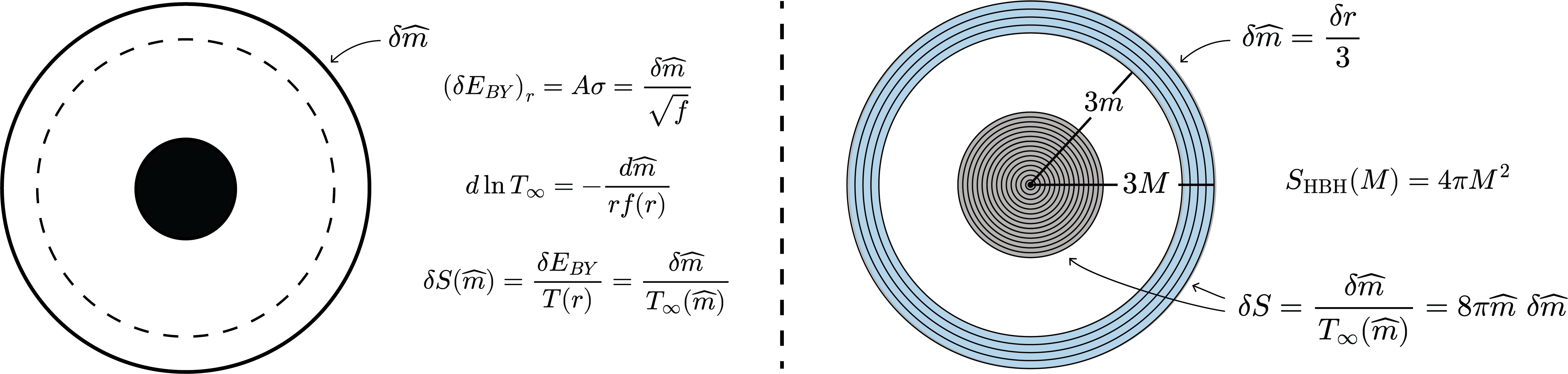}
    \caption{An illustration of the procedure where coarse-grained entropies are computed by building up the geometry layer by layer. Left: the addition of one light shell changes the Brown-York energy by $\delta \widehat{m}/{\sqrt{f}}$, and also changes the asymptotically measured Hawking temperature. This builds up some coarse-grained entropy $\delta S(\widehat{m})=\delta \widehat{m}/T_\infty(\widehat{m})$. Right: The coarse-grained entropy of a black hole of mass $m$, and an ocean of mass $M-m$, is computed by building it up layer by layer.}
    \label{fig:3Dprintingentropy}
\end{figure}

The same formalism can be readily extended to more general configurations. Now we consider Einstein clusters formed by a continuum of nested shells, recovering the results first obtained in the companion papers \cite{Riojas:2026ytt,Riojas:2026mfu} using our independent method. 

For an Einstein cluster around a black hole, the asymptotically measured Hawking temperature satisfies Eq.~\eqref{eq:tempformulafromshell}. For these solutions, Eqs.~\eqref{eq:selfsimilarrelations}-\eqref{eq:selfsimilardelta}  give $\widehat m = \nu r$ and $\delta = \frac{4 w_\perp}{1 + w_r} = 4 w_\perp$, and then:
\begin{equation}
    d \log T_{\infty} = - \frac{d \widehat m}{r f(r)} = - \frac{\nu}{1-2 \nu} \frac{d \widehat m}{\widehat m} = -\frac{\delta}{2} \frac{d \widehat m}{\widehat m}=-w_{\perp} d \log (A).
\label{eq:asymptotictemperatureA}
\end{equation}
In the last step we used $A = 4 \pi r^2$. For Schwarzschild black holes of mass $m$ this gives:
\begin{equation}
    T_\infty(M) = \frac{1}{8 \pi m} \left( \frac{m}{M}\right)^{\delta/2}.
\end{equation}
If this cluster can be held in thermal equilibrium with the black hole, Eq.~\eqref{eq:entropyofshells} shows: 
\begin{equation}
    S_O = \int_m^M \frac{d \widehat m}{T_\infty(\widehat m)} = \int_m^M d \widehat m \left( 8 \pi m \right) \left(\frac{\widehat m}{m}\right)^{\delta/2} = \frac{16 \pi}{2 + \delta} \left( M^{1+\delta/2} m^{1-\delta/2}-m^2\right).
\end{equation}
This was also obtained in Sec.~5.1 of \cite{Riojas:2026mfu} by a different method. For $\delta=2$ the total coarse-grained entropy of this system around a black hole is $S= 4 \pi M^2$. As a consistency check, note that $\frac{\delta}{2}\equiv\frac{r f'(r)}{2 f(r)}=1$ is precisely the HBH, for which massless particles orbit within an extended photon sphere. 

In anti-de Sitter spacetime we have seen that forming shells around black holes at fixed total mass cools the combined system. Unlike in asymptotically flat spacetime, Sec.~\ref{sec:Stability} shows there exists an extended region where such shells are \textit{stable} rather than marginally stable. At the local photon sphere, $f = \frac{1}{3} + \frac{9 \widehat{m}^2}{L^2}$, which using Eq.~\eqref{eq:tempformulafromshell} gives \cite{Riojas:2026mfu}: 
\begin{equation}
    d \ln T_{\infty}=-\frac{d \widehat{m}}{\widehat{m}\left(1+\frac{27 \widehat{m}^2}{L^2}\right)}\implies T_\infty(M)=T_\infty(m)\frac{m}{M} \sqrt{\frac{1+\frac{27 M^2}{L^2}}{1+\frac{27 m^2}{L^2}}} .
    \label{eq:tempcompared}
\end{equation}

This system cools as shells are added, increasing the coarse-grained entropy. The minimum temperature is reached in the regime where only a small black hole remains; it is surrounded by an ocean of radiation with $M\gg m$. The Hawking temperature, before considering shells, is $T_H(m) = {1}/{8 \pi m}$; the asymptotically measured temperature is:
\begin{equation}
    T_{\infty}(M)=\frac{1}{8 \pi M} \sqrt{1+\frac{27 M^2}{L^2}} \implies  \lim_{M \rightarrow \infty} T_{\infty}(M) = \frac{3 \sqrt{3}}{8 \pi L}.
\end{equation}
This minimum temperature satisfies $T_{\mathrm{min}} \sim 1/L$. From Eq.~\eqref{eq:entropyofshells}, $dS$ can be written as: 
\begin{equation}
    d S=\frac{d M}{T_{\infty}(M)}=\frac{8 \pi M d M}{\sqrt{1+\frac{27 M^2}{L^2}}}.
\end{equation}
This gives\footnote{This gives another perspective on the Hagedorn-like behavior (identified first in \cite{Riojas:2026mfu,RiojasStrasslerAdS}), where it was obtained from the Euclidean path integral and an Euler relation from entropy maximization, see Sec.~\ref{sec:entropyfixedpoint}.} Hagedorn-like behavior, as shown in \cite{Riojas:2026mfu,RiojasStrasslerAdS}:
\begin{equation}\label{eq:Hagedorn}
    S(M)-S(0)=\frac{8 \pi L^2}{27}\left(\sqrt{1+\frac{27 M^2}{L^2}}-1\right) \sim \frac{8 \pi L}{3 \sqrt{3}} M =\frac{M}{T_{min}} \sim M L.
\end{equation}

Let us see what occurs for more general AdS black holes of mass $m$. From Eq.~\eqref{eq:tempcompared} evaluated at $T_\infty(\widehat{m})$, the entropy from Eq.~\eqref{eq:entropyofshells} is: 
\begin{equation}
    d S_{\mathrm{shell }}=\frac{d \widehat{m}}{T_{\infty}(\widehat{m})}=\frac{1}{T_{\infty}(m)} \frac{\widehat{m}}{m}\frac{\sqrt{1+\frac{27 m^2}{L^2}} }{{\sqrt{1+\frac{27 \widehat{m}^2}{L^2}}}} \ d \widehat{m} .
\end{equation}
Including the black hole $S = \pi r_+^2$, this immediately reproduces the result of \cite{Riojas:2026mfu,RiojasStrasslerAdS}, which was obtained using the Euclidean path integral and entropy maximization methods: 
\begin{equation}
    S = \pi r_+^2 + \frac{L^2}{27 m T_\infty(m)} \sqrt{1+\frac{27 m^2}{L^2}}\left(\sqrt{1+\frac{27 M^2}{L^2}}-\sqrt{1+\frac{27 m^2}{L^2}}\right) .
\end{equation}
If thermalization is viewed as impossible for this system, which is both mechanically and thermodynamically stable (Sec.~\ref{sec:Stability}), these solutions might constrain stress-energy tensors in AdS/CFT \cite{Maldacena:1997re}, relating bulk matter content to the boundary theory on a sphere.\footnote{Hagedorn-like scaling $S \sim E L$ holds in \cite{Riojas:2026mfu} near a minimum temperature at fixed cutoff. The expected $S \sim E^{2/3}$ scaling \cite{Witten:1998zw,Witten:2024upt} of a dual CFT at large $E$ can be recovered by taking $E$ large before sending the boundary cutoff to infinity \cite{Riojas:2026mfu}. We are very grateful to Daniel Jafferis for discussion on this point. } 

If a stable shell can reach thermal equilibrium with the black hole, then canonical stability ($C_{r_0}^{-1}>0$) is incompatible with microcanonical stability. As shown in Sec.~\ref{sec:heatcapandtemp}, forming shells at fixed total mass decreases the temperature if and only if $C_{r_0}^{-1}>0$. This increases the coarse-grained entropy of the system by forming an ocean: 
\begin{equation}\label{eq:badbehavior}
S_{\mathrm {ocean }}=\int_m^M \frac{d \widehat{m}}{T_{\infty}(\widehat{m})}>\int_m^M \frac{d \widehat{m}}{T_H(\widehat{m})}=\Delta S_{\mathrm{bh}}.
\end{equation}
In any case, if black holes can cool by forming shells while remaining in thermal equilibrium, Eq.~\eqref{eq:badbehavior} will drive a feedback loop; if unchecked at large $M$, this is consistent with Eq.~\eqref{eq:Hagedorn}. 

Finally, it is interesting to consider the consequences of placing these shells where the specific heat changes sign, as we have done in asymptotically flat spacetime:
\begin{equation}
    C_R^{-1}=\frac{d T_H}{d m}+\frac{T_H}{R f(R)}= {T_H}\left(\frac{d \log T_H}{dm}+\frac{1}{R f(R)}\right)=0.
\end{equation}
We have seen in Sec.~\ref{sec:Stability} that shells at this location are mechanically stable. In thermal equilibrium, for mechanically stable matter with non-negative specific heat, this gives:
\begin{equation}
    d \ln T_{\infty}=-\frac{d \widehat{m}}{R f_{\widehat{m}}(R)}=\frac{d \ln T_H}{d \widehat{m}} d \widehat{m}=d \ln T_H(\widehat{m}),
\end{equation}

\begin{equation}
    S_{\mathrm{ocean}}=\int_m^M \frac{d \widehat{m}}{T_{\infty}(\widehat{m})}=\int_m^M \frac{d \widehat{m}}{T_H(\widehat{m})}= \int \frac{1}{\frac{1}{4 \pi r_+}\left( 1 + \frac{3 r_+^2}{L^2}\right)}\left( \frac{1}{2}\left(1+\frac{3 r_+^2}{L^2}\right)\right) dr_+ = 
    \pi r_+^2 \big|^M_m.
\end{equation}
This system has the same coarse-grained entropy as a black hole precisely because it is exactly as cold as one. In AdS this can be done only on the small black hole branch; on the large black hole branch, $C_R^{-1}$ is strictly positive. 

The result is quite general; if the inverse specific heat vanishes at the location where the shell is placed, then the entropy of the system follows directly from the first law of black hole thermodynamics. This observation is further generalized below.

\subsubsection{Euler Relation and a Self-Similar Family as Cold as a Black Hole}\label{sec:beyondnested}

Now we consider configurations where more general self-gravitating gases are contained between two Israel layers, generalizing the discussion of Sec.~\ref{sec:coarsegrainedIJC}. 

Adding the (tt) and (rr) components of Einstein's equations, Eq.~\eqref{eq:EErrcomponent} and Eq.~\eqref{eq:EEttcomponent}, the cosmological constant cancels:
\begin{equation}\label{eq:fjandstress}
    \frac{j}{8 \pi r}\left(\frac{f^{\prime}}{f}-\frac{j^{\prime}}{j}\right)=\rho+P_r.
\end{equation}
The Hawking temperature is well known to be $T_{\infty}=\frac{1}{4 \pi} \sqrt{f^{\prime}\left(r_h\right) j^{\prime}\left(r_h\right)}$. The asymptotically measured temperature therefore changes as:
\begin{equation}
    d \log T_{\infty}=\frac{1}{2}\left[\frac{j^{\prime}(r)}{j(r)}-\frac{f^{\prime}(r)}{f(r)}\right] d r=-\frac{4 \pi r\left(\rho+P_r\right)}{j(r)} d r.
\end{equation}
This recovers Eq.~\eqref{eq:tempformulafromshell} in the vacuum thin shell limit. 

The configuration is precisely as cold as a black hole when the temperature matches point-wise Eq.~\eqref{eq:asymptempshell}; for $\Lambda$=0, Eq.~\eqref{eq:EEttinteg} requires: 
\begin{equation}
    \frac{P_r}{\rho} = \frac{3j(r)-1}{1-j(r)} = \frac{r-3 \widehat m(r)}{\widehat m(r)}.
    \label{eq:temperaturemimicry}
\end{equation}
This is closely related to the condition required for an extended photon sphere to form, as explained in Appendix \ref{sec:UCO}. An extended photon sphere forms when Eq.~\eqref{eq:extendedphotonspherecondition} holds:
\begin{equation}\label{eq:extendedphotonsphereconditionmain}
    P_\perp(r) = \frac{1- 2 \widehat{m}'(r)}{8 \pi r^2}, \quad P_r = \frac{r - 3\widehat{m}(r)}{4 \pi r^3}, \quad \rho = \frac{\widehat{m}'(r)}{4 \pi r^2}, \implies \frac{P_r}{\rho} = \frac{r-3\widehat{m}(r)}{r \widehat{m}'(r)}.
\end{equation}
Thus Hod's condition \cite{Hod:2017zpi} for degenerate light rings, applied over an extended region, takes the same form as thermodynamic mimicry; for both to hold we need $\widehat{m}' = \widehat{m}/r$, i.e. self-similar gases with $\widehat{m}(r)=\nu r$ and $rf'/2f=1$. The HBH \cite{Riojas:2026ytt,Riojas:2026mfu} is the simplest case.

To explore this in more detail, we determine coarse-grained entropies of these systems, starting with the walls enclosing the self-gravitating gas. From the lower-dimensional picture on the wall, the Euler relation is that of isotropic radiation $\sigma+p=T s_{\Sigma}$; this gives:
\begin{equation}\label{eq:shellentropy}
S_{\mathrm{wall}}=\int_{\mathrm{wall}} s_{\Sigma} d A=\frac{(\sigma+p) A}{T(R)}.
\end{equation}
The IJC for this \textit{discrete} shell, in Eq.~\eqref{eq:p_over_sigma}, can be conveniently rewritten using Eq.~\eqref{eq:extrinsiccomponents}:
\begin{equation}\label{eq:Sshell}
    8 \pi(\sigma+p)=\left[K^\tau{ }_\tau\right]-\left[K^\theta{ }_\theta\right]=\frac{1}{R}\left[\sqrt{j}\left(\frac{r f^{\prime}}{2 f}-1\right)\right]_R.
\end{equation}
The induced metric must match at the interface Eq.~\eqref{eq:inducedmetricmatching}, so $f(r)$ is a continuous function, and can be moved inside the bracket. 

The locally measured Tolman temperature satisfies $T(r)=T_\infty/\sqrt{f(r)}$.  The coarse-grained entropy on the shell of area $A = 4 \pi R^2$ is then:
\begin{equation}
    S_{\mathrm{shell}}=\frac{R}{2 T_{\infty}}\left[\sqrt{f j}\left(\frac{r f^{\prime}}{2 f}-1\right)\right]_R.
\end{equation}

So it is for walls. Now we consider the coarse-grained entropy between them. The proper volume element is $dV = 4 \pi r^2 dr/\sqrt{j}$. Then Eq.~\eqref{eq:anisothermo} and Eq.~\eqref{eq:fjandstress} yield:
\begin{equation}\label{eq:Sbulk}
    S_{\mathrm{bulk}}= \frac{1}{T_\infty} \int dr \ 4 \pi r^2 (\rho + P_r) \sqrt{\frac{f(r)}{j(r)}}=\frac{1}{2 T_{\infty}} \int dr\  r \sqrt{f j}\left(\frac{f^{\prime}}{f}-\frac{j^{\prime}}{j}\right).
\end{equation}
Assembling these results, and including the black hole, the full coarse-grained entropy is: 
\begin{equation}\label{eq:totalentropy}
    \boxed{S=\frac{A_{\mathrm {H}}}{4}+\sum_{\mathrm {shells}} \frac{1}{2 T_{\infty}}\left[r\sqrt{f j}\left(\frac{r f^{\prime}}{2 f}-1\right)\right]_{R_i}+\frac{1}{2 T_{\infty}} \int_{\mathrm{gas}} r \sqrt{f j}\left(\frac{f^{\prime}}{f}-\frac{j^{\prime}}{j}\right) d r.}
\end{equation}
To sum up, the area term for the horizon is added to the coarse-grained entropy of the (an)isotropic fluid and, when they are needed, the walls. No wall is needed at the photon sphere, where its respective term vanishes. The gas term vanishes in vacuum, where $f=j$.

The Euler relation\footnote{Following identical reasoning to that of the argument originally given in Section 4.2 of \cite{Riojas:2026mfu}, this expression can be used to find an Euler relation for a \textit{misaligned} self-similar gas with \textit{massless walls}:
\begin{equation}
    2 T_{\infty} S=\frac{4+w_r(2-\delta)}{\delta+2}\left[M-m\left(\frac{m}{M}\right)^{\delta / 2}\right].
\end{equation}
Indeed, this is an Euler relation for the self-similar gas. For $\delta=2$ and $m=0$ we recover $2 T_\infty S = M$, the same as that satisfied by a black hole. This is consistent with Eq.~\eqref{eq:entropylayers} and with the earlier result in \cite{Riojas:2026mfu}.} in Eq.~\eqref{eq:totalentropy} is general, and one of the main results of this note. It permits direct evaluation of coarse-grained entropies without the Euclidean path integral. The discussion below Eq.~\eqref{eq:temperaturemimicry} now motivates us to specialize to self-similar solutions. Einstein's equations require $8 \pi r^2\left(\rho+P_{\perp}\right)=1$, the condition responsible for the optical mimicry of the HBH. These satisfy Eq.~\eqref{eq:selfsimilarrelations} and Eq.~\eqref{eq:selfsimilardelta}, which gives Eq.~\eqref{eq:extendedphotonsphereconditionmain}:
\begin{equation}\label{eq:selfsimilargas}
    \widehat{m} = \nu r, \quad \rho = \frac{\nu}{4 \pi r^2}, \quad P_r = \frac{1-3\nu}{4 \pi r^2}, \quad P_\perp = \frac{1-2\nu}{8 \pi r^2}. 
\end{equation}
Indeed, since $w_r \equiv P_r/\rho= (1-3\nu)/\nu$, the entire self-similar family satisfies the temperature mimicry condition Eq.~\eqref{eq:temperaturemimicry}.

When the assumptions given above hold, if we further assume $\widehat{m}(r)$ is \textit{continuous}, then this family has the temperature and coarse-grained entropy of a black hole of mass $M$: 
\begin{equation}
    S=4 \pi m^2+\int_m^M \frac{d \widehat{m}}{T_{\infty}\left(\widehat{m}\right)}=4 \pi m^2+8 \pi \int_m^M \widehat{m} \ d \widehat{m}=4 \pi M^2.
\end{equation}
Combining Eq.~\eqref{eq:Sshell}, Eq.~\eqref{eq:totalentropy}, and Eq.~\eqref{eq:selfsimilargas}, we find the entropy is distributed between the self-similar gas and the (massless)\footnote{These are massless because $\widehat{m}(r)$ is assumed to be continuous. Negative entropies appeared also in \cite{Brustein:2023hic}.} Israel layers as:
\begin{equation}
    S = \frac{A_{\mathrm H}}{4} - 4 \pi m(3m-R_{\mathrm{in}}) + 4 \pi M(3M-R_{\mathrm{out}})+4 \pi (1+w_r) (M^2-m^2) = 4 \pi M^2.
    \label{eq:entropylayers}
\end{equation}
These are the coarse-grained entropies of, respectively, (1) the inner black hole horizon, (2) the inner Israel layer, (3) the outer Israel layer, and (4) the self-gravitating gas between the Israel layers, where $R_\mathrm{in}=m/\nu$ and $R_{\mathrm{out}}=M/\nu$. 

This one-parameter family of self-similar solutions with $T=(8 \pi M)^{-1}$ and $S = 4 \pi M^2$ includes the $\delta=2$ line in Fig.~\ref{fig:classesofmatter}. The assumption that $\widehat{m}(r)$ is continuous imposes a constraint on the IJC; before turning to this caveat, we proceed to discuss how this coarse-grained entropy is distributed between the self-gravitating gases and the Israel layers. 

The stiffest star \cite{Banks:2002fj} satisfies $\nu=1/4$ and $w_r=1$, so the self-similar gas has $S_{\mathrm{gas}}= 8 \pi (M^2-m^2)$. The outer wall at $R_{\mathrm{out}}=4M$ has $S_{\mathrm{out}}=-4 \pi M^2$, and the inner wall at $R=4 m$ has $S_{\mathrm{in}}=4 \pi m^2$. Now turning to the frozen star \cite{Brustein:2018web,Brustein:2021lnr,Brustein:2023hic}, it satisfies $\nu=1/2$ and $w_r=-1$. In this case $sT = \rho + P_r =0$ by Eq.~\eqref{eq:anisothermo}. The self-similar gas carries no coarse-grained entropy. The walls sit at their local horizons $R_{\mathrm{in}}=2m$ and $R_{\mathrm{out}}=2M$; the inner wall cancels the horizon's area term, while the outer wall carries $S_{\mathrm{out}}=4 \pi M^2$.\footnote{
\label{fn:Frozen}This is consistent with \cite{Brustein:2023hic}, which argued that the frozen star satisfies $T=(8 \pi M)^{-1}$ and $S=4 \pi M^2$. Note the horizonless limit. The surface pressure $p\gg \sigma =0$ matches the $p_\perp \gg \rho$ transition region of \cite{Brustein:2023hic}.}

The HBH is the only \textit{aligned} member of this family, in the sense discussed in Sec.~\ref{sec:HBHintro}. The coarse-grained entropy of the ocean is $4 \pi (M^2-m^2)$. There are no Israel layers at the inner and outer edges of the HBH because $rf'/2f=1$, which is consistent with their vanishing coarse-grained entropies in Eq.~\eqref{eq:totalentropy}. For thermodynamic mimicry to hold for the other $\delta=2$ members of this family, the IJCs in Eq.~\eqref{eq:p_over_sigma} require massless walls, but $\sigma=0$ and $p = -R P_r/2 \sqrt{j}$ severely violates the DEC: $\sigma \ge |p|$ unless $P_r=0$. 

Massive walls can be considered, where $\sigma \ge 0$ is a free parameter, but then thermodynamic mimicry breaks down because Eq.~\eqref{eq:specificheatchangetemp} vanishes only for infinitesimally light shells at the photon sphere, i.e. $\sigma = 2p$. This concurs with Fig.~\ref{fig:BLP} in \cite{Brady:1991np}. There does remain the possibility of massive walls that are not thin shells. However, if $T_\infty =(8 \pi \widehat{m})^{-1}$, then $T_\infty \widehat{m}$ is constant, but at the interface of the Schwarzschild region we have, from Eq.~\eqref{eq:tempformulafromshell}:
\begin{equation}\label{eq:mimicbreakdown}
    d \log (T_\infty \widehat{m}) =\frac{R-3 \widehat{m}}{\widehat{m}(R-2 \widehat{m})} d \widehat {m} \ne 0, \quad \mathrm{unless} \quad R=3 \widehat{m}.
\end{equation}
In conclusion, massive walls break thermodynamic mimicry unless $R = 3 \widehat{m}$, but massless walls with non-vanishing transverse pressure violate DEC. This is not an issue for the HBH. Its photon sphere is \textit{aligned}; indeed, Eq.~\eqref{eq:mimicbreakdown} vanishes for the HBH because $\widehat{m}=r/3$.

\section{Discussion}

The well-known obstruction \cite{York:1986it,Brady:1991np} to placing a black hole in thermal equilibrium with self-gravitating radiation has been revisited. The thought experiment of Fig.~\ref{fig:thoughtexperiment} builds up a \textit{hillingar} black hole (HBH) \cite{Riojas:2026ytt,Riojas:2026mfu} by placing layers at $r=3\widehat{m}$, maintaining the photon sphere geometry through backreaction, thereby constructing a marginally stable exception.

The HBH has been uniquely selected by the joint stability constraints \cite{York:1986it,Brady:1991np}, and is an optical \cite{Riojas:2026ytt} and thermodynamic \cite{Riojas:2026mfu} mimic of an ordinary Schwarzschild black hole. It seems unlikely that these properties would arise together if the solution lacks a microphysical realization of any kind. A key question is whether a microscopic model can be constructed, or ruled out; ruling it out would more fully establish that self-gravitating matter cannot reach stable equilibrium with a black hole in asymptotically flat spacetime. 

Independently of the eventual dynamical status of the HBH, this work established several formal results relating photon spheres, the Israel junction conditions (IJCs), and finite-radius thermodynamics. We determined when extended photon spheres form (Eq.~\eqref{eq:extendedphotonsphereconditionmain}) and showed that the zero radial pressure limit of the anisotropic TOV equation is equivalent to the IJCs. We found the inverse specific heat at the photon sphere is proportional to $-\Lambda$, showed that forming shells cools black holes iff the specific heat is positive, and described how to compute coarse-grained entropies without the Euclidean path integral by tracking the asymptotic temperature. The stability arguments in \cite{Brady:1991np} were also generalized and extended to AdS. None of these results are conditional on the stability of the HBH.

In Sec.~\ref{sec:beyondnested} we found that a one-parameter self-similar family of solutions mimics a Schwarzschild black hole thermodynamically. Each is an extended photon sphere inside a spherical box around a black hole. The distributions of their coarse-grained entropies between the gas and walls vary throughout this family, which includes ``frozen stars" \cite{Brustein:2018web,Brustein:2021lnr,Brustein:2023hic} and ``stiffest stars" \cite{Banks:2002fj}. The HBH, the exception to the constraints \cite{York:1986it,Brady:1991np}, also happens to be the only member of this family without walls. Indeed, it satisfies $\sigma = 2p$ (as in Fig.~\ref{fig:thoughtexperiment}), and is \textit{aligned} as discussed in Sec.~\ref{sec:HBHintro} and Sec.~\ref{sec:beyondnested}. The HBH therefore satisfies thermodynamic mimicry without violating the dominant energy condition (DEC).

The HBH should be viewed as a toy model for a more realistic system, such as an equilibrium between Hawking quanta within the thermal atmosphere of a black hole and its photon sphere. Although the formal assumption of thermal equilibrium is strong, an explicit Bogoliubov computation in Appendix \ref{sec:bogoliubov} shows that the ocean redshifts the horizon so that $T_\infty=1/8\pi M$, as required for thermal equilibrium, while Eq.~\eqref{eq:specificheatchangetemp} gives the same conclusion. Still, wavelengths of modes within an ocean in equilibrium with the black hole will be of order
$M$, which seems inconsistent with a naive picture of particles on orbits. Rotation also seems likely to destabilize the ocean. For additional discussion, see \cite{Riojas:2026ytt,Riojas:2026mfu}.

Through Einstein's equations we may easily imagine a \textit{hillingar} black hole, and consider its properties. These systems may or may not exist in nature, but their construction poses formal questions that may be within reach. For example, if Eq.~\eqref{eq:temperaturemimicry} holds for a system, then $T_\infty=(8 \pi M)^{-1}$; in local thermodynamic equilibrium, Eqs.~\eqref{eq:temperaturemimicry} and \eqref{eq:anisothermo} give:
\begin{equation}
     T(r) = \frac{\hbar }{8 \pi \widehat{m}(r) \sqrt{1 - \frac{2 \widehat{m}(r)}{r}}}, \quad s(r)=\frac{2 \widehat{m}^{\prime}(r)}{G \hbar r}\left(1-\frac{2 \widehat{m}(r)}{r}\right)^{3 / 2}, \quad \widehat{m}(r) \equiv G m(r).
\end{equation}
For self-similar gases $\widehat{m} \propto r$ gives $s \sim 1/G\hbar r$ \cite{Riojas:2026mfu}: the $1/G\hbar$ of black hole entropy, in matter outside the horizon, satisfying Eq.~\eqref{eq:extendedphotonsphereconditionmain}.  Whether a microphysical model exists for any self-similar system as cold as a black hole is an open question, and one worth answering.

\begin{acknowledgments}
  I thank Matt Strassler for extensive collaboration on the companion papers \cite{Riojas:2026ytt,Riojas:2026mfu} and for many essential discussions that shaped this work.\footnote{The metric of Eq.~\eqref{fjmetric0}, which appears in \cite{DiFilippo:2024poc,Maeda:2024tsg}, and horizonless in \cite{1974RSPSA.337..529F}, was studied in \cite{Riojas:2026mfu,Riojas:2026ytt}; the present note shows it is the uniquely determined solution to the York/BLP constraints in asymptotically flat spacetime \cite{York:1986it,Brady:1991np}. The thermodynamics were obtained using the Euclidean path integral and entropy maximization in \cite{Riojas:2026mfu}; Sec.~\ref{sec:EntropyandShells} recovers these results independently using a method developed here.}   I am grateful to Andreas Karch and Lisa Randall for their support, comments, and conversations as I explored these ideas, and to Daniel Jafferis for extensive discussions. I am also grateful to Håkan Andréasson, Akshay Ghalsasi, Rashmish Mishra, Sonia Paban, Matt Reece, Suvrat Raju, Andrew Strominger, and Hongji Wei for helpful conversations. I also thank the authors of \cite{DiFilippo:2024poc} for bringing their work to my attention and for helpful correspondence. Calculations used the Mathematica packages diffgeo.m (Headrick) \cite{Headrick:diffgeo} and OGRe \cite{Shoshany:2021iuc} (Shoshany). This work was supported by the \textbf{Gra}vity, \textbf{S}pacetime, and \textbf{P}article Physics (GRASP) Initiative at Harvard University.
\end{acknowledgments}

\appendix

\section{Extended Photon Spheres, Stability, and Optical Mimicry}
\label{sec:UCO}
Ultra-compact objects (UCOs) are defined as self-gravitating systems with a light ring; substantial discussion on the stability of these objects exists in the literature \cite{Cardoso:2019rvt,Bambi:2025wjx}.

It has been conjectured, following a key observation concerning stability by Keir \cite{Keir:2014oka}, that ultra-compact objects with light rings are black holes \cite{Cardoso:2014sna}. Long-lived modes are generically expected to appear for horizonless systems with light rings, which tend to appear in pairs \cite{Cardoso:2014sna,Cardoso:2016oxy,Cunha:2017qtt}. It was shown \cite{Cunha:2017qtt} that light rings appear in pairs for smooth, horizonless, stationary ultra-compact objects under certain stated assumptions. This instability may arise from low-frequency modes trapped between the outer and inner light rings \cite{Cardoso:2016oxy,Cunha:2022gde}; there has been discussion of ``echoes" in the literature, see for example \cite{Mark:2017dnq}.

It was pointed out by Hod \cite{Hod:2017zpi} that an interesting horizonless exception exists to the argument of \cite{Cunha:2017qtt}. The exception applies to UCOs satisfying the relation:
\begin{equation}
    8 \pi r_{ps}^2\left(\rho+P_{\mathrm{\perp}}\right)=1,
\end{equation}
where $r_{ps}$ is the radius of the light ring, $\rho$ is the density of the matter fields, and $P_{\perp}$ is the transverse pressure. Hod's argument applies to constant density stars, which are horizonless objects; he gives an example where certain stars saturating $R=3M$ possess a solitary light ring in the exterior region with $r_{ps} = 3M$.

This is related to the fact that the HBH has a single light ring at $r=3M$. One finds $V_{\mathrm{eff}}(r)=\ell^2 h(r)+\epsilon f(r)$; for null geodesics $\epsilon=0$, see for example \cite{Riojas:2026ytt}. This gives:
\begin{equation}
    V_{\mathrm{eff}}=h(r) l^2, \quad V'_{\mathrm{eff}}=h'(r) l^2, \quad V''_{\mathrm{eff}}=h''(r)l^2,
\end{equation}
where $h(r) \equiv f(r)/r^2$. The stability of the light ring at the photon sphere $r=r_{\mathrm{ps}}$ hinges on the sign of  $V''_{\mathrm{eff}}(r)$, which can be straightforwardly written as: 
\begin{equation}
\left.V''_{\mathrm{eff}}\right|_{r_{\mathrm{ps}}} = \frac{2 l^2 f(r_{\mathrm{ps}})}{r_{\mathrm{ps}}^3} \frac{d}{dr} \left( \frac{r f'(r)}{2 f(r)}\right),
\label{eq:Vprimeprimecon}
\end{equation}
where we used the condition $r f'/2f=1$ to simplify the expression at $r=r_{\mathrm{ps}}$. Now we bring in the formalism developed in Sec.~\ref{sec:anisotropictov}; Eq.~\eqref{eq:EErrwithj(r)} shows that:
\begin{equation}
    \frac{r f'(r)}{2 f(r)} = \frac{\widehat m + 4 \pi r^3 P_r - \frac{1}{3} \Lambda r^3}{r - 2 \widehat m - \frac{1}{3} \Lambda r^3}.
    \label{eq:rfprimecontov}
\end{equation}
Since $r f'/2f=1$ is satisfied at the light ring, Eq.~\eqref{eq:rfprimecontov} gives $\widehat{m}=\frac{1}{3}\left( r - 4\pi r^3 P_r \right)$. Meanwhile, in this case the condition of hydrostatic equilibrium in Eq.~\eqref{eq:hydrostaticequilibrium} requires: 
\begin{equation}
    P_r'(r) = - \frac{P_r + \rho}{r}-\frac{2 \left( P_r - P_\perp\right)}{r}.
    \label{eq:hydrostaticequilibriumshadowcon}
\end{equation}
Taking the derivative in Eq.~\eqref{eq:Vprimeprimecon}, and inserting Eq.~\eqref{eq:rfprimecontov},  Eq.~\eqref{eq:hydrostaticequilibriumshadowcon}, $\widehat{m}=\frac{1}{3}\left( r - 4\pi r^3 P_r \right)$, and $\widehat{m}'(r)=4 \pi r^2 \rho$ into the resulting expression, we find:
\begin{equation}
    \left.V_{\mathrm{eff}}^{\prime \prime}\right|_{r_{\mathrm{ps}}} = \frac{-6 f(r) l^2}{r_{\mathrm{ps}}^4} \left( \frac{1-8\pi r^2 \left( P_\perp + \rho\right)}{1 - r^2 \Lambda + 8 \pi r^2 P_r} \right)=0 \iff 8 \pi r^2(P_\perp + \rho)=1.
    \label{eq:Hodcon}
\end{equation}
This generalizes Hod's \cite{Hod:2017zpi} condition. Although the HBH is not horizonless, the entire self-similar family of Sec.~\ref{sec:beyondnested}, which includes the HBH, nonetheless satisfies this relation. 

In fact there exists an infinite family of solutions satisfying this condition throughout an extended region, in addition to the HBH. Noting that $\rho(r) = \widehat{m}'(r)/4 \pi r^2$, by enforcing the condition Eq.~\eqref{eq:hydrostaticequilibriumshadowcon}, an extended photon sphere exists precisely when:\footnote{A much simpler derivation is given in Appendix C of \cite{Riojas:2026ytt}.} 
\begin{equation}\label{eq:extendedphotonspherecondition}
    \boxed{P_\perp(r) = \frac{1- 2 \widehat{m}'(r)}{8 \pi r^2}, \quad P_r = \frac{r - 3\widehat{m}(r)}{4 \pi r^3}, \quad \rho = \frac{\widehat{m}'(r)}{4 \pi r^2}}
\end{equation}
This is closely related to thermodynamic mimicry, as mentioned in the main text. The self-similar family discussed there satisfies this relation. 

It was pointed out by Hod that null circular orbits are possible throughout extended regions when a continuous and \textit{isotropic} version of his condition $8 \pi r^2\left(P_{\perp}+\rho\right)=1$ is satisfied \cite{Hod:2025bsf}. This is the isotropic case of Eq.~\eqref{eq:extendedphotonspherecondition}; it retains only the $\nu=1/4$ member of the self-similar family considered here, plus an integration constant: $\rho=(16 \pi r^2)^{-1}+\alpha$ and $p=(16 \pi r^2)^{-1}-\alpha$, where the DEC requires $\alpha \ge 0$; see also \cite{Li:2025wmd}. These solutions are self-similar only for $\alpha=0$, where $p=\rho$ recovers the stiffest star \cite{Banks:2002fj}.

\subsubsection*{Discrete Shells that Mimic Optically are Mechanically Unstable}
To supplement our discussion of optical mimicry, we show that the height of the peaks of the effective potential $V_{\mathrm{eff}}(r)\equiv h(r) \ell^2$ for nested Israel layers is controlled by their stability function $G\left(1, \alpha_{+}, \alpha_{-}\right)$ \cite{Brady:1991np} and trace $S^a{ }_a$. Ascending peaks always imply mechanical instability; hence, Israel layers always break optical mimicry, except for the HBH.

Within the vacuum Schwarzschild regions, we have: 
\begin{equation}
    \frac{d \log h}{d \log r}=-2 \frac{r-3 \widehat{m}-4 \pi r^3 P_r}{r-2 \widehat{m}-\frac{1}{3} \Lambda r^3} = -2\left(\frac{r - 3 \widehat{m}}{r j(r)}\right).
\end{equation}
The denominator is positive outside the horizon, so $h'(r)$ changes sign at $r = 3 \widehat{m}$. For static solutions Israel layers are placed \textit{outside} the photon sphere of the Schwarzschild metric. Each vacuum annulus contributes at most one point where $h'(r)=0$; these occur when an Israel layer causes $\mu$ to exceed $1/3$, as in Fig.~\ref{fig:Qmusystem}.

Geodesics that penetrate a photon sphere at $r=3M$ can, in principle, return to asymptotic infinity if a higher peak in $V_{\mathrm{eff}}$ exists within $r<3M$. For the HBH each shell sits at precisely $r = 3 \widehat{m}$. This makes $V_{\mathrm{eff}}$ a plateau; no walls are required at the edges, and so a geodesic penetrating $r=3M$ cannot return. This plateau implies optical mimicry.

The continuum limit \cite{DiFilippo:2024poc,Maeda:2024tsg} was obtained, along with the plateau, from shells in \cite{DiFilippo:2024poc}. The ascending peaks in Figure 1 of that paper, which asks whether light rings can self-gravitate, are obtained from continuity of the induced metric. Their construction reaches the same  continuum limit by a different route, which we comment on below.

Let $M_{\pm}$ be the masses enclosed on either side of the shell. Note that shells placed at the light ring $R=3M_{-}$, as in \cite{DiFilippo:2024poc}, are not necessarily rings of light. A traceless shell requires $\alpha_+ \alpha_-=1/3$, and since $\alpha_+<\alpha_-$ this places it strictly between $3M_{-}$ and $3M_{+}$; see Eq.~\eqref{eq:nesting} and Eq.~\eqref{eq:crossing}. The discrete shells considered in \cite{DiFilippo:2024poc}, each placed at $R=3M_{-}$, have positive trace and are unstable under radial perturbations \cite{Brady:1991np}, as reviewed in Sec.~\ref{sec:BLP}. Discrete shells of light have strictly descending peaks, and break optical mimicry. The ascending peaks in Fig.~1 in \cite{DiFilippo:2024poc} are not in contradiction with this, as we now explain. 

It is more generally true that ascending peaks occur only for mechanically unstable shells with positive trace; physical shells with negative trace have strictly descending peaks. It follows that the only optical mimic that can be obtained from physically realizable, or at least stable, shells is the HBH; an analogous result holds for Einstein clusters.

By continuity of the induced metric, at the location $R_k$ of the $k$'th shell:
\begin{equation}
    h\left(R_k\right)=c_{k-1} \frac{\alpha_{-}^2}{R_k^2}=c_k \frac{\alpha_{+}^2}{R_k^2} \implies \frac{c_k}{c_{k-1}}=\frac{\alpha_{-}^2}{\alpha_{+}^2}.
\end{equation}
When they occur, peaks appear at or between shells, and just outside the final shell, with
\begin{equation}
    h_k^{\text {peak }}=\frac{c_k\left(1-\frac{2}{3}\right)}{9 M_k^2}=\frac{c_k}{27 M_k^2}.
\end{equation}
The ratio of successive peak heights is then
\begin{equation}
    \frac{h_k^{\text {peak }}}{h_{k-1}^{\text {peak }}}=\left(\frac{\alpha_{-}}{\alpha_{+}}\right)^2\left(\frac{M_{-}}{M_{+}}\right)^2=\frac{\alpha_{-}^2\left(1-\alpha_{-}^2\right)^2}{\alpha_{+}^2\left(1-\alpha_{+}^2\right)^2}.
    \label{eq:peakheightsratio}
\end{equation}
To determine if the peaks increase or decrease in height, we calculate
\begin{equation}
\begin{aligned}
    \alpha_{-}\left(1-\alpha_{-}^2\right)&-\alpha_{+}\left(1-\alpha_{+}^2\right)=\left(\alpha_{-}-\alpha_{+}\right)\left(1-\left(\alpha_{+}^2+\alpha_{-}^2+\alpha_{+} \alpha_{-}\right)\right) \\
    &= 4 \pi R \sigma G\left(1, \alpha_{+}, \alpha_{-}\right)+8 \pi R \alpha_{+} \alpha_{-}\left(5 \alpha_{+}^2 \alpha_{-}^2+3 \alpha_{+} \alpha_{-}+1\right) S^a{ }_a.
    \label{eq:peakdiff}
\end{aligned}
\end{equation}
We used the stability function $G(\lambda, \alpha_+, \alpha_-)$ of BLP \cite{Brady:1991np} in Eq.~\eqref{eq:stabBLP} to rewrite the expression, as well as the trace of the stress-energy tensor on the shell:
\begin{equation}
    S^a{ }_a=2 p-\sigma=\frac{\left(\alpha_{-}-\alpha_{+}\right)\left(1-3 \alpha_{+} \alpha_{-}\right)}{8 \pi R \alpha_{+} \alpha_{-}}.
    \label{eq:traceshell}
\end{equation}
For a finite-mass shell of photons, $S^a{}_a$ vanishes and $G$ is negative, so the peak  heights necessarily \textit{decrease}. However, Fig.~1 of \cite{DiFilippo:2024poc} shows peaks of increasing height; by Eq.~\eqref{eq:peakdiff} this requires $S^a{}_a>0$, which holds for their shells at $R=3M_{-}$. Note that mechanical stability implies non-positive trace \cite{Brady:1991np}; see Fig.~\ref{fig:BLP}. 

Take $G\left(1, \alpha_{+}, \alpha_{-}\right) \geq 0$. Then from Eq.~\eqref{eq:stabBLP}
\begin{equation}
    15 \alpha_{+}^3 \alpha_{-}^3+4 \alpha_{+}^2 \alpha_{-}^2 \geq \alpha_{+}^2+\alpha_{+} \alpha_{-}+\alpha_{-}^2 \geq 3 \alpha_{+} \alpha_{-},
    \label{eq:BLPsresult}
\end{equation}
where we used $\left(\alpha_{-}-\alpha_{+}\right)^2 \geq 0$, hence $\alpha_{+}^2+\alpha_{-}^2 \geq 2 \alpha_{+} \alpha_-$. It follows that: 
\begin{equation}
    15 \alpha_{+}^2 \alpha_{-}^2+4 \alpha_{+} \alpha_{-}-3=\left(3 \alpha_{+} \alpha_{-}-1\right)\left(5 \alpha_{+} \alpha_{-}+3\right) \geq 0,
\end{equation}
which implies $\alpha_+ \alpha_- \ge 1/3$, hence $S^a{}_a \le 0$ by Eq.~\eqref{eq:traceshell}. Then mechanical stability implies non-positive trace; this is one of the results of \cite{Brady:1991np}. But then Eq.~\eqref{eq:BLPsresult} implies $\alpha_{+}^2+\alpha_{+} \alpha_{-}+\alpha_{-}^2 \geq 1$, and so the peaks descend by Eq.~\eqref{eq:peakdiff}. Similarly, we have just seen that mechanical stability $G\left(1, \alpha_{+}, \alpha_{-}\right)>0$ implies $S^a{}_a<0$, so $G\ge0$ and $S^a{}_a>0$ cannot occur together; the peaks ascend only for mechanically unstable shells with positive trace.

On the other hand, both $G(1, \alpha_+, \alpha_-)$ and $S^a{}_a$ vanish together in the traceless continuum limit, which yields a plateau. This underlies the optical mimicry of the HBH \cite{Riojas:2026ytt}.

\subsection*{Regge-Wheeler Modes of the Ocean}

Stability of the ocean requires that neither the Regge-Wheeler equation (for axial perturbations), nor the Zerilli equation (for polar perturbations), admit unstable QNMs with $\mathrm{Im}(\omega)>0$. Both equations satisfy:
\begin{equation}
    \frac{d^2 Z}{d r_*^2}+\left[\omega^2-V(r)\right] Z=0.
\end{equation}
The tortoise coordinate $r_*$ satisfies $d r_*=d r / \sqrt{f j}$. For the Schwarzschild metric it was shown by Chandrasekhar that the Regge-Wheeler and Zerilli equations have identical spectra. It is interesting to consider when this holds more generally, see \cite{Weller:2026jqu}. For non-vacuum spacetimes we use the (inverse-Cowling) Regge-Wheeler potential of \cite{Kokkotas:1999bd,Ferrari:2011rb,Boonserm:2013dua,Kojima:1991np}:
\begin{equation}
    V_{R W}(r)=f(r)\left[\frac{\ell(\ell+1)}{r^2}-\frac{6 \widehat{m}(r)}{r^3}+4 \pi\left(\rho-P_r\right)\right] .
    \label{eq:RW}
\end{equation}
An Israel layer at the edge of the ocean would complicate matters; fortunately, the HBH has none.  Additionally, the Regge-Wheeler potential is constant throughout the ocean because $f(r) \propto r^2$ and $m(r) \propto r$, with $P_r = 0$:
\begin{equation}
    \left.V_0 \equiv V_{\mathrm{RW}}\right|_{\mathrm{ocean}}=\frac{3 \ell(\ell+1)-5}{81 M^2}.
\end{equation}

Since $3 \ell(\ell+1)-5 \geq 1$ for all $\ell \ge 1$, we have $V_0 > 0$. In the ocean, $\sqrt{f j}=r /(9 M)$; then $d r_*=\frac{9 M d r}{r}$, so the tortoise coordinate is $r_*=9 M \ln r+c_1$, with $c_1$ some constant. Then the tortoise coordinate is smooth and logarithmic within the ocean: 
\begin{equation}
    \left.\Delta \equiv r_*\right|_{3 M}-\left.r_*\right|_{3 m}=9 M \ln \left(\frac{M}{m}\right)>0 .
\end{equation}
Here we establish some simple results for axial perturbations of the ocean.

\begin{proposition}
    Consider an ``ocean" region containing gravitating matter, extending from $r_{in} < r < r_{out}$, with (1) no Israel layers at the interface. Suppose (2) the tortoise coordinate is smooth and logarithmic within that region, and (3) the RW potential is constant and non-negative within that region. Then subject to a WKB approximation at the edges of this region, the Regge-Wheeler equation has no unstable QNMs. 
\end{proposition}
By Assumption (1) there are no Israel layers, so the Regge-Wheeler equation applies. By Assumption (2), the tortoise coordinate is smooth and logarithmic within the ``ocean":
\begin{equation}
    \frac{dr_*}{dr} = \frac{\alpha}{r}  \implies \Delta \equiv r_*\left(r_{\mathrm{out}}\right)-r_*\left(r_{\mathrm{in}}\right)=\alpha \ln \frac{r_{\mathrm{out}}}{r_{\mathrm{in}}}>0 .
\end{equation}
By Assumption (3) the Regge-Wheeler potential is constant within this region, so $V_{RW}\equiv V_0$. Then the Regge-Wheeler equation is: 
\begin{equation}
    \frac{d^2 Z}{d r_*^2}+\left(\omega^2-V_0\right) Z=0.
\end{equation}
It is straightforward to show that the solution to the Regge-Wheeler equation takes the form:
\begin{equation}
    Z(r)=C r^{i k}+D r^{-i k}, \quad  \Omega^2\equiv\omega^2-V_0, \quad k\equiv\alpha \Omega.
\end{equation}
Now we obtain a condition on the QNMs satisfying the Regge-Wheeler equation, subject to the conditions (1), (2), and (3). It follows by direct calculation that: 
\begin{equation}
    \frac{Z^{\prime}}{Z}=\frac{i k}{r} \frac{1-R r^{-2 i k}}{1+R r^{-2 i k}} .
    \label{eq:logarithmicderivative}
\end{equation}
Here we have used the definitions: 
\begin{equation}
    R \equiv \frac{D}{C}, \quad x \equiv \frac{\omega}{\Omega}, \quad \varphi \equiv\left(\frac{r_{\mathrm {in }}}{r_{\mathrm {out }}}\right)^{2 i k}=e^{-2 i \Omega \Delta}.
\end{equation}
In asymptotically flat spacetime QNMs are strictly ingoing at the horizon and strictly outgoing at asymptotic infinity. The boundary conditions for the ocean are not fixed by this; we take:\footnote{Long-lived modes within the ocean can be an issue for horizonless UCOs; see \cite{Keir:2014oka,Cardoso:2014sna,Cunha:2017qtt,Hod:2017zpi,Cunha:2022gde}. The HBH satisfies Hod's condition in \cite{Hod:2017zpi}, and its flat effective potential seems unlikely to trap such modes. }
\begin{equation}
    Z \sim e^{-i \omega r_*} \quad\left(r=r_{\mathrm{in}}\right), \quad Z \sim e^{+i \omega r_*} \quad\left(r=r_{\mathrm{out}}\right).
\end{equation}
The mode is strictly ingoing toward the horizon at the bottom of the ocean, and strictly outgoing toward asymptotic infinity at the top of the ocean. This then simplifies the problem. Since $\frac{d r_*}{d r}=\frac{1}{\sqrt{f j}}$ is continuous at both boundaries, $Z^{\prime} / Z= \pm i \omega \frac{d r_*}{d r}$, which gives: 
\begin{equation}
    \left.\frac{Z^{\prime}}{Z}\right|_{r=r_{\mathrm{in}}} = -\frac{i \omega \alpha}{r_{\mathrm{in}}}, \quad \quad \left.\frac{Z^{\prime}}{Z}\right|_{r=r_{\mathrm{out}}} = \frac{i \omega \alpha}{r_{\mathrm{out}}}.
\end{equation}
Combined with Eq.~\eqref{eq:logarithmicderivative}, this leads to two matching conditions for the QNMs: 
\begin{equation}
    R r_{\mathrm{in}}^{-2 i k}=\frac{1+x}{1-x}, \quad \quad R r_{\mathrm{out}}^{-2 i k}=\frac{1-x}{1+x} .
\end{equation}
Subsequently eliminating $R$ yields the condition: 
\begin{equation}
    \left(\frac{r_{\mathrm{in}}}{r_{\mathrm{out}}}\right)^{2 i k} \frac{1+x}{1-x}=\frac{1-x}{1+x}.
\end{equation}
This QNM condition can be rewritten sharply in terms of the Möbius transformation $g(z)$: 
\begin{equation}
    \varphi=g(x)^2, \quad g(z) \equiv \frac{1-z}{1+z}.
    \label{eq:QNMcon}
\end{equation}
The properties of the Möbius transformation permit us to rule out unstable QNMs for Regge-Wheeler potentials of this form. 

\begin{lemma}
    \normalfont For any $z \in \mathbb{C}$, $|g(z)|<1$ if and only if $\text{Re}(z)>0$. The proof follows from direct computation: 
\begin{equation}
    |g(z)|^2=\frac{|1-z|^2}{|1+z|^2}=\frac{1-2 \operatorname{Re}(z)+|z|^2}{1+2 \operatorname{Re}(z)+|z|^2} .
\end{equation}
\end{lemma}
\begin{proposition}
    \normalfont No unstable QNMs exist with  $\operatorname{Im}(\omega)>0$, because $\varphi=g(x)^2$ is not satisfied. Here we implicitly assume the WKB approximation used previously is justified.

    \noindent\normalfont \textbf{Proof.} 
    It suffices to consider modes with $\mathrm{Re}(\omega)
    \ge0$ because Eq.~\eqref{eq:QNMcon} is invariant under $\omega \rightarrow - \omega^*$. Let $\omega = a+ i b$ with $a \ge 0 $ and $b >0$. Write $\Omega=\sqrt{\omega^2-V_0}= A+ i B$ for either choice of square root branch; we show the contradiction holds for both cases.

   Note that $|\varphi| = |e^{-2 i \Omega \Delta}|=e^{2 B \Delta}$, so $B>0$ implies $|\varphi|>1$ and $B<0$ implies $|\varphi|<1. $ 

    \medskip 
    \noindent\textbf{Case} $a>0$. From $\Omega^2=\omega^2-V_0$ with $V_0$ real: 

    \begin{equation}
\operatorname{Im}\left(\Omega^2\right)=2 A B=\operatorname{Im}\left(\omega^2\right)=2 a b \implies AB = ab.
    \end{equation}
    Then we have $\operatorname{Im}\left(\omega^2-V_0\right)=2 a b>0$. Both branches give $AB > 0$, hence $A$ and $B$ have the same sign. A short computation gives $\operatorname{Re}(x)=\frac{b\left(a^2+B^2\right)}{B|\Omega|^2}$.
    \begin{itemize}
        \item If $B>0$, then $|\varphi|>1$. However, $\mathrm{Re}(x)>0$, so $|g(x)^2|<1$. Contradiction.
        \item If $B<0$, then $|\varphi|<1$. However, $\mathrm{Re}(x)<0$, so $|g(x)^2|>1$. Contradiction. 
    \end{itemize}
    
\noindent \textbf{Case} $a=0$. Now we have $\omega = i b$ and $\omega^2 - V_0<0$. Write $\Omega= \pm i \kappa$ with $\kappa=\sqrt{b^2+V_0}>0$. Since $\varphi \equiv  e^{-2 i \Omega \Delta}=e^{\pm 2 \kappa \Delta}$,  it follows that $\varphi>1$ when $\Omega = +i \kappa$ and $\varphi <1$ when $\Omega = - i \kappa$. Note $x=\frac{\omega}{\Omega}=\frac{i b}{ \pm i \kappa}= \pm \frac{b}{\kappa}\in \mathbb{R}$;  $|x| = \frac{b}{\kappa}<1$, with $g(x) = \frac{1-x}{1+x}>0$. Since $g(0)=1$ and $g'(x)<0$, $\Omega = i \kappa$ implies $x>0$ and $g(x)^2<1$; $\Omega = - i \kappa$ implies $x<0$ and $g(x)^2>1$. Then $\varphi$ and $g(x)^2$ are on opposite sides of $1$, which contradicts $\varphi = g(x)^2$. 

\medskip

In each case the QNM condition $|\varphi|=|g(x)^2|$ cannot be satisfied for $\mathrm{Im}(\omega)>0$. $\Box$ 
\end{proposition}

\begin{proposition}
    The only geometry where Assumptions (1), (2) and (3) apply is the HBH. 
\end{proposition}

To show this result we need only Einstein's equations and the IJCs. The Regge-Wheeler potential Eq.~\eqref{eq:RW} in terms of $\widehat{m}'(r)=4 \pi r^2 \rho(r)$ reads: 
\begin{equation}
    V_{R W}=f(r)\left[\frac{\ell(\ell+1)}{r^2}-\frac{6 \widehat{m}(r)}{r^3}+\frac{\widehat{m}^{\prime}(r)}{r^2}-4 \pi P_r(r)\right] .
\end{equation}
We require that $V_{RW}$ is constant in the ocean region; since the $\ell(\ell+1)$ term is the only one carrying $\ell$ dependence, its coefficient must be constant: 
\begin{equation}
    \frac{f(r)}{r^2}=A, \quad \Rightarrow \quad f(r)=A r^2.
\end{equation}
Equivalently, for the RW to be constant within this region, the region must be an extended photon sphere. Since $f(r) \propto r^2$, Einstein's equations Eq.~\eqref{eq:EErrwithj(r)} with $\Lambda=0$ require: 
\begin{equation}
    1 = \frac{r f^{\prime}}{2 f}=\frac{\widehat{m}(r)+4 \pi r^3 P_r(r)}{r-2 \widehat{m}(r)} \implies  4 \pi r^2 P_r=1-\frac{3 \widehat{m}}{r} .
\end{equation}
Inserted into the Regge-Wheeler equation, this gives: 
\begin{equation}
    V_{R W}=A\left[\ell(\ell+1)-1+\widehat{m}^{\prime}(r)-\frac{3 \widehat{m}(r)}{r}\right] .
\end{equation}
Then for $V_{RW}$ to be constant, $\widehat{m}(r)$ takes the form: 
\begin{equation}
    \widehat{m}^{\prime}(r)-\frac{3 \widehat{m}(r)}{r}= \text{const}. \implies \widehat{m}(r)=C r^3+a r .
\end{equation}
Now we enforce the requirement that there are no Israel layers. At a static interface $r=R$, the IJCs in Eq.~\eqref{eq:Ktautau} and Eq.~\eqref{eq:Kthetatheta} require:\footnote{We note in passing that the useful relations, equivalent to Eq.~\eqref{eq:p_over_sigma}, follow from Eq.~\eqref{eq:extrinsiccomponents}:
\begin{equation}\label{eq:anotherIJCcon}
    8 \pi \sigma=-\frac{2}{R}\left[\sqrt{j}\right]_{R}, \quad \quad 8 \pi(\sigma+p)=\frac{1}{R}\left[\sqrt{j}\left(\frac{r f^{\prime}}{2 f}-1\right)\right]_R.
\end{equation}
The notion of ``alignment" in the main text can be understood as following from the thought experiment in Fig.~\ref{fig:thoughtexperiment} ($\sigma = np$ for $n=2$), or, equivalently, from the fact that $rf'/2f=1$ at the interface for the HBH.}
\begin{equation}\label{eq:extrinsiccomponents}
    K^\theta{ }_\theta=\frac{\sqrt{j}}{R}, \quad K^t{ }_t=\frac{\sqrt{j}}{2} \frac{f^{\prime}}{f} .
\end{equation}
Here we let $\epsilon=1$ without loss of generality. But since $f(r)=A r^2$, the $K^t{}_t$ and $K^\theta {}_\theta$ components of the extrinsic curvature are equal throughout the ocean: 
\begin{equation}
    K^t{ }_t=\frac{\sqrt{j}}{R}=K^\theta{ }_\theta .
\end{equation}
Then we require continuity of $j(r)$, equivalently continuity of $m(r)$, at the two interfaces: $ \widehat{m}(r_{\mathrm{in}})=m$, and $\widehat{m}(r_{\mathrm{out}})=M$. The exterior of the ocean region must be Schwarzschild by Birkhoff's theorem; the extrinsic curvature at the exterior edge reads: 
\begin{equation}
    K^\theta{ }_\theta=\frac{\sqrt{1-2M/R}}{R}, \quad K^t {}_t=\frac{M}{R^2 \sqrt{1-2 M / R}}.
\end{equation}
Within the ocean the extrinsic curvatures are equal; since there is no Israel layer, we must have $K^t{}_t = K^\theta {}_\theta$ on the Schwarzschild side as well, which gives: 
\begin{equation}
    R = 3M. 
\end{equation}
Then we have: 
\begin{equation}
    j(r)=1-\frac{2 \widehat{m}(r)}{r} \implies \widehat{m}(r) = \frac{r}{3}.
\end{equation}
Then: 
\begin{equation}
    4 \pi r^2 P_r=1-\frac{3 \widehat{m}}{r} \implies P_r = 0.
\end{equation}
This completes the argument, since we have already seen that $P_r=0$ with $r f'/2f=1$ uniquely determines the HBH.

\section{The Euclidean Action and Thermodynamics}
\subsection*{Euclidean Action from Shells}
Here we perform a consistency check. Consider the Euclidean action:
\begin{equation}
    I_E=-\frac{1}{16 \pi G} \int_{\mathcal{M}} \sqrt{g} R-\frac{1}{8 \pi G} \int_{\partial \mathcal{M}} \sqrt{\gamma} K.
\end{equation}
Decompose the Euclidean geometry into $N$ vacuum annuli separated by Israel shells at radii $R_1 < \dots <R_N$, with Schwarzschild masses $m_0 < m_1 < \dots < m_N = M$. Between shells the bulk is vacuum, so $R=0$ and the Einstein-Hilbert term vanishes. The full action reduces to a sum of GHY boundary terms. 

Using $K$, see Eq.~\eqref{eq:K_static}, the integrated GHY contribution on a constant $r$ surface is: 
\begin{equation}
    -\frac{1}{8 \pi G} \int \sqrt{\gamma} K d^3 x=-\frac{\beta r f}{2}\left(\frac{r f^{\prime}}{2 f}+2\right).
\end{equation}
At each interface $R_k$, the inner annulus (mass $m_{k-1}$, period $\beta_{k-1})$ contributes an outward facing GHY term. The outer annulus (mass $m_k$, period $\beta_k$) contributes an inward-facing GHY term. Their mismatch is: 
\begin{equation}
    \Delta_k=-\frac{R_k}{2G}\left[\beta_{k-1} f_{-}\left(\frac{R_k f_{-}^{\prime}}{2 f_{-}}+2\right)-\beta_k f_{+}\left(\frac{R_k f_{+}^{\prime}}{2 f_{+}}+2\right)\right].
\end{equation}

The induced metric on the shell must be continuous, so the proper circumference of the Euclidean time-circle matches across the interface:\footnote{For a static shell, with $\dot{t} = dt/d\tau$:
\begin{equation}\label{eq:inducedmetricmatching}
\dot{t}^2=\frac{j(R)+\dot{R}^2}{f(R) j(R)} =\frac{1}{f(R)} \quad \quad d\tau = dt_{-}\sqrt{f_{-}(R)} = d t_{+}\sqrt{f_{+}(R)} .
\end{equation}
The induced metric on the shell must match on both sides of the shell, so the proper circumference of the time-circle must agree. This recovers: 
\begin{equation}
    \beta_{-} \sqrt{f_{-}(R)}=\beta_{+} \sqrt{f_{+}(R)} \implies \frac{T_{+}}{T_{-}}=\sqrt{\frac{f_{+}(R)}{f_{-}(R)}}.
\end{equation}} 
\begin{equation}
    \beta_{k-1} \sqrt{f_{-}}=\beta_k \sqrt{f_{+}} \equiv \beta_*.
\end{equation}
Then:
\begin{equation}
    \Delta_k=-\frac{R_k \beta_*}{2G}\left[\sqrt{f_{-}}\left(\frac{R_k f_{-}^{\prime}}{2 f_{-}}+2\right)-\sqrt{f_{+}}\left(\frac{R_k f_{+}^{\prime}}{2 f_{+}}+2\right)\right]
\end{equation}
The static extrinsic curvature components are $K_\theta^\theta=\sqrt{f} / R$ and $K_\tau^\tau=f^{\prime} /(2 \sqrt{f})$, so:
\begin{equation}
    K=K_\tau^\tau+2 K_\theta^\theta=\frac{\sqrt{f}}{R}\left(\frac{R f^{\prime}}{2 f}+2\right).
\end{equation}
The GHY mismatch at the interface is: 
\begin{equation}
    \Delta_k=\frac{\beta_* R_k^2}{2G}[K].
\end{equation}
But by the Israel junction condition: 
\begin{equation}
    [K]=\frac{8 \pi G}{n} T_a^a,
\end{equation}
where $T^a {}_a=-\sigma + n p$ is the trace of the shell stress-energy. This gives:
\begin{equation}
    \Delta_k=\frac{4 \pi \beta_* R_k^2}{n} T_a^a=\frac{4 \pi \beta_* R_k^2}{n}\left(-\sigma_k+n p_k\right).
\end{equation}

The shell surface densities become the bulk densities in the continuum limit; that is, $\sigma_k \equiv \rho \delta R_k/\sqrt{j}$, $p_k \equiv P_\perp \delta R_k/\sqrt{j}$. Note also that $\beta_* \equiv \beta(R_k) \sqrt{f(R_k)}$. 

The volume element takes the form:
\begin{equation}
    \sqrt{g} d^{n+2} x=\frac{\sqrt{f}}{\sqrt{j}} r^n d \tau d r d \Omega_n,
\end{equation}
so the shell stress-tensor's trace is: 
\begin{equation}
    \frac{1}{2} \int T^\mu{ }_\mu \sqrt{g} d^{n+2} x=\frac{\beta \Omega_n}{2} \int \frac{\sqrt{f}}{\sqrt{j}} r^n\left(-\rho+n P_{\perp}\right) d r.
\end{equation}
Then, to sum up, we have:
\begin{equation}
    \lim _{k \rightarrow \infty} \sum_k \Delta_k=\frac{4 \pi}{n} \int \beta \sqrt{f} \frac{r^2}{\sqrt{j}}\left(-\rho+n P_{\perp}\right) d r=\frac{1}{2} \int T^\mu{ }_\mu \sqrt{g} d^4 x=-\frac{1}{16 \pi G} \int R \sqrt{g} d^4 x,
\end{equation}
which is the Einstein-Hilbert action for $n=2$.

\subsection*{Energy Fluctuations Near the Spinodal Point, and Non-Logarithmic Correction}
\label{sec:fluctuations}
For a black hole in thermal equilibrium with a cavity wall, the root-mean-square energy fluctuations $\Delta E$ appear to diverge when the cavity wall coincides with the photon sphere \cite{York:1986it}. The following calculation suggests these fluctuations are in fact not infinite, because the saddle-point approximation breaks down near the spinodal temperature. 

The divergence can be regulated because the saddles merge at third order in $F(E)$. Expanding to \textit{second} order in $E$ gives:
\begin{equation}
    F(E) = F(E_0) + F'(E_0)(E-E_0) + \frac{1}{2} F''(E_0) (E-E_0)^2 + \dots
\end{equation}
Near the saddle $F'(E) = 1 - T S'(E)=0$, which gives $S'(E)=1/T$. Then the quadratic term in the expansion satisfies $F''=-T S''=T^{-1} (dT/dE)=(T C_R)^{-1}$, where $C_R\equiv dE/dT$ is the heat capacity. In asymptotically flat spacetime, $C_R^{-1}$ vanishes when the cavity wall coincides with the photon sphere of the black hole. As reviewed in the introduction, see Fig.~\ref{fig:swallowtailYork}, this occurs when the branches merge at the minimum temperature $T_c$:
\begin{equation}
    M_c = \frac{r}{3}, \quad T_c = \frac{3 \sqrt{3}}{8 \pi r}, \quad E_c = r\left(1-\frac{1}{\sqrt{3}}\right), \quad F_c = r\left(1-\frac{\sqrt{3}}{2}\right).
\end{equation}
Define $\delta E \equiv E-E_c$ and $\epsilon \equiv \left(T-T_c\right)/T_c$. Expanding to \textit{third} order in $\delta E$ near $T_c$:
\begin{equation}
    F(E) = F(E_0) +\frac{1}{6} F'''(E_0) \ \delta E^3, \quad M^2 = \frac{r^2}{9} + \frac{2r}{3 \sqrt{3} }\delta E - \frac{1}{\sqrt{3} r} \delta E^3 .
\end{equation}
Since $M(E) = E-\frac{E^2}{2r}$, the free energy $F(E)$ near the minimum temperature is given by:
\begin{equation}
    F = F_c - \frac{\sqrt{3} r}{6}\epsilon - \epsilon \ \delta E + \frac{3}{2 r^2} \ \delta E^3.
\end{equation}
The partition function $Z$ can then be written as:
\begin{equation}
    Z \equiv  \int dE \ e^{- \beta F(E,T)} \simeq \int d (\delta E) e^{- \beta \left( F_c - \frac{\sqrt{3}r}{6}\epsilon \right)}e^{- \beta \left( -\epsilon \ \delta E + \frac{3}{2 r^2} \ \delta E^3\right)}.
\end{equation}
It is convenient to define an integration coordinate $x$ and auxiliary variable $\xi$, where:
\begin{equation}\label{eq:deltaEandx}
    \delta E \equiv \left( \frac{2 r^2}{3 \beta}\right)^{1/3} x, \quad \quad \xi \equiv \beta \epsilon \left( \frac{2 r^2}{3 \beta}\right)^{1/3}.
\end{equation}
Near the minimum temperature we have $0<\xi \ll 1$. Then:
\begin{equation}
    Z \simeq e^{- \beta \left( F_c - \frac{\sqrt{3} r}{6} \epsilon\right)}\left( \frac{2 r^2}{3 \beta}\right)^{\frac{1}{3}} \int dx \ e^{\xi x - x^3}.
\end{equation}
Note that $\beta \epsilon=\beta_c-\beta$. Then applying $S=\left(1-\beta \partial_\beta\right) \log Z$ gives:
\begin{equation}
    \left(1-\beta \partial_\beta\right)\left[\frac{\sqrt{3} r}{6} \beta_c-\beta\left(F_c+\frac{\sqrt{3} r}{6}\right)\right]=\frac{\sqrt{3} r}{6} \beta_c =4 \pi M_c^2 .
\end{equation}
\begin{equation}
    \left(1-\beta \partial_\beta\right) \frac{1}{3} \log \left(\frac{2 r^2}{3 \beta}\right)=\frac{1}{3} \log \left(\frac{2 r^2}{3 \beta}\right)+\frac{1}{3} .
\end{equation}
\begin{equation}
    \left.\left(1-\beta \partial_\beta\right) \log \left(\int_0^{\infty} d x e^{\xi x-x^3}\right)\right|_{\beta_c}=\log \left(\frac{1}{3} \Gamma\left(\frac{1}{3}\right)\right)+\beta_c\langle\delta E\rangle,
\end{equation}
where we used Eq.~\eqref{eq:deltaEandx} in the last line. This is a non-logarithmic correction:
\begin{equation}\label{eq:bigentropycorrection}
    S(E_c)=4 \pi M_c^2 + \left(\frac{128 \pi^2}{81}\right)^{1 / 3} \frac{\Gamma(2 / 3)}{\Gamma(1 / 3)} r^{4 / 3}+\frac{1}{6} \log (A_{PS}) +\mathcal{O}(1),
\end{equation}
at third order in $\delta E$. The energy fluctuation is computed below: 
\begin{equation}
    \left\langle(\Delta E)^2\right\rangle=\left\langle(\delta E)^2\right\rangle-\langle(\delta E)\rangle^2, \quad \quad \frac{d^n Z}{d \xi^n}\propto \int d x x^n e^{\xi x-x^3}
\end{equation}
\begin{equation}\label{eq:deltaE}
    \left\langle\left(\delta E^n\right)\right\rangle=\left(\tfrac{\int d x \exp \left(\xi x-x^3\right) x^n}{\int d x \exp \left(\xi x-x^3\right)}\right)\left(\frac{2r^2}{3\beta}\right)^{n / 3}, \quad \beta_c\langle\delta E\rangle=\left(\tfrac{128 \pi^2}{81}\right)^{1 / 3} \tfrac{\Gamma(2 / 3)}{\Gamma(1 / 3)} r^{4 / 3}.
\end{equation}
The stable black hole branch lies on the $E > E_c$ side of the spinodal transition, so we restrict to $\delta E \ge 0$, hence $x \ge 0$. These integrals can be expanded as $e^{\xi x}=\sum_{n=0}^{\infty} \frac{(\xi x)^n}{n!}$:
\begin{equation}
   \int_0^{\infty} d x x^m e^{\xi x-x^3}=\sum_{n=0}^{\infty} \frac{\xi^n}{n!} \int_0^{\infty} d x x^{m+n} e^{-x^3}=\frac{1}{3} \sum_{n=0}^{\infty} \frac{\xi^n}{n!} \Gamma\left(\frac{m+n+1}{3}\right).
\end{equation}
Only the scaling is important, because the integral is simply an $\mathcal{O}(1)$ number. Since $\beta\sim r$ and $E_c \sim r$, the energy fluctuations are suppressed as the system becomes large:
\begin{equation}
    \Delta E \equiv \sqrt{(\Delta E)^2} \sim r^{1 / 3} \implies \frac{\Delta E}{E_c }\sim A^{-1/3}.
\end{equation}
The key point here is that $\Delta E/E_c$ is suppressed at large $A$. In any case, it does not diverge. The ratio of the leading and subleading terms in Eq.~\eqref{eq:bigentropycorrection} scales as $A^{-1/3}$, which is consistent with energy fluctuations driving a large fluctuation in the horizon area.

For a shell at $r_{ps}$ we have $\xi=0$, and $\int_0^{\infty} d x x^m e^{-x^3}=\frac{1}{3} \Gamma\left(\frac{m+1}{3}\right)$. This gives:
\begin{equation}
    \langle\delta E\rangle=\left(\frac{2 r^2}{3 \beta_c}\right)^{1 / 3} \frac{\Gamma(2 / 3)}{\Gamma(1 / 3)}, \quad\left\langle(\Delta E)^2\right\rangle=\left(\frac{2 r^2}{3 \beta_c}\right)^{2 / 3} \frac{\Gamma(1 / 3)-\Gamma(2 / 3)^2}{\Gamma(1 / 3)^2}.
\end{equation}
Restoring units, these fluctuations are large:
\begin{equation}
    \left\langle(\Delta E)^2\right\rangle=m_P^2 c^4\left(\frac{\sqrt{3} r}{4 \pi \ell_P}\right)^{2 / 3} \frac{\Gamma(1 / 3)-\Gamma(2 / 3)^2}{\Gamma(1 / 3)^2} .
\end{equation}
Note that these fluctuations were computed on the \textit{stable} side of the black hole branch, where the thermal bath is within the photon sphere of the black hole. This gives energy fluctuations on the order of \textit{tons} for a solar mass black hole; we have not checked it independently of the Euclidean path integral.

\subsection*{Bogoliubov Coefficients in 1+1 Dimensions for the HBH}
\label{sec:bogoliubov}

In the usual 1+1 dimensional $s$-wave approximation, the HBH has the same Hawking spectrum as a Schwarzschild black hole of mass $M$; this holds for any interior mass $m$. This review has been included to set the normalization for the temperature of the HBH.

As is well-known, the Hawking temperature is set by the surface gravity $\kappa$ at the horizon. This can be understood in the standard way as follows. Starting with the general metric $ d s^2=-f(r) d t^2+j(r)^{-1} d r^2$, the surface gravity at the horizon takes the form: 
 \begin{equation}
     \kappa=\left.\frac{1}{2} \sqrt{f^{\prime}\left(r_{+}\right) j^{\prime}\left(r_{+}\right)}\right|_{r=2 m}.
 \end{equation}
 To determine the modes we expand in plane waves $e^{- i \omega t + i k r_*}$, where $\partial_t$ is the future-directed timelike Killing vector that defines positive frequency. The tortoise coordinate $r_*$ tracking the phase of this wave satisfies $dr_*/dr=1/\sqrt{f(r) j(r)}$ ; its light-cone coordinates are $u = t -r_*$ and $v = t+r_*$. The $(u,v)$ coordinates that describe these waves for asymptotic observers at infinity break down at the horizon; however, the horizon is regular in the local coordinate system of an infalling observer that we now consider.

 Expanding near that horizon, the local  metric takes the form:
 \begin{equation}
    d s^2 \approx-f^{\prime}\left(r_{+}\right)\left(r-r_{+}\right) d t^2+\frac{d r^2}{j^{\prime}\left(r_{+}\right)\left(r-r_{+}\right)}.
 \end{equation}
 The finite proper distance from the horizon is: 
 \begin{equation}
     \rho=\int_{r_{+}}^r \frac{d r^{\prime}}{\sqrt{j^{\prime}\left(r_{+}\right)\left(r^{\prime}-r_{+}\right)}}=\frac{2 \sqrt{r-r_{+}}}{\sqrt{j^{\prime}\left(r_{+}\right)}}.
 \end{equation}
 That gives Rindler coordinates:
 \begin{equation}
    d s^2=   \frac{-f^{\prime}\left(r_{+}\right)j^{\prime}\left(r_{+}\right) \rho^2}{4} d t^2+d \rho^2=-\kappa^2 \rho^2 d t^2+d \rho^2.
 \end{equation}
 It is well-known that Minkowski coordinates are exponentially related to Rindler: 
 \begin{equation}
     T=\rho \sinh (\kappa t), \quad X=\rho \cosh (\kappa t).
 \end{equation}
 The local Minkowski light-cone coordinates for an infalling Rindler observer are essentially the usual Kruskal coordinates:
 \begin{equation}
    V\equiv\frac{\sqrt{j^{\prime}\left(r_{+}\right)}}{2}(T+X)=e^{+\kappa v}, \quad U\equiv\frac{\sqrt{j^{\prime}\left(r_{+}\right)}}{2}(T-X)=-e^{-\kappa u}.
 \end{equation}
The reason for the Hawking effect is that Minkowski time is exponentially related to asymptotic time. This is closely related to the fact that the horizon is regular in the $(U,V)$ coordinates of an infalling observer, but not in the asymptotic coordinates:
 \begin{equation}
    \left.d s^2\right|_{r=r_+}= -d T^2+d X^2=-\frac{2}{\kappa} \sqrt{\frac{f^{\prime}\left(r_+\right)}{j^{\prime}\left(r_+\right)}} d U d V.
 \end{equation}

To find the modes we quantize, e.g., the Klein-Gordon equation for a scalar field $\psi$; in 1+1 this gives $\partial_u \partial_v \psi =\partial_U\partial_V \psi= 0$. Only the right-moving (outgoing) modes at the horizon head to asymptotic infinity. Infalling observers expand in $a$-modes with positive frequency $\nu$ in $U$; asymptotic observers use $b$-modes with positive frequency $\omega$ in $u$: 
 \begin{equation}
     \psi_R=\int_0^{\infty} \frac{d \nu}{2 \pi(2 \nu)^{1 / 2}}\left(a_\nu e^{-i \nu U}+a_\nu^{\dagger} e^{+i \nu U}\right)=\int_0^{\infty} \frac{d \omega}{2 \pi(2 \omega)^{1 / 2}}\left(b_\omega e^{-i \omega u}+b_\omega^{\dagger} e^{+i \omega u}\right).
 \end{equation}

The goal is to compute the expectation value of the number operator $b^\dagger_\omega b_{\omega'}$ in the $a$-vacuum $|\psi\rangle$ where the horizon is regular, so we find the annihilation operator $b_\omega$ for the asymptotic observer using textbook Fourier analysis: 
 \begin{equation}
    b_\omega=\sqrt{2 \omega} \int_{-\infty}^{\infty} \frac{d u}{2 \pi} e^{i \omega u} \psi_R =\int_0^{\infty} \frac{d \nu}{2 \pi}\left[\sqrt{\frac{\omega}{\nu}} \int_{-\infty}^{\infty} \frac{d u}{2 \pi} e^{i \omega u} e^{-i \nu U} a_\nu+\sqrt{\frac{\omega}{\nu}} \int_{-\infty}^{\infty} \frac{d u}{2 \pi} e^{i \omega u} e^{+i \nu U} a_\nu^{\dagger}\right]
 \end{equation}
 The prefactors for $a_\nu$ and $a_\nu^\dagger$ are the well-known Bogoliubov coefficients. The exponential map $U(u)$ converts plane waves into power laws, which do not have definite frequencies; the Mellin transform then leads to the phase difference responsible for the Hawking spectrum: 
 \begin{equation}
     \alpha_{\omega \nu}\equiv\sqrt{\frac{\omega}{\nu}} \int_{-\infty}^{\infty} \frac{d u}{2 \pi} e^{i \omega u} e^{-i \nu U(u)} = \frac{1}{2 \pi \kappa} \sqrt{\frac{\omega}{\nu}} \nu^{i \omega / \kappa} e^{\pi \omega / 2 \kappa} \Gamma\left(\frac{-i \omega}{\kappa}\right),
 \end{equation}
 \begin{equation}
     \beta_{\omega \nu}\equiv\sqrt{\frac{\omega}{\nu}} \int_{-\infty}^{\infty} \frac{d u}{2 \pi} e^{i \omega u} e^{+i \nu U(u)}=\frac{1}{2 \pi \kappa} \sqrt{\frac{\omega}{\nu}} \nu^{i \omega / \kappa} e^{-\pi \omega / 2 \kappa} \Gamma\left(\frac{-i \omega}{\kappa}\right).
 \end{equation}
The number operator for the asymptotic observer is: 
 \begin{equation}
     b_\omega^{\dagger} b_{\omega^{\prime}}=\int_0^{\infty} \frac{d \nu}{2 \pi} \frac{d \nu^{\prime}}{2 \pi}\left(\alpha_{\omega \nu}^* a_\nu^{\dagger}+\beta_{\omega \nu}^* a_\nu\right)\left(\alpha_{\omega^{\prime} \nu^{\prime}} a_{\nu^{\prime}}+\beta_{\omega^{\prime} \nu^{\prime}} a_{\nu^{\prime}}^{\dagger}\right).
 \end{equation}
 This integral sums over $\nu$-frequencies of outgoing $a$-modes, in the local Minkowski frame of an infalling observer. This notion of time is exponentially related to asymptotic time. In any case, these modes contribute to a single outgoing mode of frequency $\omega$. 
 
 By the adiabatic principle \cite{Polchinski:2016hrw} these infalling modes annihilate the $a$-vacuum: $a_\nu|\psi\rangle=0$. Then since $\langle\psi| a_{\nu^{}} a_{\nu^{\prime}}^{\dagger}|\psi\rangle=2 \pi \delta\left(\nu-\nu^{\prime}\right)$, the expectation value of the number operator is: 
 \begin{equation}
     \langle\psi| b_\omega^{\dagger} b_{\omega^{\prime}}|\psi\rangle=\int_0^{\infty} \frac{d \nu}{2 \pi} \frac{d \nu^{\prime}}{2 \pi} \beta_{\omega \nu}^* \beta_{\omega^{\prime} \nu^{\prime}}\langle\psi| a_\nu a_{\nu^{\prime}}^{\dagger}|\psi\rangle =\int_0^{\infty} \frac{d \nu}{2 \pi} \beta_{\omega \nu}^* \beta_{\omega^{\prime} \nu}=\frac{2 \pi \delta\left(\omega-\omega^{\prime}\right)}{e^{2 \pi \omega / \kappa}-1}.
 \end{equation}
The Hawking spectrum is determined from the surface gravity $\kappa$, which sets the scale for the exponential map. For the HBH we have $\kappa=1 /(4M)$, the same as for a Schwarzschild black hole of mass $M$: a non-trivial result. 

This is due to the redshift of the near-horizon geometry of an ordinary Schwarzschild black hole of mass $m$: 
 \begin{equation}
     d s^2= -\frac{2}{\kappa} \sqrt{\frac{f^{\prime}(r_+)}{j^{\prime}(r_+)}} d U d V = - 8 m \ dU dV.
 \end{equation}
We emphasize that this reproduces the asymptotic temperature $T_\infty = 1/8 \pi M$ obtained from shells in Eq.~\eqref{eq:asymptempshell}; since both share the redshift factor $f(r)$, the local temperatures $T(r) = T_\infty/\sqrt{f(r)}$ match, both within the ocean and at asymptotic infinity. 

This calculation, as usual, does not reveal where Hawking radiation originates. If a significant fraction is emitted at or near the photon sphere, it would clarify the physical meaning of our results. This is plausible in the eikonal limit, where modes near the critical impact parameter in the spherical harmonics will be amplified. 

\section{Braneworld Cosmology and Photon Spheres}
\label{sec:RSbranes}
It is interesting to consider the influence of the photon sphere on braneworld cosmology. Branes move toward and away from photon spheres in analogy to ordinary particles.\footnote{From a braneworld cosmology perspective, an expanding universe is similar to a particle escaping from the photon sphere of a black hole. Similar comments apply in other cases.}

This methodology follows Per Kraus \cite{Kraus:1999it}, but similar methods have also appeared elsewhere in the literature. We begin with the usual Randall-Sundrum ansatz:
\begin{equation}
    d s^2=d y^2+e^{-2 A(y)} \bar{\gamma}_{a b}(x) d x^a d x^b.
\end{equation}
It was shown in Sec.~\ref{sec:IJC} that only pure tension branes are static. The induced metric is a warped product: $\gamma_{a b}(y, x)=e^{-2 A(y)} \bar{\gamma}_{a b}(x)$; then: 
\begin{equation}
    K_{a b}=\frac{1}{2} \partial_y \gamma_{a b}=-A^{\prime}(y) \gamma_{a b} \implies \left[K^\tau{ }_\tau\right]=\left[K^\theta{ }_\theta\right].
\end{equation}
Then for a static brane, it follows from Eq.~\eqref{eq:p_over_sigma} that: 
\begin{equation}
    \frac{p}{\sigma}=-\frac{n-1}{n}-\frac{1}{n} \frac{\left[K^\tau {}_\tau\right]}{\left[K^\theta{ }_\theta\right]} =-1 \implies p = - \sigma.
\end{equation}
The stress-energy tensor for moving branes is not pure tension. 

For the pure tension case we can straightforwardly recover the RS solution. Note that:
\begin{equation}
    K^a{ }_b=-A'(y) \delta^a{ }_b \Rightarrow K^\tau{ }_\tau=K^\theta{ }_\theta=-A'(y) .
\end{equation}
Then from Eq.~\eqref{eq:jump_Ktautau} and $\left[K^\theta{ }_\theta\right]=-[A']$:
\begin{equation}
    [A']=\left(8 \pi G_D\right) \frac{\sigma}{n} .
\end{equation}
The relevant bulk Einstein equations are:
\begin{equation}
    G_{y y}=6\left(A^{\prime}(y)\right)^2, \quad \quad 
G_{r r}=3 e^{-2 A(y)}\left(2\left(A^{\prime}(y)\right)^2-A^{\prime \prime}(y)\right).
\end{equation}
Then assuming the bulk is AdS$_5$, we have $G_{M N}=-\Lambda_5 g_{M N}$, and the $yy$ equation gives:
\begin{equation}
    6\left(A^{\prime}\right)^2=-\Lambda_5 \quad \Longrightarrow \quad\left(A^{\prime}(y)\right)^2= k^2 .
\end{equation}
The $rr$ equation can be written as $6\left(A^{\prime}\right)^2-3 A^{\prime \prime}=-\Lambda_5$. 

It follows immediately that $A^{\prime \prime}(y)=0$. Then in each bulk region:
\begin{equation}
    A^{\prime}(y)= \pm k \quad \Longrightarrow \quad A(y)= \pm k y+A_0 .
\end{equation}
Imposing $Z_2$ symmetry around a brane at $y=0$ gives the standard RS solution. The critical energy density for Randall-Sundrum follows from the $Z_2$ symmetry. 
\begin{equation}
    A(y)=k|y|+A_0, \quad \quad  2 k=\frac{8 \pi G_5}{3} \sigma \Rightarrow \sigma=\frac{3 k}{4 \pi G_5} .
    \label{eq:criticaltension}
\end{equation}
For $Z_2$ symmetric thin shells, the angular junction condition can be written as: 
\begin{equation}
    \dot{R}^2=\left(\frac{4 \pi G_5}{3} \sigma R\right)^2-j(R) .
    \label{eq:angularjunction}
\end{equation}
Defining $H \equiv \dot{R}/R$ gives the Friedmann equation on the brane: 
\begin{equation}
    H^2=\left(\frac{4 \pi G_5}{3} \sigma\right)^2-\frac{j(R)}{R^2}.
    \label{eq:Friedmangeneral}
\end{equation}
The time-dependence of this brane is given by:
\begin{equation}
\ddot{R}+\frac{f^{\prime}(R)}{2 f(R)}\left(j(R)+\dot{R}^2\right)-\frac{j^{\prime}(R)}{2 j(R)} \dot{R}^2=-\frac{4 \pi G_5}{3} \sigma R\left[4 \pi G_5\left(\frac{2}{3} \sigma+p\right)\right].
\end{equation}
These expressions depend on the metric functions $j(r)$ and $f(r)$. If we assume empty AdS$_5$, as in the RS case, then we have $f(r)=j(r)=k_s+\frac{r^2}{\ell^2}$, which gives:
\begin{equation}
    H^2=\left(\frac{4 \pi G_5}{3} \sigma\right)^2-\frac{1}{\ell^2}-\frac{k_s}{R^2}.
    \label{eq:FRW}
\end{equation}
Here we have $k_s=(-1, 0, 1)$, with $k_s=0$ for planar slicing. This is static for $k_s=0$ when $ \sigma=\sigma_0 \equiv \frac{3k}{4 \pi G_5}$. As usual $k=1/\ell$; this gives the same critical RS tension that was determined earlier in Eq.~\eqref{eq:criticaltension}. 

In fact, consistent with previous sections, the static brane for $k=1$ sits at the photon sphere at the boundary of AdS$_5$. For this case, Eq.~\eqref{eq:angularjunction} gives: 
\begin{equation}
    1+\frac{R^2}{\ell^2}=\left(\frac{4 \pi G_5}{3} \sigma R\right)^2 \Rightarrow 1=R^2\left[\left(\frac{4 \pi G_5}{3} \sigma\right)^2-\frac{1}{\ell^2}\right] .
\end{equation}
For the RS critical tension, the bracket vanishes and the only way to satisfy the equation is $R \rightarrow \infty$. There is a ``photon sphere" at this position, since: 
\begin{equation}
    \frac{R f^{\prime}(R)}{2 f(R)}=\frac{R\left(2 R / \ell^2\right)}{2\left(1+R^2 / \ell^2\right)}=\frac{R^2 / \ell^2}{1+R^2 / \ell^2} \xrightarrow{R \rightarrow \infty} 1.
\end{equation}
This brane sits at the AdS boundary. The $k=0$ and $k=1$ cases are different foliations; the $k=0$ coordinates cover the Poincaré patch. 

The critical brane tension $\sigma_0$ makes the brane static. For $k=0$ the same condition applies at all $R$; for $k \ne 0$ the pure tension brane can only be static at the photon sphere, where $r f'(r)/2f(r)=1$. 

It is useful to rewrite these equations in terms of $\sigma_0$. Let $\rho(\tau)$ be the energy density of matter on the brane, and let $p_m(\tau)$ be the pressure of matter on the brane. Then $\sigma_{\mathrm{tot}}=\sigma_0+\rho$ and $ p_{\mathrm{tot}}=-\sigma_0+p_m$. The stress-energy tensor is:
\begin{equation}
    T_{a b}=-\sigma_0 h_{a b}+T_{a b}^{(m)}, \quad \quad T_{a b}^{(m)}=(\rho+p_m) u_a u_b+p_m h_{a b}
\end{equation}
Here $T^{(m)}_{ab}$ is the stress tensor for the matter on the brane, which is conserved by the Bianchi identities. The conservation law $ \dot{\rho}+3 H(\rho+p_m)=0 $ gives the familiar expression: 
\begin{equation}
    \frac{d \ln \rho}{d \ln R}=-3\left(1+\frac{p_m}{\rho}\right).
\end{equation}
If we assume constant coefficients $p_m=w \rho$, brane energy conservation gives: 
\begin{equation}
    \frac{d \ln \rho}{d \ln R}=-3(1+w) \implies \rho(R)=\rho_0 R^{-3(1+w)}.
\end{equation}
For $1+w < 0$, the energy density on the brane grows as the universe expands. If $\rho=0$, then $\sigma_{tot}=\sigma_0$ and $p_{tot} = -\sigma_0$, which is a critical tension brane. It is static only at the photon sphere; for $k=0$ it is static for any $R$.

In this notation, the Friedmann equation becomes:
\begin{equation}
    H^2=\left(\frac{4 \pi G_5}{3}\right)^2\left(2 \sigma_0 \rho+\rho^2\right)-\frac{k_s}{R^2} .
\end{equation}
The usual Friedmann equation is recovered when $\rho \ll \sigma_0$. This defines $G_4$ as: 
\begin{equation}
    \frac{8 \pi G_4}{3} \rho \equiv\left(\frac{4 \pi G_5}{3}\right)^2\left(2 \sigma_0 \rho\right), \quad \quad H^2=\frac{8 \pi G_4}{3} \rho\left(1+\frac{\rho}{2 \sigma_0}\right)-\frac{k_s}{R^2} .
\end{equation}
Although $\sigma_0=\frac{3}{4 \pi G_5 \ell}$ in RS, it has a slightly different value for other metrics:
\begin{equation}
    \sigma_0 = \frac{3}{4 \pi G_5} \frac{\sqrt{j\left(R\right)}}{R}, \quad \quad G_4=\frac{G_5}{\ell}=G_5 k .
\end{equation}
This is fully consistent with the well-known expression \cite{Randall:1999ee,Randall:1999vf} in the limit $r_c \rightarrow \infty$:
\begin{equation}
    M_{P l}^2=M^3 r_c \int_{-\pi}^\pi d \phi e^{-2 k r_c|\phi|}=\frac{M^3}{k}\left[1-e^{-2 k r_c \pi}\right].
\end{equation}

If we include a black hole in the bulk AdS$_5$, the argument goes through mostly unchanged. The Friedmann equation Eq.~\eqref{eq:Friedmangeneral} for $f_{ \pm}(r)=j_{\pm}(r) = k_s+\frac{r^2}{\ell^2}-\frac{\mu_{ \pm}}{r^2}$ becomes: 
\begin{equation}
    H^2=\frac{8 \pi G_4}{3} \rho\left(1+\frac{\rho}{2 \sigma_0}\right)+\frac{\mu}{R^4}-\frac{k_s}{R^2}.
\end{equation}
The only difference is that a radiation term $\mu/R^4$ has been added to the equation. 

It is convenient to write:
\begin{equation}
    H^2=\left(\frac{4 \pi G_5}{3}\right)^2\left(2 \sigma_0 \rho+\rho^2\right)+\frac{\mu}{R^4}-\frac{k_s}{R^2} .
\end{equation}
Taking the derivative with respect to $\tau$, using $\dot{\rho}+3 H\left(\rho+p_m\right)=0 $ and $\dot{R} = H R$ gives:
\begin{equation}
    \dot{H}=-3 \left(\frac{4 \pi G_5}{3}\right)^2\left(\sigma_0+\rho\right)\left(\rho+p_m\right)-\frac{2 \mu}{R^4}+\frac{k_s}{R^2}
\end{equation}
The curvature term $k_s/R^2$ can give $\dot{H}>0$ without violating NEC. 

More generally this can be written as a radius dependent curvature term: 
\begin{equation}
    \dot{H}=-3\left(\frac{4 \pi G_5}{3}\right)^2\left(\sigma_0+\rho\right)\left(\rho+p_m\right)+\frac{j(R)}{R^2}\left(1-\frac{R j^{\prime}(R)}{2 j(R)}\right) .
\end{equation}
The curvature changes at the photon sphere: negative inside, positive outside. The curvature vanishes for $k=0$ slicing, causing $R$ dependence to drop out. Since $p_m = \rho w$:
\begin{equation}
    \dot{H}\left(R_{\mathrm{ps}}\right)=-3 \left(\frac{4 \pi G_5}{3}\right)^2\left(\sigma_0+\rho\right) \rho(1+w) .
\end{equation}
The usual FRW equation holds at low energy $\left(\sigma_0+\rho\right) \approx \sigma_0$, since $\frac{8 \pi G_4}{3} \rho \equiv\left(\frac{4 \pi G_5}{3}\right)^2\left(2 \sigma_0 \rho\right)$. 

Matter with $w<-1$ and $\rho>0$ gives $\dot{H}>0$, but only by violating NEC. For braneworlds, the geometric term in $\dot{H}$ vanishes at the higher-dimensional photon sphere. This is only a different interpretation of the usual circumstance. In any case, this geometric contribution can drive $\dot{H}>0$ even for brane matter satisfying $w \ge -1$.

\bibliographystyle{JHEP}
\bibliography{biblio}

@article{Andreasson:2007ck,
    author = "Andreasson, Hakan",
    title = "{Sharp bounds on 2m/r of general spherically symmetric static objects}",
    eprint = "gr-qc/0702137",
    archivePrefix = "arXiv",
    doi = "10.1016/j.jde.2008.05.010",
    journal = "J. Diff. Eq.",
    volume = "245",
    pages = "2243--2266",
    year = "2008"
}

@article{Witten:2024upt,
    author = "Witten, Edward",
    title = "{Introduction to black hole thermodynamics}",
    eprint = "2412.16795",
    archivePrefix = "arXiv",
    primaryClass = "hep-th",
    doi = "10.1140/epjp/s13360-025-06288-y",
    journal = "Eur. Phys. J. Plus",
    volume = "140",
    number = "5",
    pages = "430",
    year = "2025"
}

@article{York:1986it,
    author = "York, Jr., James W.",
    title = "{Black hole thermodynamics and the Euclidean Einstein action}",
    doi = "10.1103/PhysRevD.33.2092",
    journal = "Phys. Rev. D",
    volume = "33",
    pages = "2092--2099",
    year = "1986"
}

@article{Gibbons:1976ue,
    author = "Gibbons, G. W. and Hawking, S. W.",
    title = "{Action Integrals and Partition Functions in Quantum Gravity}",
    reportNumber = "PRINT-76-0995 (CAMBRIDGE)",
    doi = "10.1103/PhysRevD.15.2752",
    journal = "Phys. Rev. D",
    volume = "15",
    pages = "2752--2756",
    year = "1977"
}

@article{Hawking:1982dh,
    author = "Hawking, S. W. and Page, Don N.",
    title = "{Thermodynamics of Black Holes in anti-De Sitter Space}",
    reportNumber = "PRINT-83-0019 (CAMBRIDGE)",
    doi = "10.1007/BF01208266",
    journal = "Commun. Math. Phys.",
    volume = "87",
    pages = "577",
    year = "1983"
}

@article{Witten:1998zw,
    author = "Witten, Edward",
    editor = "Bergstrom, L. and Lindstrom, U.",
    title = "{Anti-de Sitter space, thermal phase transition, and confinement in gauge theories}",
    eprint = "hep-th/9803131",
    archivePrefix = "arXiv",
    reportNumber = "IASSNS-HEP-98-21",
    doi = "10.4310/ATMP.1998.v2.n3.a3",
    journal = "Adv. Theor. Math. Phys.",
    volume = "2",
    pages = "505--532",
    year = "1998"
}

@article{Brady:1991np,
    author = "Brady, P. R. and Louko, J. and Poisson, Eric",
    title = "{Stability of a shell around a black hole}",
    doi = "10.1103/PhysRevD.44.1891",
    journal = "Phys. Rev. D",
    volume = "44",
    pages = "1891--1894",
    year = "1991"
}

@article{Sorkin:1981wd,
    author = "Sorkin, Rafael D. and Wald, Robert M. and Zhang, Zhen Jiu",
    title = "{Entropy of selfgravitating radiation}",
    doi = "10.1007/BF00759862",
    journal = "Gen. Rel. Grav.",
    volume = "13",
    pages = "1127--1146",
    year = "1981"
}

@article{Sorkin:1981jc,
    author = "Sorkin, Rafael",
    title = "{A Criterion for the onset of instability at a turning point}",
    doi = "10.1086/159282",
    journal = "Astrophys. J.",
    volume = "249",
    pages = "254--257",
    year = "1981"
}

@article{Green:2013ica,
    author = "Green, Stephen R. and Schiffrin, Joshua S. and Wald, Robert M.",
    title = "{Dynamic and Thermodynamic Stability of Relativistic, Perfect Fluid Stars}",
    eprint = "1309.0177",
    archivePrefix = "arXiv",
    primaryClass = "gr-qc",
    doi = "10.1088/0264-9381/31/3/035023",
    journal = "Class. Quant. Grav.",
    volume = "31",
    pages = "035023",
    year = "2014"
}

@article{Brustein:2023hic,
    author = "Brustein, Ram and Medved, A. J. M. and Simhon, Tamar",
    title = "{Thermodynamics of frozen stars}",
    eprint = "2310.11572",
    archivePrefix = "arXiv",
    primaryClass = "gr-qc",
    doi = "10.1103/PhysRevD.110.024066",
    journal = "Phys. Rev. D",
    volume = "110",
    number = "2",
    pages = "024066",
    year = "2024"
}

@article{Kokkotas:1999bd,
    author = "Kokkotas, Kostas D. and Schmidt, Bernd G.",
    title = "{Quasinormal modes of stars and black holes}",
    eprint = "gr-qc/9909058",
    archivePrefix = "arXiv",
    doi = "10.12942/lrr-1999-2",
    journal = "Living Rev. Rel.",
    volume = "2",
    pages = "2",
    year = "1999"
}

@article{Ferrari:2011rb,
    author = "Ferrari, Valeria",
    title = "{Gravitational waves from perturbed stars}",
    eprint = "1105.1678",
    archivePrefix = "arXiv",
    primaryClass = "gr-qc",
    journal = "Bull. Astron. Soc. India",
    volume = "39",
    pages = "203",
    year = "2011"
}

@article{Boonserm:2013dua,
    author = "Boonserm, Petarpa and Ngampitipan, Tritos and Visser, Matt",
    title = "{Regge-Wheeler equation, linear stability, and greybody factors for dirty black holes}",
    eprint = "1305.1416",
    archivePrefix = "arXiv",
    primaryClass = "gr-qc",
    doi = "10.1103/PhysRevD.88.041502",
    journal = "Phys. Rev. D",
    volume = "88",
    pages = "041502",
    year = "2013"
}

@article{Kojima:1991np,
    author = "Kojima, Y. and Yoshida, S. and Futamase, T.",
    title = "{Nonradial pulsation of a boson star. 1: Formulation}",
    doi = "10.1143/PTP.86.401",
    journal = "Prog. Theor. Phys.",
    volume = "86",
    pages = "401--410",
    year = "1991"
}

@article{Cunha:2017qtt,
    author = "Cunha, Pedro V. P. and Berti, Emanuele and Herdeiro, Carlos A. R.",
    title = "{Light-Ring Stability for Ultracompact Objects}",
    eprint = "1708.04211",
    archivePrefix = "arXiv",
    primaryClass = "gr-qc",
    doi = "10.1103/PhysRevLett.119.251102",
    journal = "Phys. Rev. Lett.",
    volume = "119",
    number = "25",
    pages = "251102",
    year = "2017"
}

@article{Keir:2014oka,
    author = "Keir, Joe",
    title = "{Slowly decaying waves on spherically symmetric spacetimes and ultracompact neutron stars}",
    eprint = "1404.7036",
    archivePrefix = "arXiv",
    primaryClass = "gr-qc",
    doi = "10.1088/0264-9381/33/13/135009",
    journal = "Class. Quant. Grav.",
    volume = "33",
    number = "13",
    pages = "135009",
    year = "2016"
}

@article{Cardoso:2014sna,
    author = "Cardoso, Vitor and Crispino, Lu{\'\i}s C. B. and Macedo, Caio F. B. and Okawa, Hirotada and Pani, Paolo",
    title = "{Light rings as observational evidence for event horizons: long-lived modes, ergoregions and nonlinear instabilities of ultracompact objects}",
    eprint = "1406.5510",
    archivePrefix = "arXiv",
    primaryClass = "gr-qc",
    doi = "10.1103/PhysRevD.90.044069",
    journal = "Phys. Rev. D",
    volume = "90",
    number = "4",
    pages = "044069",
    year = "2014"
}

@article{Mark:2017dnq,
    author = "Mark, Zachary and Zimmerman, Aaron and Du, Song Ming and Chen, Yanbei",
    title = "{A recipe for echoes from exotic compact objects}",
    eprint = "1706.06155",
    archivePrefix = "arXiv",
    primaryClass = "gr-qc",
    reportNumber = "LIGO-P1700145",
    doi = "10.1103/PhysRevD.96.084002",
    journal = "Phys. Rev. D",
    volume = "96",
    number = "8",
    pages = "084002",
    year = "2017"
}

@article{Hawking:1975vcx,
    author = "Hawking, S. W.",
    editor = "Gibbons, G. W. and Hawking, S. W.",
    title = "{Particle Creation by Black Holes}",
    doi = "10.1007/BF02345020",
    journal = "Commun. Math. Phys.",
    volume = "43",
    pages = "199--220",
    year = "1975",
    note = "[Erratum: Commun.Math.Phys. 46, 206 (1976)]"
}

@article{Bekenstein:1972tm,
    author = "Bekenstein, J. D.",
    title = "{Black holes and the second law}",
    doi = "10.1007/BF02757029",
    journal = "Lett. Nuovo Cim.",
    volume = "4",
    pages = "737--740",
    year = "1972"
}

@article{Bekenstein:1974ax,
    author = "Bekenstein, Jacob D.",
    title = "{Generalized second law of thermodynamics in black hole physics}",
    doi = "10.1103/PhysRevD.9.3292",
    journal = "Phys. Rev. D",
    volume = "9",
    pages = "3292--3300",
    year = "1974"
}

@article{Israel:1966rt,
    author = "Israel, W.",
    title = "{Singular hypersurfaces and thin shells in general relativity}",
    doi = "10.1007/BF02710419",
    journal = "Nuovo Cim. B",
    volume = "44S10",
    pages = "1",
    year = "1966",
    note = "[Erratum: Nuovo Cim.B 48, 463 (1967)]"
}

@article{Oppenheimer:1939ne,
    author = "Oppenheimer, J. R. and Volkoff, G. M.",
    title = "{On massive neutron cores}",
    doi = "10.1103/PhysRev.55.374",
    journal = "Phys. Rev.",
    volume = "55",
    pages = "374--381",
    year = "1939"
}

@article{Tolman:1939jz,
    author = "Tolman, Richard C.",
    title = "{Static solutions of Einstein's field equations for spheres of fluid}",
    doi = "10.1103/PhysRev.55.364",
    journal = "Phys. Rev.",
    volume = "55",
    pages = "364--373",
    year = "1939"
}

@article{Tolman:1930zza,
    author = "Tolman, Richard C.",
    title = "{On the Weight of Heat and Thermal Equilibrium in General Relativity}",
    doi = "10.1103/PhysRev.35.904",
    journal = "Phys. Rev.",
    volume = "35",
    pages = "904--924",
    year = "1930"
}

@article{Buchdahl:1959zz,
    author = "Buchdahl, Hans A.",
    title = "{General Relativistic Fluid Spheres}",
    doi = "10.1103/PhysRev.116.1027",
    journal = "Phys. Rev.",
    volume = "116",
    pages = "1027",
    year = "1959"
}

@article{Bowers:1974tgi,
    author = "Bowers, Richard L. and Liang, E. P. T.",
    title = "{Anisotropic Spheres in General Relativity}",
    doi = "10.1086/152760",
    journal = "Astrophys. J.",
    volume = "188",
    pages = "657--665",
    year = "1974"
}

@article{Herrera:1997plx,
    author = "Herrera, L. and Santos, N. O.",
    title = "{Local anisotropy in self-gravitating systems}",
    doi = "10.1016/S0370-1573(96)00042-7",
    journal = "Phys. Rept.",
    volume = "286",
    pages = "53--130",
    year = "1997"
}

@article{Cardoso:2021wlq,
    author = "Cardoso, Vitor and Destounis, Kyriakos and Duque, Francisco and Macedo, Rodrigo Panosso and Maselli, Andrea",
    title = "{Black holes in galaxies: Environmental impact on gravitational-wave generation and propagation}",
    eprint = "2109.00005",
    archivePrefix = "arXiv",
    primaryClass = "gr-qc",
    doi = "10.1103/PhysRevD.105.L061501",
    journal = "Phys. Rev. D",
    volume = "105",
    number = "6",
    pages = "L061501",
    year = "2022"
}

@article{Brown:1992br,
    author = "Brown, J. David and York, Jr., James W.",
    title = "{Quasilocal energy and conserved charges derived from the gravitational action}",
    eprint = "gr-qc/9209012",
    archivePrefix = "arXiv",
    reportNumber = "IFP-423-UNC, TAR-009-UNC",
    doi = "10.1103/PhysRevD.47.1407",
    journal = "Phys. Rev. D",
    volume = "47",
    pages = "1407--1419",
    year = "1993"
}

@article{Maldacena:1997re,
    author = "Maldacena, Juan Martin",
    title = "{The Large $N$ limit of superconformal field theories and supergravity}",
    eprint = "hep-th/9711200",
    archivePrefix = "arXiv",
    reportNumber = "HUTP-97-A097, HUTP-98-A097",
    doi = "10.4310/ATMP.1998.v2.n2.a1",
    journal = "Adv. Theor. Math. Phys.",
    volume = "2",
    pages = "231--252",
    year = "1998"
}

@article{Randall:1999ee,
    author = "Randall, Lisa and Sundrum, Raman",
    title = "{A Large mass hierarchy from a small extra dimension}",
    eprint = "hep-ph/9905221",
    archivePrefix = "arXiv",
    reportNumber = "MIT-CTP-2860, PUPT-1860, BUHEP-99-9",
    doi = "10.1103/PhysRevLett.83.3370",
    journal = "Phys. Rev. Lett.",
    volume = "83",
    pages = "3370--3373",
    year = "1999"
}

@article{Randall:1999vf,
    author = "Randall, Lisa and Sundrum, Raman",
    title = "{An Alternative to compactification}",
    eprint = "hep-th/9906064",
    archivePrefix = "arXiv",
    reportNumber = "MIT-CTP-2874, PUPT-1867, BUHEP-99-13",
    doi = "10.1103/PhysRevLett.83.4690",
    journal = "Phys. Rev. Lett.",
    volume = "83",
    pages = "4690--4693",
    year = "1999"
}

@article{Bardeen:1972fi,
    author = "Bardeen, James M. and Press, William H. and Teukolsky, Saul A",
    title = "{Rotating black holes: Locally nonrotating frames, energy extraction, and scalar synchrotron radiation}",
    reportNumber = "OAP-288",
    doi = "10.1086/151796",
    journal = "Astrophys. J.",
    volume = "178",
    pages = "347",
    year = "1972"
}

@article{Andreasson:2021lsh,
    author = "Andr{\'e}asson, H{\r{a}}kan",
    title = "{Existence of Steady States of the Massless Einstein{\textendash}Vlasov System Surrounding a Schwarzschild Black Hole}",
    eprint = "2102.08170",
    archivePrefix = "arXiv",
    primaryClass = "gr-qc",
    doi = "10.1007/s00023-021-01104-6",
    journal = "Annales Henri Poincare",
    volume = "22",
    number = "12",
    pages = "4271--4297",
    year = "2021"
}

@article{Frauendiener_1990,
doi = {10.1088/0264-9381/7/4/011},
url = {https://doi.org/10.1088/0264-9381/7/4/011},
year = {1990},
month = {apr},
publisher = {},
volume = {7},
number = {4},
pages = {585},
author = {J Frauendiener and C Hoenselaers and W Konrad},
title = {A shell around a black hole},
journal = {Classical and Quantum Gravity}
}

@article{Hawking:1976de,
    author = "Hawking, S. W.",
    title = "{Black Holes and Thermodynamics}",
    doi = "10.1103/PhysRevD.13.191",
    journal = "Phys. Rev. D",
    volume = "13",
    pages = "191--197",
    year = "1976"
}

@article{Bekenstein:1973ur,
    author = "Bekenstein, Jacob D.",
    title = "{Black holes and entropy}",
    doi = "10.1103/PhysRevD.7.2333",
    journal = "Phys. Rev. D",
    volume = "7",
    pages = "2333--2346",
    year = "1973"
}

@article{Einstein:1939ms,
    author = "Einstein, Albert",
    title = "{On a stationary system with spherical symmetry consisting of many gravitating masses}",
    doi = "10.2307/1968902",
    journal = "Annals Math.",
    volume = "40",
    pages = "922--936",
    year = "1939"
}

@ARTICLE{1974RSPSA.337..529F,
       author = {{Florides}, P.~S.},
        title = "{A New Interior Schwarzschild Solution}",
      journal = {Proceedings of the Royal Society of London Series A},
         year = 1974,
        month = apr,
       volume = {337},
       number = {1611},
        pages = {529-535},
          doi = {10.1098/rspa.1974.0065},
       adsurl = {https://ui.adsabs.harvard.edu/abs/1974RSPSA.337..529F}
}

@article{Kraus:1999it,
    author = "Kraus, Per",
    title = "{Dynamics of anti-de Sitter domain walls}",
    eprint = "hep-th/9910149",
    archivePrefix = "arXiv",
    reportNumber = "EFI-99-43",
    doi = "10.1088/1126-6708/1999/12/011",
    journal = "JHEP",
    volume = "12",
    pages = "011",
    year = "1999"
}

@article{Karch:2000ct,
    author = "Karch, Andreas and Randall, Lisa",
    editor = "Duff, Michael J. and Liu, J. T. and Lu, J.",
    title = "{Locally localized gravity}",
    eprint = "hep-th/0011156",
    archivePrefix = "arXiv",
    reportNumber = "MIT-CTP-3099",
    doi = "10.1088/1126-6708/2001/05/008",
    journal = "JHEP",
    volume = "05",
    pages = "008",
    year = "2001"
}

@article{PhysRevD.25.330,
  title = {Instability of flat space at finite temperature},
  author = {Gross, David J. and Perry, Malcolm J. and Yaffe, Laurence G.},
  journal = {Phys. Rev. D},
  volume = {25},
  issue = {2},
  pages = {330--355},
  numpages = {0},
  year = {1982},
  month = {Jan},
  publisher = {American Physical Society},
  doi = {10.1103/PhysRevD.25.330},
  url = {https://link.aps.org/doi/10.1103/PhysRevD.25.330}
}

@article{Banks:2002fj,
    author = "Banks, T. and Fischler, W. and Kashani-Poor, A. and McNees, R. and Paban, S.",
    title = "{Entropy of the stiffest stars}",
    eprint = "hep-th/0206096",
    archivePrefix = "arXiv",
    reportNumber = "RUNHETC-2002-20, UTTG-04-02",
    doi = "10.1088/0264-9381/19/18/307",
    journal = "Class. Quant. Grav.",
    volume = "19",
    pages = "4717--4728",
    year = "2002"
}

@article{Andre:2021ctu,
    author = "Andr{\'e}, Rui and Lemos, Jos{\'e} P. S.",
    title = "{Thermodynamics of $d$-dimensional Schwarzschild black holes in the canonical ensemble}",
    eprint = "2101.11010",
    archivePrefix = "arXiv",
    primaryClass = "hep-th",
    doi = "10.1103/PhysRevD.103.064069",
    journal = "Phys. Rev. D",
    volume = "103",
    number = "6",
    pages = "064069",
    year = "2021"
}

@article{Hod:2017zpi,
    author = "Hod, Shahar",
    title = "{On the number of light rings in curved spacetimes of ultra-compact objects}",
    eprint = "1710.00836",
    archivePrefix = "arXiv",
    primaryClass = "gr-qc",
    doi = "10.1016/j.physletb.2017.11.021",
    journal = "Phys. Lett. B",
    volume = "776",
    pages = "1--4",
    year = "2018"
}

@article{Kim:2019ygw,
    author = "Kim, Hyeong-Chan and Lee, Youngone",
    title = "{Entropy of self-gravitating anisotropic matter}",
    eprint = "1901.03148",
    archivePrefix = "arXiv",
    primaryClass = "hep-th",
    doi = "10.1140/epjc/s10052-019-7189-2",
    journal = "Eur. Phys. J. C",
    volume = "79",
    number = "8",
    pages = "679",
    year = "2019",
    note = "[Erratum: Eur.Phys.J.C 79, 977 (2019)]"
}

@article{Cardoso:2016oxy,
    author = "Cardoso, Vitor and Hopper, Seth and Macedo, Caio F. B. and Palenzuela, Carlos and Pani, Paolo",
    title = "{Gravitational-wave signatures of exotic compact objects and of quantum corrections at the horizon scale}",
    eprint = "1608.08637",
    archivePrefix = "arXiv",
    primaryClass = "gr-qc",
    doi = "10.1103/PhysRevD.94.084031",
    journal = "Phys. Rev. D",
    volume = "94",
    number = "8",
    pages = "084031",
    year = "2016"
}

@article{Poincare:1886xr,
    author = "Poincare, H.",
    title = "{On the irregular integrals of linear equations}",
    doi = "10.1007/BF02417092",
    journal = "Acta Math.",
    volume = "8",
    pages = "295--344",
    year = "1886"
}

@ARTICLE{1972NPhS..240...77G,
       author = {{Gibbons}, G.~W.},
        title = "{On Lowering a Rope into a Black Hole}",
      journal = {Nature Physical Science},
         year = 1972,
        month = nov,
       volume = {240},
       number = {100},
        pages = {77},
          doi = {10.1038/physci240077a0},
       adsurl = {https://ui.adsabs.harvard.edu/abs/1972NPhS..240...77G}
}

@article{Brown:2012un,
    author = "Brown, Adam R.",
    title = "{Tensile Strength and the Mining of Black Holes}",
    eprint = "1207.3342",
    archivePrefix = "arXiv",
    primaryClass = "gr-qc",
    doi = "10.1103/PhysRevLett.111.211301",
    journal = "Phys. Rev. Lett.",
    volume = "111",
    number = "21",
    pages = "211301",
    year = "2013"
}

@article{Almheiri:2012rt,
    author = "Almheiri, Ahmed and Marolf, Donald and Polchinski, Joseph and Sully, James",
    title = "{Black Holes: Complementarity or Firewalls?}",
    eprint = "1207.3123",
    archivePrefix = "arXiv",
    primaryClass = "hep-th",
    doi = "10.1007/JHEP02(2013)062",
    journal = "JHEP",
    volume = "02",
    pages = "062",
    year = "2013"
}

@inproceedings{Polchinski:2016hrw,
    author = "Polchinski, Joseph",
    title = "{The black hole information problem.}",
    booktitle = "{Theoretical Advanced Study Institute in Elementary Particle Physics}: {New Frontiers in Fields and Strings}",
    eprint = "1609.04036",
    archivePrefix = "arXiv",
    primaryClass = "hep-th",
    doi = "10.1142/9789813149441_0006",
    pages = "353--397",
    year = "2017"
}

@article{Hawking:1974rv,
    author = "Hawking, S. W.",
    title = "{Black hole explosions}",
    doi = "10.1038/248030a0",
    journal = "Nature",
    volume = "248",
    pages = "30--31",
    year = "1974"
}

@article{Cardoso:2019rvt,
    author = "Cardoso, Vitor and Pani, Paolo",
    title = "{Testing the nature of dark compact objects: a status report}",
    eprint = "1904.05363",
    archivePrefix = "arXiv",
    primaryClass = "gr-qc",
    doi = "10.1007/s41114-019-0020-4",
    journal = "Living Rev. Rel.",
    volume = "22",
    number = "1",
    pages = "4",
    year = "2019"
}

@article{Cardoso:2008bp,
    author = "Cardoso, Vitor and Miranda, Alex S. and Berti, Emanuele and Witek, Helvi and Zanchin, Vilson T.",
    title = "{Geodesic stability, Lyapunov exponents and quasinormal modes}",
    eprint = "0812.1806",
    archivePrefix = "arXiv",
    primaryClass = "hep-th",
    doi = "10.1103/PhysRevD.79.064016",
    journal = "Phys. Rev. D",
    volume = "79",
    number = "6",
    pages = "064016",
    year = "2009"
}

@article{Brustein:2018web,
    author = "Brustein, Ram and Medved, A. J. M.",
    title = "{Resisting collapse: How matter inside a black hole can withstand gravity}",
    eprint = "1805.11667",
    archivePrefix = "arXiv",
    primaryClass = "hep-th",
    doi = "10.1103/PhysRevD.99.064019",
    journal = "Phys. Rev. D",
    volume = "99",
    number = "6",
    pages = "064019",
    year = "2019"
}

@article{Martinez:1996ni,
    author = "Martinez, Erik A.",
    title = "{Fundamental thermodynamical equation of a selfgravitating system}",
    eprint = "gr-qc/9601037",
    archivePrefix = "arXiv",
    reportNumber = "CGPG-96-1-6",
    doi = "10.1103/PhysRevD.53.7062",
    journal = "Phys. Rev. D",
    volume = "53",
    pages = "7062--7072",
    year = "1996"
}

@article{Andre:2019zzo,
    author = "Andr{\'e}, Rui and Lemos, Jos{\'e} P. S. and Quinta, Gon{\c{c}}alo M.",
    title = "{Thermodynamics and entropy of self-gravitating matter shells and black holes in $d$ dimensions}",
    eprint = "1905.05239",
    archivePrefix = "arXiv",
    primaryClass = "hep-th",
    doi = "10.1103/PhysRevD.99.125013",
    journal = "Phys. Rev. D",
    volume = "99",
    number = "12",
    pages = "125013",
    year = "2019"
}

@article{Lemos:2023yiz,
    author = "Lemos, Jos{\'e} P. S. and Zaslavskii, Oleg B.",
    title = "{Black holes and hot shells in the Euclidean path integral approach to quantum gravity}",
    eprint = "2304.06740",
    archivePrefix = "arXiv",
    primaryClass = "hep-th",
    doi = "10.1088/1361-6382/ad0515",
    journal = "Class. Quant. Grav.",
    volume = "40",
    number = "23",
    pages = "235012",
    year = "2023"
}

@article{Boehmer:2007az,
    author = "Boehmer, C. G. and Harko, T.",
    title = "{On Einstein clusters as galactic dark matter halos}",
    eprint = "0705.1756",
    archivePrefix = "arXiv",
    primaryClass = "gr-qc",
    doi = "10.1111/j.1365-2966.2007.11977.x",
    journal = "Mon. Not. Roy. Astron. Soc.",
    volume = "379",
    pages = "393--398",
    year = "2007"
}

@article{Brustein:2021lnr,
    author = "Brustein, Ram and Medved, A. J. M. and Simhon, Tamar",
    title = "{Black holes as frozen stars}",
    eprint = "2109.10017",
    archivePrefix = "arXiv",
    primaryClass = "gr-qc",
    doi = "10.1103/PhysRevD.105.024019",
    journal = "Phys. Rev. D",
    volume = "105",
    number = "2",
    pages = "024019",
    year = "2022"
}

@inproceedings{Bambi:2025wjx,
    author = "Bambi, Cosimo and others",
    title = "{Black hole mimickers: from theory to observation}",
    eprint = "2505.09014",
    archivePrefix = "arXiv",
    primaryClass = "gr-qc",
    month = "5",
    year = "2025"
}

@article{Cunha:2022gde,
    author = "Cunha, Pedro V. P. and Herdeiro, Carlos and Radu, Eugen and Sanchis-Gual, Nicolas",
    title = "{Exotic Compact Objects and the Fate of the Light-Ring Instability}",
    eprint = "2207.13713",
    archivePrefix = "arXiv",
    primaryClass = "gr-qc",
    doi = "10.1103/PhysRevLett.130.061401",
    journal = "Phys. Rev. Lett.",
    volume = "130",
    number = "6",
    pages = "061401",
    year = "2023"
}

@article{Shoshany:2021iuc,
    author = "Shoshany, Barak",
    title = "{OGRe: An Object-Oriented General Relativity Package for Mathematica}",
    eprint = "2109.04193",
    archivePrefix = "arXiv",
    primaryClass = "cs.MS",
    doi = "10.21105/joss.03416",
    journal = "J. Open Source Softw.",
    volume = "6",
    pages = "3416",
    year = "2021"
}

@misc{Headrick:diffgeo,
    author       = "Headrick, Matthew",
    title        = "{diffgeo.m: A package for doing GR-type tensor algebra and calculus}",
    howpublished = "\url{https://sites.google.com/view/matthew-headrick/mathematica}",
    note         = "Unpublished Mathematica package",
}

@article{Riojas:2026mfu,
    author = "Riojas, Marcos and Strassler, Matthew J.",
    title = "{On Black Holes Surrounded by Radiation II: Thermodynamics}",
    eprint = "2606.30797",
    archivePrefix = "arXiv",
    primaryClass = "hep-th",
    month = "6",
    year = "2026"
}

@article{Riojas:2026ytt,
    author = "Riojas, Marcos and Strassler, Matthew J.",
    title = "{On Black Holes Surrounded by Radiation: I. Classical Considerations}",
    eprint = "2606.30794",
    archivePrefix = "arXiv",
    primaryClass = "hep-th",
    month = "6",
    year = "2026"
}

@unpublished{RiojasStrasslerAdS,
  author = {Marcos Riojas and Matthew J. Strassler},
  title = {In Preparation},
  note = {To appear},
  year = {2026}
}

@article{Maeda:2024tsg,
    author = "Maeda, Kei-ichi and Cardoso, Vitor and Wang, Anzhong",
    title = "{Einstein cluster as central spiky distribution of galactic dark matter}",
    eprint = "2410.04175",
    archivePrefix = "arXiv",
    primaryClass = "gr-qc",
    doi = "10.1103/PhysRevD.111.044060",
    journal = "Phys. Rev. D",
    volume = "111",
    number = "4",
    pages = "044060",
    year = "2025"
}

@article{Jusufi:2022jxu,
    author = "Jusufi, Kimet",
    title = "{Black holes surrounded by Einstein clusters as models of dark matter fluid}",
    eprint = "2202.00010",
    archivePrefix = "arXiv",
    primaryClass = "gr-qc",
    doi = "10.1140/epjc/s10052-023-11264-w",
    journal = "Eur. Phys. J. C",
    volume = "83",
    number = "2",
    pages = "103",
    year = "2023"
}

@ARTICLE{1970GReGr...1...19K,
       author = {{Kumar Datta}, Bidyut},
        title = "{Non-static spherically symmetric clusters of particles in general relativity: I}",
      journal = {General Relativity and Gravitation},
         year = 1970,
        month = mar,
       volume = {1},
       number = {1},
        pages = {19-25},
          doi = {10.1007/BF00759199},
       adsurl = {https://ui.adsabs.harvard.edu/abs/1970GReGr...1...19K}
}

@ARTICLE{1971GReGr...2..321B,
       author = {{Bondi}, Hermann},
        title = "{On Datta's spherically symmetric systems in general relativity}",
      journal = {General Relativity and Gravitation},
         year = 1971,
        month = dec,
       volume = {2},
       number = {4},
        pages = {321-329},
          doi = {10.1007/BF00758151},
       adsurl = {https://ui.adsabs.harvard.edu/abs/1971GReGr...2..321B}
}

@article{Joshi:2011zm,
    author = "Joshi, Pankaj S. and Malafarina, Daniele and Narayan, Ramesh",
    title = "{Equilibrium configurations from gravitational collapse}",
    eprint = "1106.5438",
    archivePrefix = "arXiv",
    primaryClass = "gr-qc",
    doi = "10.1088/0264-9381/28/23/235018",
    journal = "Class. Quant. Grav.",
    volume = "28",
    pages = "235018",
    year = "2011"
}

@article{DiFilippo:2024poc,
    author = "Di Filippo, Francesco and Rezzolla, Luciano",
    title = "{Can light-rings self-gravitate?}",
    eprint = "2407.13832",
    archivePrefix = "arXiv",
    primaryClass = "gr-qc",
    doi = "10.1103/PhysRevD.111.L021504",
    journal = "Phys. Rev. D",
    volume = "111",
    number = "2",
    pages = "L021504",
    year = "2025"
}

@article{Hod:2025bsf,
    author = "Hod, Shahar",
    title = "{Spherically symmetric curved spacetimes with a continuum of light rings}",
    eprint = "2504.00087",
    archivePrefix = "arXiv",
    primaryClass = "gr-qc",
    doi = "10.1088/1475-7516/2025/06/006",
    journal = "JCAP",
    volume = "06",
    pages = "006",
    year = "2025"
}

@article{Li:2025wmd,
    author = "Li, Long-Yue and Liu, Xia-Yuan and Cai, Rong-Gen and Gong, Yungui and Zhou, Wenting",
    title = "{Images and photon regions of continuous photon sphere spacetime}",
    eprint = "2508.12002",
    archivePrefix = "arXiv",
    primaryClass = "gr-qc",
    doi = "10.1016/j.dark.2025.102209",
    journal = "Phys. Dark Univ.",
    volume = "51",
    pages = "102209",
    year = "2026"
}

@article{Cunha:2025oeu,
    author = "Cunha, Pedro V. P.",
    title = "{Backreaction of perturbations around a stable light ring}",
    eprint = "2503.00117",
    archivePrefix = "arXiv",
    primaryClass = "gr-qc",
    doi = "10.1088/1475-7516/2025/05/083",
    journal = "JCAP",
    volume = "05",
    pages = "083",
    year = "2025"
}

@article{Zeldovich:1961sbr,
    author = "Zel'dovich, Ya. B.",
    title = "{The Equation of State at Ultrahigh Densities and Its Relativistic Limitations}",
    journal = "Zh. Eksp. Teor. Fiz.",
    volume = "41",
    pages = "1609--1615",
    year = "1961"
}

@article{Virbhadra:2002ju,
    author = "Virbhadra, K. S. and Ellis, G. F. R.",
    title = "{Gravitational lensing by naked singularities}",
    doi = "10.1103/PhysRevD.65.103004",
    journal = "Phys. Rev. D",
    volume = "65",
    pages = "103004",
    year = "2002"
}

@article{Weller:2026jqu,
    author = "Weller, Colin and Li, Dongjun and Wagle, Pratik and Laeuger, Andrew and Chen, Yanbei and Yunes, Nicol{\'a}s",
    title = "{Chiral symmetry and black hole isospectrality}",
    eprint = "2608.21532",
    archivePrefix = "arXiv",
    primaryClass = "gr-qc",
    month = "8",
    year = "2026"
}

@article{Claudel:2000yi,
    author = "Claudel, Clarissa-Marie and Virbhadra, K. S. and Ellis, G. F. R.",
    title = "{The Geometry of photon surfaces}",
    eprint = "gr-qc/0005050",
    archivePrefix = "arXiv",
    doi = "10.1063/1.1308507",
    journal = "J. Math. Phys.",
    volume = "42",
    pages = "818--838",
    year = "2001"
}

@article{Gondolo:1999ef,
    author = "Gondolo, Paolo and Silk, Joseph",
    title = "{Dark matter annihilation at the galactic center}",
    eprint = "astro-ph/9906391",
    archivePrefix = "arXiv",
    reportNumber = "MPI-PHT-99-10, OUAST-99-9",
    doi = "10.1103/PhysRevLett.83.1719",
    journal = "Phys. Rev. Lett.",
    volume = "83",
    pages = "1719--1722",
    year = "1999"
}

@article{Ferrer:2017xwm,
    author = "Ferrer, Francesc and da Rosa, Augusto Medeiros and Will, Clifford M.",
    title = "{Dark matter spikes in the vicinity of Kerr black holes}",
    eprint = "1707.06302",
    archivePrefix = "arXiv",
    primaryClass = "astro-ph.CO",
    doi = "10.1103/PhysRevD.96.083014",
    journal = "Phys. Rev. D",
    volume = "96",
    number = "8",
    pages = "083014",
    year = "2017"
}

@book{Wald:1984rg,
    author = "Wald, Robert M.",
    title = "{General Relativity}",
    doi = "10.7208/chicago/9780226870373.001.0001",
    publisher = "Chicago Univ. Pr.",
    address = "Chicago, USA",
    year = "1984"
}
\end{document}